\documentclass[12pt]{article}
\usepackage[utf8]{inputenc}

\usepackage{svgcolor}
\def\showauthornotes{0}
\ifnum\showauthornotes=1

\newcommand{\jnote}[1]{\textcolor{blue}{ {\textbf{(Jeff: #1)}}}}
\newcommand{\snote}[1]{\textcolor{cyan}{ {\textbf{(SiuOn: #1)}}}}
\newcommand{\tnote}[1]{\textcolor{red}{ {\textbf{(Tommaso: #1)}}}}
\else

\newcommand{\jnote}[1]{}
\newcommand{\snote}[1]{}
\newcommand{\tnote}[1]{}
\fi

\def \N {\mathbb{N}}
\def \R {\mathbb{R}}

\def \Z {\mathbb{Z}}

\def \E {\mathop{\mathbb{E}}}

\def \eps {\epsilon}
\def \al {\alpha}
\renewcommand{\Pr}{\mathop{\bf Pr\/}}

\def\ceil#1{\lceil #1 \rceil}
\def\floor#1{\lfloor #1 \rfloor}

\newcommand{\pE}{\mathop{\widetilde{\E}}}

\usepackage{amsmath,amssymb,amsthm,amsfonts,latexsym,bm,xspace,graphicx,float,mathtools,mathdots,physics}
\usepackage{braket,caption,subcaption,ellipsis,xcolor,textcomp,hhline,pifont,combelow,booktabs}
\usepackage[colorlinks=true, allcolors=blue]{hyperref}
\usepackage{color}
\usepackage{times}
\usepackage{fullpage}
\usepackage{tikz-cd}
\usepackage[shortlabels]{enumitem}
\usepackage{thm-restate}
\usepackage[capitalise]{cleveref}
\usepackage{mdframed}
\usepackage{csquotes}
\usepackage{cancel}

\newtheorem{theorem}{Theorem}[section]
\newtheorem{lemma}[theorem]{Lemma}
\newtheorem{claim}[theorem]{Claim}
\newtheorem{proposition}[theorem]{Proposition}
\newtheorem{fact}[theorem]{Fact}

\newtheorem{conjecture}[theorem]{Conjecture}
\newtheorem{question}[theorem]{Question}

\newtheorem{definition}[theorem]{Definition}
\newtheorem{remark}[theorem]{Remark}

\newtheorem{observation}[theorem]{Observation}

\newtheorem{examples}[theorem]{Example}

\newcommand{\poly}{\operatorname{poly}}
\newcommand{\polylog}{\operatorname{polylog}}
\newcommand{\size}{\textnormal{size}}

\newcommand{\supp}{\textnormal{supp}}

\DeclareMathOperator{\Span}{span}

\DeclareMathAlphabet{\mathsfit}{\encodingdefault}{\sfdefault}{m}{sl}
\SetMathAlphabet{\mathsfit}{bold}{\encodingdefault}{\sfdefault}{bx}{n}

\newcommand{\Paren}[1]{\left(#1\right)}
\newcommand\Brac[1]{\left[#1\right]}
\newcommand\Abs[1]{\left\lvert#1\right\rvert}

\newcommand{\Norm}[1]{\left\lVert#1\right\rVert}
\newcommand*{\Normtv}[1]{\Norm{#1}_{\mathrm{TV}}}
\newcommand{\inner}[1]{\langle#1\rangle}
\usepackage[only,llbracket,rrbracket]{stmaryrd}
\newcommand{\Iv}[1]{\left\llbracket#1\right\rrbracket}
\newcommand{\bracbb}[1]{\left\llbracket#1\right\rrbracket}

\DeclareMathOperator{\Val}{Val}
\DeclareMathOperator{\opt}{opt}
\newcommand{\Axioms}{\textnormal{Axioms}}

\newcommand{\cC}{\mathcal C}
\newcommand{\cD}{\mathcal D}
\newcommand{\cE}{\mathcal E}
\newcommand{\cF}{\mathcal F}

\newcommand{\cH}{\mathcal H}
\newcommand{\cI}{\mathcal I}

\newcommand{\cK}{\mathcal K}

\newcommand{\cP}{\mathcal P}

\newcommand{\cR}{\mathcal R}

\newcommand{\cT}{\mathcal T}
\newcommand{\cU}{\mathcal U}

\newcommand{\1}{\text{\usefont{U}{bbold}{m}{n}1}}

\mathchardef\mhyphen="2D
\newcommand{\val}{\text{val}}
\newcommand{\calP}{\mathcal{P}}
\newcommand{\calH}{\mathcal{H}}

\newcommand{\calS}{\mathcal{S}}

\newcommand\numberthis{\addtocounter{equation}{1}\tag{\theequation}}

\newcommand{\mul}{\text{mul}}

\renewcommand{\int }{\text{Int}}

\newcommand{\Erdos}{Erd\H{o}s\xspace}
\newcommand{\Renyi}{R\'enyi\xspace}

\usepackage[notes=true,later=false,camera=false]{dtrt}
\usepackage{algorithm}
\usepackage{algorithmic}

\allowdisplaybreaks

\begin{document}
\title{Strongly Refuting Random CSP without Literals}
\author{
Siu On Chan
\thanks{sochan@gmail.com}
\and
Tommaso d'Orsi
\thanks{{Bocconi University}. {tommaso.dorsi@unibocconi.it}}
\and
Jeff Xu
\thanks{{Toyota Technological Institute at Chicago}.{jeffxusichao@ttic.edu}}
}

\maketitle
\thispagestyle{empty}
\abstract{
  Under what condition is a random constraint satisfaction problem hard to refute by the sum-of-squares (SoS) algorithm?
  A sufficient condition is ``$t$-wise uniformity'', that is, each constraint has a $t$-wise uniform distribution of
  satisfying assignments, as shown by the lower bounds of Kothari, Mori, O'Donnell, and Witmer (STOC 2017).
  This condition is also necessary for random CSPs given by a predicate and uniformly random literals, due to the
  constant-degree SoS refutation of Allen, O'Donnell, and Witmer (FOCS 2015).
  For higher degree, Raghavendra, Rao, Schramm (STOC 2017) gave a refutation for boolean random CSPs with uniformly
  random literals, matching the lower bounds optimally, in terms of the three-way tradeoff between constraint density,
  SoS degree and strength of refutation.

  Two long-standing open problems are to find a more general sufficient condition for SoS lower bounds, and to refute
  similar random CSPs not involving literals.
  We show that for a general random $k$-CSP, the necessary and sufficient hardness condition is not $t$-wise uniformity,
  but $t$-wise independence.
  We generalize the optimal three-way tradeoff to any random $k$-CSP, without assuming boolean domain or uniformly
  random literals.

  Our analysis involves new Kikuchi matrices for odd-order and for asymmetric tensors.
  Our analysis also uses the global correlation rounding technique of Barak, Raghavendra and Steurer (FOCS 2011).
  To avoid the running time penalty of this technique, we also give a spectral refutation algorithm.
}

\clearpage\newpage
\thispagestyle{empty}
\setcounter{page}{0}

\tableofcontents
\clearpage\newpage
\thispagestyle{empty}
\setcounter{page}{1}

\clearpage
\newpage

\section{Introduction}

Because of their centrality to complexity theory \cite{ben2002gap}, cryptography \cite{applebaum2010public},
statistical physics \cite{crisanti20023sat} and learning theory \cite{daniely2014average},  constraint
satisfaction problems play a major role in computer science.
A $k$-CSP is a pair $(D,\cR)$ where  $D$ is a finite domain and  $\cR$ is a collection of $k$-ary relations over $D.$
An instance $\cI$ is given by  a set $V$ of $n$ variables and by a set $\cC\coloneqq\set{(\sigma_i,R_i)}_{i\in [m]}$
of clauses, where $R\in\cR$ and $\sigma_i:[k]\to V$ is an injection.
Given an instance $\cI,$ the algorithmic task is to efficiently find a satisfying assignment, or conversely,
an efficiently veriable certificate that every assignment fails to satisfy all the clauses.
The largest fraction of constraints simultaneously satisfiable is defined to be
\[
  \opt(\cI)=\max_{x\in D^n} \Pr_{(\sigma,R)\sim \cC} [x\circ \sigma \in R]\,.
\]
A certificate is then called a strong refutation if it proves $\opt(\cI)\leq 1-\Omega_k(1),$ and a weak refutation
if it proves $\opt(\cI) < 1$.

As one can always encode into an algorithm the optimal solution to a particular instance, most of the
literature has focused on settings in which the instance is randomly generated (see \cite{AOW15,
KMOW17, d2023ihara} and references therein).
In 2002, Feige \cite{Fei02} conjectured that it is computationally intractable to find strong
refutations of random 3-SAT formulas when $m\geq C\cdot n$ for a sufficiently large $C,$ and tied this
statement to the inapproximability of certain optimization problems.
Similarly, different CSPs have been connected to numerous other problems \cite{alekhnovich2003more,
  goerdt2004approximation, briest2008uniform, alon2011inapproximability, chuzhoy2015approximation,
bhaskara2012quadratic, alon2012optimizing,o2014hardness, razenshteyn2016weighted}  including OWF, PRGs and
widely studied cryptographic schemes (see the excellent survey by Applebaum \cite{applebaum2016cryptographic}).

Much of the work towards understanding Feige's conjecture has focused on  refutations via (approximately)
automatizable proof systems such as spectral techniques \cite{GK01, FG01}
and semialgebraic proof systems \cite{buresh2003rank, alekhnovich2005towards, grigoriev2001complexity,
Schoenebeck08, Tul09, Chan16, tulsiani2013lsplus, BCK15, mari2016lower, KMOW17, chan2024how}.
In particular, a long line of research has studied the complexity of refutations for $k$-CSP via the
sum-of-squares hierarchy, leading to the strongest known positive results
\cite{AOW15,RRS17, WAM19, AhnStrongRefutationKXOR, d2023ihara}.
Most of these works on upper bounds focus on a restricted class of $k$-CSP that we call ``predicate CSP with literals''.
An instance of a random predicate CSP with literals comes with predicate $P: \set{0,1}^k \to \set{0,1}$, and every
constraint applies $P$ to $k$ uniformly random literals (note that it is also possible to generalize this to  larger alphabets).

Previous works identified the condition governing the sum-of-squares complexity of random predicate CSPs: pairwise
uniformity.
Given $t \in \N$, a CSP is $t$-wise uniform if every constraint supports a $t$-wise uniform distribution of satisfying
assignments.
Kothari, Mori, O'Donnell, Witmer \cite{KMOW17} ruled out $\Omega(n)$-degree refutation for $t$-wise uniform CSPs,
improving on Barak, Chan, Kothari \cite{BCK15} who ruled out $\Omega(n)$-degree strong refutation for pairwise
uniform CSPs.
Further, \cite{KMOW17} also ruled out $O(n/\Delta^{2/(t-1)}\log\Delta)$-degree refutations for $t$-wise uniform
CSPs with $\Delta n$ random constraints, for every $t \geq 2$ and every constraint density $\Delta$ ($\Delta$ may
depend on $n$). In particular this implies that pairwise uniform CSPs require $\Omega(n)$-degree refutations.
For the upper bound, Allen, O'Donnell, and Witmer \cite{AOW15} strongly refute a non-$t$-wise-uniform
predicate CSP (with uniformly random literals) using $O(1)$-degree SoS when there are $\tilde O(n^{t/2})$ constraints,
so non-pairwise uniform predicate CSPs have $O(1)$-degree upper bounds.
These upper bounds were extended by Raghavendra, Rao, Schramm \cite{RRS17}, Wein, El Alaoui, Moore
\cite{WAM19} and Ahn \cite{AhnStrongRefutationKXOR} to higher-degree SoS refutations, matching the degree lower bounds
of \cite{KMOW17} up to polylog factors at every constraint density (that is at least polylogarithmic).
In a different direction, d'Orsi and Trevisan \cite{d2023ihara} removed logarithmic factors in the constraint density
requirement of \cite{AOW15} when $t = k$.

However, predicate CSPs with literals fail to capture many interesting problems in complexity theory, combinatorics and
algorithms, such as monotone $1$-in-$k$ SAT, graph coloring, approximate graph homomorphism
\cite{brakensiek2021promise}.
Many major CSP results apply equally well when literals are absent
\cite{Raghavendra:2008:OAI:1374376.1374414,ABS15,bulatov2017dichotomy,zhuk2020proof}.
It is therefore of interest to consider refutation of random CSP without literals.
Unfortunately, the techniques in \cite{AOW15} and subsequent works require the random CSP to have uniformly random
literals.
Going beyond this barrier has been open for over a decade:
A direction for future work from \cite[Section~7]{AOW15} in 2015 was ``to show analogous efficient refutation results
for models of random CSP(P) in which literals are not used.''
In the other direction, Austrin and Risse \cite[Section~6.3]{austrin2022perfect} mentioned the open problem ``to prove
SoS lower bounds for random CSPs that do not support pairwise uniform distributions.''
Indeed, the SoS lower bound for pairwise uniform CSPs \cite{BCK15,KMOW17} have resisted generalization for nearly a
decade.
Even if one is only interested in predicate CSPs with literals, previous SoS refutations
\cite{AOW15,RRS17,WAM19,AhnStrongRefutationKXOR,d2023ihara}
assumed uniformly random literals, so no SoS refutation was known for biased literals, let alone arbitrary
distributions of literals.

In this work, we thoroughly answer these decade-old questions.
We identify the precise condition governing the SoS complexity of random $k$-CSP without literals.
The condition turns out not to be $t$-wise uniformity, but (what we call) $t$-wise independence.
Given a domain $D$, natural numbers $t \leq k$, and a distribution $\nu$ over $D$, we say that a distribution $\mu$ over
$D^k$ is $t$-wise $\nu$-independent if the $t$-wise marginals of $\mu$ are all $\nu^t$.
A $k$-CSP is $t$-wise independent if there is a distribution $\nu$ over $D$ such that every constraint supports a
$t$-wise $\nu$-independent distribution of satisfying assignments.

\subsection{Our Lower Bound Results} We show SoS lower bounds for $t$-wise independent CSPs \tnote{(We defer formal definitions to \cref{sec:preliminaries})}:

\begin{theorem} \label{thm:independent-sdp}
  Let $t \geq 2$.
  If a $k$-CSP is $t$-wise independent, then except with probability $o_n(1)$, a random binomial instance of
  the CSP with $n$ variables and $\Delta n$ expected constraints has an SoS solution of value $1$ and of degree
  $\Omega_k(n/(\Delta^{2/(t-1)} \log \Delta))$.
\end{theorem}

Unless a CSP is trivially satisfiable by a single value (in which case every instance is satisfiable), a random instance
is far from satisfiable with high probability (whp) \cite[Lemma~B.2]{chan2024how}.
And yet $\Omega(n)$-degree SoS cannot refute such random CSPs by \cref{thm:independent-sdp}.

We also show SoS lower bounds for CSP that are not necessarily $t$-wise independent.
Given a distribution $\rho$ over $\cR$, define the $t$-wise independent value to be the expected fraction of constraints
satisfied by the best tuple of $t$-wise independent distributions, when constraints are chosen from $\rho$:
\[ \opt_t(\rho) \coloneqq \max \Set{\E_{R \sim \rho} \Pr_{x \sim \mu^R} [x \in R] | \Paren{\mu^R}_{R\in\cR}
\text{ is $t$-wise $\nu$-independent for some $\nu$}}\,. \]
This definition generalizes ``distance from supporting a $t$-wise uniform distribution'' \cite{KMOW17}.

\begin{theorem} \label{thm:t-wise-value-sdp}
  Let $t \geq 2$, and $\rho$ be a distribution over $\cR$.
  Given any $k$-CSP $(D, \cR)$, except with probability $o_n(1)$, a binomial $\rho$-random instance of the CSP with
  $n$ variables and $\Delta n$ expected constraints has an SoS solution of degree
  $\Omega_k(n/(\Delta^{2/(t-1)} \log \Delta))$ and of value $\opt_t (\rho) + o_n(1)$.
\end{theorem}

Our \cref{thm:t-wise-value-sdp,thm:independent-sdp} generalize \cite[Theoreoms~1.1 and 1.2]{KMOW17} from
$t$-wise uniformity to $t$-wise independence, matching the SoS degree lower bound of the latter.
In particular, when the number of constraints $m = O(n^{t/2})$, the SoS degree lower bound is
$\tilde\Omega(n^{1/(t-1)})$ (and is in fact $\Omega(n)$ when $t = 2$).
In \cref{sec:lower-bounds}, we also show variants of \cref{thm:independent-sdp} for combined relaxations/hierarchies
such as SDP+AIP and C(BLP+AIP) studied in recent works
\cite{chan2024how,chan2025how,ciardo2025semidefinite,zhuk2025singleton}.
\subsection{Our Upper Bound Results}

\paragraph{Pinning Down the Threshold with SoS}
For the upper bound, when the number of constraints is $m = \Omega(n^{t/2})$, we show there is a constant degree strong refutation with
high probability, upper bounding the SoS value to be arbitrarily close to $\opt_t (\rho)$:
\begin{theorem}\label{thm:main-upper-bound}
  For any $k$-CSP $(D, \cR)$, any $2 \leq t \leq k$, any $\eps > 0$, there are $K = K(k, \abs{D}, \eps)$ and
  $C = C(k,\abs{D},t,\eps)$ such that whenever
  \[
    m \geq \begin{cases}
      C n^{t/2} & \text{if $t$ is even} \\
      C n^{t/2} \sqrt{\log n} & \text{if $t$ is odd}
    \end{cases},
  \]
  except with probability $o_n (1)$ over a binomial $\rho$-random instance $\cI$ having $n$ variables and $m$
  expected constraints, every pseudo-distribution $\mu$ of degree $K$ satisfies
  \[ \pE_\mu \Val_\cI(x) \leq \opt_t(\rho) + \eps. \]
\end{theorem}
Here $\pE_\mu \Brac{\Val_\cI(x)}$ denotes the value of the fractional assignment given by $\mu$
(see \cref{sec:preliminaries} for a precise definition).
Theorem ~\ref{thm:main-upper-bound} generalizes \cite[Theorem~2.4]{AOW15} from $t$-wise uniformity to $t$-wise
independence and, also, remove their assumption about predicate CSPs with uniformly random literals.
When $t$ is even, our requirement on the number of constraints $m = \Omega(n^{t/2})$ also avoids the polylogarithmic
factor in \cite[Theorem~2.4]{AOW15}.

Taken together, our theorems prove the following dichotomy:
For any $2 \leq t \leq k$ such that $t$ is even, a random $k$-CSP has $O(1)$-degree SoS refutation whp whenever
$m = O(n^{t/2})$ if and only if it is $t$-wise independent.
Further, the value of the optimal $O(1)$-degree SoS solution is approximately $\opt_t(\rho)$ whp.

We also generalize the upper bound in \cref{thm:main-upper-bound} to higher degree SoS:
\jnote{double check transition below}
\snote{I changed the statement of \cref{thm:main-upper-bound-higher} so that $t-1 \leq \ell$.
We need $\ell \geq t-1$ for odd $t$, while we need $\ell \geq t/2$ for even $t$, so I just choose a bound that works for
both odd and even $t$. We can also state the lower bound for $\ell$ separately for odd and even $t$, but I think a tight
lower bound for $\ell$ is not important.}
\begin{theorem} \label{thm:main-upper-bound-higher}
  For any $k$-CSP $(D, \cR)$, any  $2 \leq t \leq k$, any $t-1 \leq \ell \leq \tilde{O}(n)$, any $\eps > 0$, there are
  $K = K(k, \abs{D}, \eps)$ and $C = C(k,\abs{D},t,\eps)$, such that whenever
  \[
    m \geq \begin{cases}
      \displaystyle C n^{t/2} \frac{\ell}{\ell^{t/2}} \log n & \text{if $t$ is even} \\
      \displaystyle C n^{t/2} \frac{\ell}{\ell^{t/2}} \sqrt{\log n} & \text{if $t$ is odd}
    \end{cases},
  \]
  except with probability $o_n (1)$ over a binomial $\rho$-random instance $\cI$ having $n$ variables and $m$
  expected constraints, every pseudo-distribution $\mu$ of degree $K+2(\ell+1)$ satisfies
  \[ \pE_\mu \Val_\cI(x) \leq \opt_t(\rho) + \eps. \]
\end{theorem}
Our \cref{thm:main-upper-bound-higher} removes the boolean domain and the uniform literal assumptions in
\cite{RRS17,WAM19,AhnStrongRefutationKXOR}.
\cref{thm:main-upper-bound-higher,thm:t-wise-value-sdp} together achieve the optimal three-way tradeoff
between constraint density, SoS degree, and strength of refutation up to a polylogarithmic factor in the degree and an
arbitrarily small additive term in the strength of refutation.
This generalizes the three-way tradeoff previously for boolean CSP with uniformly random literals \cite{KMOW17,RRS17};
see \cite[Section~1.5]{KMOW17}.
\tnote{I would move the next phrase directly into the paragraph below.}
While this result matches the optimal clause-density threshold, it introduces an additive $\eps$ slack in the certified value, along with a dependence of $O_k(1/\eps^2)$ on the SoS degree. In contrast, in the classical setting of uniformly random signings, spectral algorithms achieve tight certification without these overheads.

\paragraph{Improved Refutation with Spectral Algorithms}

 Motivated by this observation, we next present a spectral refutation algorithm that removes both the runtime dependence in the additive $\varepsilon$ slack and the need for high-degree relaxations and rounding, while retaining the optimal threshold $m \sim n^{t/2}$, matching both the classical uniform-literal setting and the best-known Sum-of-Squares lower bounds. 
 
\begin{restatable}{theorem}{LargeAlphabetRefutation} 
\label{thm: alphabet-ref}
\tnote{I don't see why we need "Application to Predicates with Large Alphabets [].." given that we don't have a header in the previous theorems. Also no probability appears in the theorem statement?}\jnote{yup deleted- and yup looks good}
  \snote{I avoid the terminology ``random literals according to distribution $\rho$''. See if the next sentence is good
  for you. $\rho$-random instance is defined in Preliminaries.}
	For a $\rho$-random instance $\Psi$, for any $t> 1$,  there is an algorithm that runs in time ${O}( \binom{k}{\leq t}  \cdot  n^{O(\ell)} ) + O( (\frac{3}{\eps})^{|D|} ) = O_{k} (n^{O(\ell)}) + O( (\frac{3}{\eps})^{|D|} )  $ that certifies, \[ 
	\max_{x\in D^n} \Val_{\Psi}( x) \leq  \opt_t(\rho)   + \eps +o_n(1) 
	\]
		  provided $m$ satisfies the following density constraint
		 \[
    m \geq \begin{cases}
      \displaystyle C \cdot \frac{1}{\eps^2}\cdot  n^{t/2} \frac{\ell}{\ell^{t/2}} \log n & \text{if $t$ is even} \\
      \displaystyle C \frac{1}{\eps^2}\cdot n^{t/2} \frac{\ell}{\ell^{t/2}} \sqrt{\log n} & \text{if $t$ is odd}
    \end{cases},
  \]
  for any $\ell\in \N$ s.t. $t-1 \leq \ell \leq \tilde{O}(n)$.
		  \end{restatable}
\snote{I changed $t/2 \leq \ell$ to $t-1 \leq \ell$. See my comment before \cref{thm:main-upper-bound-higher}}
\snote{The running time of the spectral algorithm is the sum of two terms: $O_k(n^{O(\ell)})$, and
$\min(O_{|D|,k}(1/\eps),n)^{O(|D|)}$. The second term coming from bias/marginal enumeration is awkward to state, and may
only be stated in full in a remark.
The main theorem statement may replace it with the weaker bound $O_{|D|,k}(1/\eps)^{O(|D|)}$ or $n^{O(|D|)}$.
The weaker bound $n^{O(|D|)}$ is nicer as it can be combined with the first term as $O_k(n^{O(\ell)}+n^{O(|D|)})$ or
even $O_k(n^{O(\ell+|D|)})$.}
\jnote{looks good!} 

\begin{remark}
	When specialized to predicates on the boolean domain, our algorithm runs in time independent of $\eps$ (See \cref{thm:spectral-refutation}  for formal statements).
\end{remark}
 \begin{remark}
 	Our density constraint on $m$ can be improved to contain a factor of $\sqrt{\log n}$ as opposed to $\log n$. (See \cref{thm:smooth-spectral-refutation} for formal statements).
 \end{remark}
A closer inspection of our spectral refutation algorithm reveals that it is already captured by a level-$O(\ell)$ Sum-of-Squares relaxation—without any $1/\varepsilon$ dependence—once augmented with marginal (bias) constraints. Whether this augmented program is in fact subsumed by the canonical SoS formulation for Max-$k$-CSP, without explicitly imposing such constraints, is less clear.

Notably, \cite{KOS19} shows that for CSPs with random literals, this augmentation does not strengthen the power of the standard SoS SDP. It remains an intriguing open question whether, in our setting, the bias constraints are automatically enforced by the standard SoS relaxation at a comparable degree.
\jnote{added discussion}




\subsection{Related Work}
\jnote{other idea for the subtitle?}
Historically, the terminology ``$t$-wise independent'' had three different meanings in the literature.
A distribution $\mu$ over $D^k$ has been called ``$t$-wise independent'' if its marginal $\pi_T (\mu)$ on a subset $T
\subseteq [k]$ of size $t$ is:
\begin{enumerate}
  \item (Definition 1; most general) $\prod_{i \in T} \nu_i$ for some distributions $\set{\nu_i}_{i \in [k]}$
  \item (Definition 2) $\nu^T$ for some distribution $\nu$
  \item (Definition 3; most restrictive) uniform distribution over $D^t$
\end{enumerate}
The terminology ``$t$-wise uniform'' was coined by \cite[definition~3.10]{AOW15} in 2015 to refer to the most
restrictive definition.
Numerous CSP lower bounds (e.g.~\cite{AM09,benabbas12sdp,BCK15}) that appeared before \cite{AOW15} called the same
concept by ``balanced $t$-wise independent'' or even just ``$t$-wise independent''.
By contrast, our ``$t$-wise independence'' refers to Definition 2.
The closest to our definition of $t$-wise independence appeared in Austrin and H\r{a}stad \cite{AustrinHastad2011} (in
a related but slightly different context of worst-case approximation resistance under the Unique-Games Conjecture),
whose definition agrees with ours when the CSP has only one type of constraint ($\abs{\cR} = 1$).
For CSPs with multiple types of constraints ($\abs{\cR} > 1$), our definition of ``$t$-wise independent CSP'' is new.
Part of our contribution is identifying the precise condition that determines the complexity of random $k$-CSP
against SoS.
Before our work, it was far from obvious that when literals are not considered, the correct generalization of $t$-wise
uniform CSP is the new concept of $t$-wise independent CSP introduced here.

We demonstrate in \cref{thm:indep-not-uniform} a CSP that is pairwise independent but not pairwise uniform (even when
restricted to any subdomain), so our lower bound (\cref{thm:independent-sdp}) applies to a strictly larger class of CSPs
than \cite{KMOW17}.
We also show in \cref{thm:predicate-independent-uniform} that for the special case of predicate CSP with literals,
$t$-wise uniformity coincides with $t$-wise independence.\tnote{Should we write the above phrase earlier in the section? E.g. on page 2 when we say "the condition turns out not ot be t-wise uniformity, but (what we call) [..]", or in some other place we are sure the readers will see it.}
The proof of \cref{thm:predicate-independent-uniform} also suggests that the $t$-wise uniformity condition in previous
upper bounds \cite{AOW15} ultimately come from the symmetry of uniformly random literals.
Our upper bounds instead come from a different symmetry: That every $k$-tuple of variables is equally likely to take up
a random constraint.

More broadly, our results fit in the study of hierarchy algorithms for random CSPs.
It is of interest to identify general conditions for random CSPs to be intractable for restricted classes of algorithms
such as hierarchies.
Sufficient conditions include:
\begin{itemize}
  \item immunity for polynomial calculus \cite{alekhnovich2001lower}
  \item null-constraining for resolution/bounded-width \cite{chan2013dichotomy,chan2025how}
  \item pairwise neutrality for the affine integer programming (AIP) hierarchy \cite{chan2024how}
  \item null-constraining and lax for the AIP hierarchy and polynomial calculus \cite{chan2025how,conneryd2026lower}
  \item pairwise uniformity for LP and SoS hierarchies \cite{benabbas12sdp,BCK15,KMOW17}
\end{itemize}
Finding general sufficient conditions is useful because immunity and pairwise uniformity ``are the source of a
majority of existing CSP lower bounds in PC and SoS'' \cite[section~1.2]{austrin2022perfect}.
Our work completed the picture for the average-case complexity in SoS by generalizing ``pairwise uniformity'' to
``pairwise independence'', which we also show to be necessary.
Necessary and sufficient conditions have rarely been identified; the only other hierarchy whose average-case
intractability condition is known precisely is resolution/bounded-width \cite{chan2013dichotomy,chan2025how}.


Our upper bounds are also particularly interesting to the approximate graph homomorphism conjecture of Brakensiek and
Guruswami \cite{brakensiek2021promise}:
\begin{conjecture}[{\cite[conjecture~1.2]{brakensiek2021promise}}]
  Let $G$ and $H$ be undirected nonbipartite graphs with a homomorphism from $G$ to $H$.
  Then the promise digraph homomorphism problem associated with $G$ and $H$ is NP-hard.
\end{conjecture}
As a special case, this conjecture implies that $3$-vs-$K$ graph coloring is NP-hard for every $K \geq 3$, a question
that is wide open.
For every interesting case of the approximate graph homomorphism, \cite{ciardo2025semidefinite} showed lower
bound against basic SDP+AIP, while \cite{CZ24Periodic,chan2025how} showed lower bound against bounded-width algorithms;
see also an open question \cite[section~8]{conneryd2026lower} regarding polynomial calculus lower bound.
Would similar lower bounds hold for $\omega(1)$-degree SoS and SDP+AIP?
Currently, the only $\omega(1)$-degree SoS lower bounds for any CSP come from random CSP, its derandomization, or
reduction.
For the special case of approximate graph coloring (when $G$ and $H$ are cliques of size $3$ and $K$), average-case SoS
upper bounds are well known, but what about other undirected nonbipartite $G$?
Our \cref{thm:main-upper-bound} implies $O(1)$-degree average-case SoS upper bound for every such $G$.
(Reason: $G$ has no self-loops and cannot support any pairwise independent distribution.)
This rules out a promising approach towards $\omega(1)$-degree SoS lower bounds for every interesting case of
approximate graph homomorphism.

Our discussion so far focused on average-case complexity, but our work also has potential implications to worst-case
complexity, especially in the context of combined relaxations and hierarchies.
Bhangale, Khot, Minzer \cite{bhangale2025approximability} recently proposed a candidate approximation algorithm
for satisfiable CSPs.
Their algorithm combines SDP and Gaussian elimination, and is closely related to SDP+AIP (also known as SDA) relaxation
\cite{ciardo2025semidefinite,chan2024how}.
If the approximation algorithm of \cite{bhangale2025approximability} is indeed the optimal approximation algorithm for
satisfiable CSPs, it would be a new and simpler decision algorithm for tractable CSP, alternative to Bulatov's
\cite{bulatov2017dichotomy} and Zhuk's \cite{zhuk2020proof}.
Similar candidate algorithms for tractable CSP have been proposed
\cite{brakensiek2020power,ciardo2022clap,oconghaile2022cohomology,dalmau2024local}, prompting recent works to
investigate their power \cite{chan2024how,chan2025how,lichter2025limitations,zhuk2025singleton}.
Notably, Lichter and Pago \cite{lichter2025limitations} and Zhuk \cite{zhuk2025singleton} ruled out many candidate
algorithms for tractable CSP.
The only existing lower bounds for SDP+AIP (a close cousin of the approximation algorithm of
\cite{bhangale2025approximability}) are based on \cite{ciardo2025semidefinite,chan2024how}, while the only existing
C(BLP+AIP) hierarchy lower bound is based on \cite{chan2025how}.
Our new lower bounds for these combined hierarchies (\cref{thm:cblpaip,thm:sdpaip}) extend the applicability of
\cite{chan2024how,chan2025how} to more CSPs, boosting the hope of ruling out some of the remaining candidate algorithms
for the CSP Dichotomy Theorem and the Approximation Dichotomy Conjecture
\cite[Conjecture~1.1]{bhangale2025approximability}; see \cite[Question~1.5]{chan2024how}.

\paragraph{Organization}\jnote{i touched here}
The remainder of the paper is organized as follows. In \cref{sec:overview}, we present the main ideas underlying both the upper and lower bounds, and in \cref{sec:preliminaries} we introduce the necessary preliminaries.

We then establish our SoS upper bound in \cref{sec:refute-via-independence}, followed by an improved spectral algorithm in \cref{sec:sos-refutations}. In \cref{sec:norm-bound}, we prove the spectral norm bounds required for these results. Finally, we present our lower bound in \cref{sec:lower-bounds}. The appendix contains deferred proofs and other notions that flesh out the exposition.

\section{Techniques} \label{sec:overview}

\subsection{SoS Refutation}\tnote{Our references use lower case letters (e.g. theorem 1.1) I'd argue we should be consistent with that also when referencing other papers. Alternatively we can also capitalize references in our template. }

\cite{AOW15} refutes random predicate CSPs with uniformly random literals using the following three main ingredients:
\begin{enumerate}
  \item Certify the ``induced distribution'' of the instance is close to $t$-wise uniform by bounding its Fourier
    coefficients \cite[Lemma~3.16, Lemma~B.8]{AOW15}
  \item Bound the Fourier coefficients via tensor concentration (in injective norm) \cite[Lemma~4.3, Lemma~B.5]{AOW15}
  \item Show tensor concentration using random matrix theory for even $t$ \cite[Proposition~A.4]{AOW15} or trace method
    for odd $t$ \cite[Lemma~A.5]{AOW15}
\end{enumerate}
All these ingredients (except random matrix theory) crucially depend on uniformly random literals.
Subsequent works \cite{RRS17,WAM19,AhnStrongRefutationKXOR,d2023ihara} modified ingredient 3 without changing the other
two.

To obtain SoS refutations for CSPs without literals, we change all these ingredients substantially.
First, we introduce a new notion of \emph{induced distribution} that generalizes \cite[Definition~3.2]{AOW15} in
the absence of literals.
Given a collection $\cC$ of $k$-ary constraints from an instance and an assignment $x \in D^V$, our induced
distribution, denoted $\cD_{\cC,x}$, is the probability distribution on $D^k$ defined by
\[ \cD_{\cC,x}(a) \coloneqq \Pr_{(\sigma,R) \in \cC} [x_\sigma = a] \qquad \forall a \in D^k\,. \]
In the above definition, the $k$-ary relation $R$ of the constraint is irrelevant;
only the (ordered) tuple $\sigma$ of variables in the constraint matters.\tnote{Can you articulate on how it differs from AOW15? Is the difference only in the fact that we do not have literals? The predicate does not matter in Def 3.2 in AOW15 if I am not mistaken. Possible solution: "We introduce a new notion of \emph{induced distribution} which generalizes \cite[Definition~3.2]{AOW15} to CSPs without literals."}
Our refutation is then based on showing the induced distribution of an instance must be close to $t$-wise independent.
While distance to $t$-wise uniformity can be bounded by Fourier coefficients, distance to $t$-wise independence
(when $\nu$ is unknown) cannot.
\tnote{Can you add something here to cover the fact that we cannot just bound the Fourier coefficients. This allows us to underline how we depart from ingredient 2.}

For simplicity of the discussion, consider an instance $\cI$ of a CSP with only one type of constraint, namely $R'$.
Also consider the case of $t = 2$.
We cannot optimize over integral assignments, but we can do so over pseudo-distributions.
Thus our goal is to show that the induced distribution $\cD_{\cC,\mu} \coloneqq \pE_\mu \cD_{\cC,x}$ coming from a
pseudo-distribution $\mu$ is close to pairwise independent.
After all, if $\cD_{\cC,\mu}$ is pairwise independent, then\footnote{We use the Iverson bracket notation
$\bracbb{\cdot}$ to denote an indicator function.}
\[ \pE_\mu \Val_\cI (x) = \pE_\mu \E_{(\sigma,R) \in \cC} \bracbb{x_\sigma \in R}
= \E_{x \sim \cD_{\cC,\mu}} \bracbb{x \in R'}. \]
That is, the SoS value is exactly the probability that $R'$ is satisfied by a random $x$ from $\cD_{\cC,\mu}$, which
is pairwise independent by assumption.
Hence the SoS value is bounded away from $1$, if $R'$ is far from supporting any pairwise independent distribution.

In general, however $\cD_{\cC,\mu}$ is far from pairwise independence.
To work around this issue we argue that, given a pseudo-distribution $\mu$ of sufficiently large degree as above, for random instances $\cI$ with sufficiently many clauses, we can always construct a pseudo-distribution $\mu',$ of lower degree but with nearly the same objective value, such that  $\cD_{\cC,\mu'}$ is close to pairwise independent.
This motivates us to consider the global correlation rounding framework of Barak, Raghavendra, Steurer \cite{BRS11},
which implies that on an expander graph $G = (V, E)$, in expectation over suitably chosen random vertex subset $T$ and
assignment $\alpha$ to $T$, the conditional pesudo-distribution $\mu_{|\alpha}$ given $\alpha$ is nearly a product
distribution across a typical edge:
\[ \E_{(u,v) \in E} \mu_{u v|\alpha} \approx \E_{u \in V} \mu_{u|\alpha} \cdot \E_{v \in V} \mu_{v|\alpha} \,. \]
In other words, we recover independence of pairwise marginals of the induced distribution, in expectation over
conditioning.
This idea can also be generalized to larger $t$, by repeatedly splitting the $t$-wise marginals into halves.


Once we show that $t$-wise marginals of an induced distribution are all close to being $\nu$-independent,
the next step is to show that the induced distribution itself is close to some $t$-wise $\eta$-independent distribution.
Similar statements under different formulations have appeared before in the literature:
For example, \cite[Theorem~2.1]{alon03almost} by Alon, Goldreich and Mansour when $\nu$ is the uniform Boolean
distribution, and \cite[Theorem~1.2]{rubinfeld13robust} by Rubinfeld and Xie for arbitrary finite distribution whose
minimum nonzero probability is bounded away from zero.
We cannot apply previous results directly, because we want to consider an arbitrary $\nu$ over any finite domain,
without any assumption about its minimal nonzero probability.
For this reason, we need to first remove tiny probability masses from $\nu$.
Doing so means that our $\eta$ may be different from $\nu$, unlike in previous works, but fortunately this weaker
conclusion suffices for our application.
We then generalize the short proof of O'Donnell and Zhao \cite[Theorem~1.1]{odonnell18closness} from the uniform Boolean
distribution to an arbitrary $\nu$.

We remark that, before our work, global correlation rounding has not been used to refute random CSP.
We hope our work will inspire more refutation algorithms for random CSP.
Unlike all previous applications of global correlation rounding (including its generalization to $k$-CSP using
high-dimensional expanders \cite{AGJT19}), our refutation does not come with a rounding algorithm
that outputs an assignment nearly matching the SDP value --- because such an assignment does not exist in general.

\subsection{Spectral Refutation: Removing Global Correlation Rounding }
\jnote{ how do you feel about the transition here?}\snote{Looks alright to me.}

\tnote{I like this paragraph, but I think this might not be the optimal place to put it. We should have made our motivation clear in the introduction already. Can you perhaps try to merge this paragraph with the discussion above \cref{thm: alphabet-ref}.}
For readers familiar with refutation algorithms for many classical CSPs\tnote{If you move this to the intro, remove the part "For readers familiar .." Just tell them wgat they should think: "The use of global correlation in our SoS algorithm may appear.."}, \jnote{im leaning towards keeping it here- cuz thers already similar one at main theorem, and we need sth here as transition filles anyway- feel free to modify tho} the use of global correlation in our SoS algorithm may appear counterintuitive. Such heavy machinery is unnecessary for several well-studied CSPs without literals, including Max-Cut and graph coloring. For these predicates,  spectral refutation algorithms are known to suffice. From this perspective, a natural question arises,
\begin{displayquote}
Can we replace global-correlation rounding  with some simple spectral methods?	
\end{displayquote}

Moreover, in the aforementioned problems, spectral methods proceed by analyzing the nontrivial spectrum of a natural matrix representation of the instance while explicitly discarding the trivial eigenvalue, hinting upon an intriguing connection:
\begin{displayquote}
the trivial eigenvalue in classical spectral analyses appears closely related to the failure of the underlying predicate to support an independent distribution.
\end{displayquote}

\tnote{Suggestion to avoid posing another question: "A first obstacle to formalize this connection is the fact that, for general predicates without literals, the choice of the operator may not be immediately clear."}
A first obstacle to formalize this connection is the fact that, for general predicates without literals, the choice of the operator may not be immediately clear. For graph-based CSPs (e.g. Max-Cut) or even $k$-XOR predicates, the choice of matrix is straightforward (e.g. the normalized adjacency matrix or analogous Laplacian). However, for other predicates it is much less obvious — even in the simplest non-trivial case ($t=2$), such as monotone 1-in-3 SAT. Furthermore, the global-correlation approach makes this choice even murkier. The object one ends up analyzing (whether a matrix or an associated polynomial tensor) is far from being canonical, as it is unclear what the resulting polynomial will be
  after roughly $1/\varepsilon^2$ levels of conditioning.
Taking a step back, even with uniformly random literals, identifying the correct polynomial or matrix requires a ``careful'' setup. \cite{AOW15} achieves this by relating refutation to the failure of the predicate to support a
$t$-wise uniform distribution via a carefully constructed linear program, which we briefly recap below.

	\paragraph{Revisiting Refutation with Uniform Literals} \jnote{clear?}
	 For a predicate $P$ that does not support a 
$t$-wise uniform distribution, the \cite{AOW15} approach proceeds in two conceptual steps:	\begin{enumerate}
		\item \textbf{Finding the Dominating Polynomial:}   Viewing a $k$-CSP predicate $P$ as a polynomial in its Fourier expansion, one can formulate a constant-size linear program that produces a degree-$t$ polynomial $Q$ s.t.
	 $Q(x)\geq P(x)$ for any $x\in \{\pm 1\}^k$. Importantly, note that $Q$ is of degree-$t$ as opposed to degree-$k$. We call this polynomial a \emph{dominating} polynomial, analogous to the notion of separating polynomial appearing in the hardness of approximation literature \cite{AustrinHastad2011, AustrinH13}. 
		\item  \textbf{Spectral Certification:}  With the dominating-polynomial $Q$ identified, the certification of $P$ reduces to a certification of the maximum value of $Q$ over the boolean hypercube. This is equivalent to bounding the injective norm of the associated tensor, which can be done via spectral methods or degree-$2t$ Sum-of-Squares.
	\end{enumerate}
	Crucially, this framework relies on the randomness of the literals in the second step to ensure concentration around the constant term of $Q$.

\paragraph{Incorporating a Fixed Bias}	
A natural attempt to lift the above approach to our setting with non-uniform literals is \tnote{to generalize the above approach to our settings of interest (alternatively recap the settings, though that seemed a bit verbose when I tried)} \jnote{changed slightly}to start with a fixed marginal bias $\nu$, and 
recover the separating polynomial from \cite{AOW15} while incorporating the knowledge of the bias into this polynomial. This begets the following LP that searches for the best independent distribution with marginal bias $\nu$.

\begin{mdframed}[frametitle =Primal LP for the Maximal $t$-wise $\nu$-independent Distribution]	

\begin{align*}
	&\max_{\mu}\quad  \val_{P,t}(\nu)\coloneqq  \sum_{x \in \{\pm 1\}^k} P(x) \cdot \mu(x) \nonumber \\
	\text{s.t.} \sum_{x\in \{\pm 1\}^k } \mu(x) \cdot x_S &= \nu^{|S|} \quad \forall S\subseteq [k], |S|\leq t  \qquad \text{($t$-wise Marginal Consistency)}\\
	& \sum_{x\in \{\pm 1\}^k} \mu(x) = 1 \qquad \text{(Normalization)}\\ & \mu(x)\geq 0 \quad \forall x\in \{\pm 1\}^k \qquad \text{(Non-negativity)}\\
\end{align*}
\end{mdframed}

Notice that $\nu=0$ recovers the \cite{AOW15} formulation for uniform distribution in the support. For convenience, we introduce the shorthand  $Q_{\nu}(x) \coloneqq  \sum_{S:|S|\leq t}  c_S \cdot x_S $ from a set of coefficients $\{c_S\}_S$.  Taking the dual of the above LP gives us the desired dominating polynomial.
\begin{mdframed}[frametitle={Dual LP for the Minimizing Polynomial Certificate for $P(x)$ with bias $\nu$}]
\begin{align*}
\numberthis \label{eq:dual-lp-boolean}
	&\min_{\{c_S\}_S} \quad \val_t(Q_\nu) \coloneqq \sum_{S\subseteq [k],\, |S|\le t} c_S \cdot \nu^{|S|} \\
	&\text{s.t.}\quad P(x) \le Q_\nu(x) \qquad \forall x \in \{\pm 1\}^k 
	\quad \text{(Pointwise Dominance)}
\end{align*}
\end{mdframed}

\paragraph{Partial Certification from Bias}
At this point, it is tempting to ask the following question: is the objective value of the LP what we can efficiently certify for the instances of $Q_\nu(\cdot)$ over random hypergraphs (with fixed signings)? If true, by the strong duality of LP, we would indeed have certified a value given by the best $t$-wise independent distribution with bias $\nu$ as desired.
More formally:
\begin{question}
For any given polynomial $Q_\nu(\cdot)$ obtained above on $k$-variables of degree-$t$, let $\Psi_{Q_\nu} $ be an instance on $n$-vertex hypergraph $\calH$ (in other words, $\Psi_{Q_\nu}$ is obtained from $\Psi_P$ by replacing each predicate $P$ by $Q_\nu$), is there  a spectral algorithm that certifies that \[ 
\max_{x\in\{\pm 1\}^n}  \Psi_{Q_\nu}(x) \leq  \val_t(Q_\nu) + o_n(1).
\] 	
provided the instance is dense enough above the anticipated spectral threshold but without uniformly random literals?
\end{question}
Note that the above is true for the case of $\nu=0$ as the objective value is simply given by the constant term $c_\emptyset$ that one would certify assuming the instance additionally comes with uniformly random literals. However, in the setting of fixed signings, a uniform certification over the entire boolean hypercube like the above turns out to be too much to ask. Consider the following example of $2$-XOR:

\begin{examples}
	Given $m = \tilde{\Omega}(n)$ random edges in an $n$-vertex graph viewed as a random (unsigned) $2$-xor instance $\Psi(x) = \frac{1}{m} \cdot  \sum_{e=(i,j) } x_{i} x_{j} $.
\end{examples}

Let $M_\Psi$ be the adjacency matrix of the graph sample. It is not surprising that we can at best certify $\max_{x\in \{\pm 1\}^n} \Psi(x) \leq \frac{1}{m}  \|M_{\Psi}\|_{sp} \cdot \|x\|_2^2 = \frac{1}{m}  \frac{m}{n} \cdot n = 1   $ spectrally since this is a trivial bound that can indeed be achieved by the all-$1$ assignment. Put differently, the above bound is trivial because the spectral norm that usually gives us a square-root saving reduces to a naive  row-sum bound. 

This leads to our key observation -- such failure does not render our approach doomed. While it is
impossible to certify a uniform upper bound over all assignments using a
single low-degree polynomial, we can try to certify sharp bounds for
assignments with a \emph{fixed} marginal bias.
Notice that
  the spectral norm of the matrix is overwhelmingly dominated by the trivial eigenvalue; in particular, its spectrum lies in $[- c \sqrt{\frac{m}{n}}, c\sqrt{\frac{m}{n}}] \cup \{ \frac{m}{n} \}  $ for some constant $c>0$. Thus
    we can rewrite the above bound as 
    \begin{align*}	
 m\cdot \Psi_P( x) &\leq \max_{x: \|x\|^2_2 =n } x^T M_{\Psi} x = \max_{\|x\|^2=n} \sum_{i\in [n]} \lambda_i \langle v_i, x\rangle^2\\ 
 &\leq \max_{\|x\|^2_2 = n }	(1+o_n(1)) \cdot \frac{m}{n} \cdot  \langle \1_n, x\rangle^2 + O\left(\sqrt{\frac{m}{n}}\right) \cdot  (n - \langle \1_n, x\rangle^2) \\
 &\leq (1+o_n(1)) \cdot   \frac{m}{n} \langle \1_n, x\rangle^2 + o(m)
 \end{align*}
where the second-to-last line follows as $\1_n$ is (approximately) the top eigenvector that correspond to the eigenvalue of $(1+o_n(1)) \frac{m}{n}$. Crucially, the certifiable bound is dominated by the coordinate bias of the assignment $x$.

Therefore,  for each fixed bias
$\nu$, even in the presence of the trivial eigenvalue,  spectral methods continue to certify that the value of the instance
restricted to assignments with marginal bias $\nu$ concentrates around
the value predicted by $\val_t(Q_\nu)$. We refer to such bounds as
\emph{partial certificates}, and establish the following lemma as a key ingredient of our spectral refutation.

\begin{restatable}{lemma}{PartialCertification} \label{lem:partial-certification}
(Partial Certification in the $\nu$-Biased Hypercube (Informal of \cref{lem:formal-partial-certification}) ) 	 There is an efficient spectral algorithm that certifies \[ 
\max_{\substack{x\in\{\pm 1\}^n \\ \E_{i\in [n]} [x_i] = \nu } } \Psi_{Q_\nu}(x) \leq \sum_{S\subseteq  [k], |S|\leq t} c_S\cdot \nu^{|S|}  + o_n(1) = \val_{t}(Q_\nu) +  o_n(1) \]
provided the instance is sufficiently dense above the anticipated spectral threshold (see the formal condition in \cref{thm:spectral-refutation})  . \end{restatable}

\paragraph{From Partial to Global Certification: Don't Bother Committing}
To conclude the argument, it remains to show how a partial certification is sufficient to certify the optimum over all assignments. For instance, the refutation algorithm of \cite{AOW15} ultimately reduces to certifying a single degree-$t$ polynomial. Ideally, one would like to select the dual polynomial from \cref{eq:dual-lp-boolean} corresponding to the optimal bias and argue that certifying this single polynomial suffices. However, this is less clear when the optimal bias is unknown.

\tnote{Can you explicit the issue? E.g. "Because we do not know the bias this approach cannot work }\jnote{added}

That said, there is a simple workaround that avoids this issue altogether: we do not need to commit to a single polynomial. Instead, we consider all possible biases. In the Boolean setting, the bias of an assignment is determined by the number of $+1$ variables, and hence there are only $O(n)$ possible biases. We can therefore enumerate over all of them, leading to the following certification procedure.

\begin{enumerate}
	\item For each bias $\nu \in [-1,1]$, solve the dual LP in \cref{eq:dual-lp-boolean} to obtain a polynomial $Q_\nu$.
	\item Apply spectral certification to $Q_\nu$ to certify that every assignment $x$ with bias $\nu$ satisfies
	\[
	P(x) \le \val_t(Q_\nu).
	\]
	\item Output the bound
	\[
	\max_{\nu} \val_t(Q_\nu).
	\]
\end{enumerate}
Since every assignment has some bias $\nu$, this procedure yields a valid global certification. To extend to the large domain, it suffices for us to replace the gird-search component by an epsilon-net of bias. 

%

\paragraph{One More Thing about Spectral Norm Bounds}
It is by now standard that certifying the injective tensor norm of a random polynomial can be reduced to analyzing associated random matrices and, with high probability, bounding their spectral norms~\cite{BarakMoitra2016, AOW15, Bhattiprolu2017SumofSquaresCF, RRS17, WAM19, AhnStrongRefutationKXOR}.
In our setting, these matrices can be viewed as centered adjacency matrices of directed (hyper)graphs, but we are not 
aware of general-purpose results that directly apply to matrices derived from odd-order tensors.
Accordingly, we give a self-contained trace-moment analysis, yielding the following by-product.

To the best of our knowledge, prior works~\cite{GKM, HKM} obtaining comparably tight bounds (including the log factors) via the Kikuchi hierarchy rely on hypergraph decomposition and pruning, which introduces an additional $\frac{1}{\epsilon^2}$ factor in the odd-order case.
While these arguments apply to general undirected hypergraphs (with random literals), we expect that some form of decomposition or conditioning---beyond row pruning---is still needed in the random-hypergraph setting.
In contrast, our approach applies directly to the vanilla Kikuchi matrices for odd-order tensors and achieves the currently optimal tradeoff.
The improvement follows from a modification of a vanilla encoding argument, which may be of independent interest.

Additionally, we also introduce a generalized Kikuchi matrix tailored to non-boolean domains and asymmetric tensors
coming from directed hypergraphs.
This allows us to use the Kikuchi matrix technique to bound the injective norm of an asymmetric tensor, which arises due
to the non-boolean domain. \snote{I changed this sentence.}

\subsection{Lower bounds}
Our lower bounds leverage the recent LP and SoS lower bound framework of Chan, Ng, and Peng \cite{chan2024how}.
Previous LP and SoS lower bounds \cite{benabbas12sdp,BCK15,KMOW17} for $t$-wise uniform CSPs considered the
same ``canonical'' or ``planted'' distribution of satisfying assignments to a small subinstance: Pick an assignment from
the given $t$-wise uniform distribution on each constraint, conditioned on agreement at common variables.
In this work, we generalize such a distribution for $t$-wise independent CSPs as follows.
For each constraint $C$, there is a $t$-wise $\nu$-independent distribution $\mu^C$ of satisfying assignments.
The distribution $\mu^C$ corresponds to a probability density $f^C$ with respect to the product distribution $\nu^k$.
Our canonical distribution $\mu_J$ to a small subinstance $J$ with variable set $V$ and constraint set $\cC$ is then
given by
\[ \mu_J (b) \coloneqq \frac{\prod_{C \in \cC} f_C (b) \nu^V (b)}{\sum_{a \in D^V} \prod_{C \in \cC} f_C (a) \nu^V (a)}
\qquad \text{for } b \in D^V. \]
Using the LP lower bound framework of \cite{chan2024how}, we prove LP lower bounds for $t$-wise independent CSPs,
strictly generalizing the pairwise uniform lower bounds in Benabbas, Georgiou, Magen, and Tulsiani
\cite{benabbas12sdp}.
A variant of this distribution applied to the frameworks of \cite{chan2024how,chan2025how} also leads to lower bounds
for combined hierarchies such as SDP+AIP and C(BLP+AIP).

For the SoS lower bound, one further needs to show that the canonical distribution satisfies a certain conditional
independence property (\cref{lem:cond-indep}), in order to decompose the SDP solution into orthogonal subparts.
\tnote{Can you add one line to explain why this requirement is needed?}
We interpret the above canonical distribution as a Markov random field.
It turns out our canonical distribution factorizes according to cliques, and the required conditional independence then follows from the global Markov property.
Plugging the global Markov property into \cite{chan2024how} yields our SoS lower bound.

Our work is the first to apply the recent frameworks of \cite{chan2024how,chan2025how} to prove LP and SoS lower bounds
beyond $t$-wise uniform CSPs.
Our SoS lower bounds do not to follow from the pseudo-calibration framework of Barak, Hopkins, Kelner, Kothari, Moitra,
and Potechin \cite{BHKKMP19} that is conjectured to give a universal recipe for lower bounds
\cite{PR20,JPRTX,JPRX23,KPX24,PX25,Xu26}.\tnote{Can you elaborate on this? Do you want to say that the design process of the planted distribution is not "canonical"?}
Our work is also the first to apply Markov random field techniques to SoS lower bounds for CSPs.

\section{Preliminaries} \label{sec:preliminaries}

\subsection{Notations}
Throughout this work, we use $O_k, \Omega_k$ to hide constant dependence on some parameter $k$ (and analogously in $|D|, |\calP|$).
Write $f \lesssim_k g$ or $g \gtrsim_k f$ if $f \leq O_k(g)$.
With some abuse of notation, we use $\binom{n}{r}$ to denote the collection of all $r$-size subset in $[n]$. We use $[n]^r$ to denote all $r$-size order-tuples from $[n]$. 
Given $n, k \in \N$, let $n^{\underline k} \coloneqq n(n-1)\cdots(n-k+1)$ denote the falling factorial.
Given any set $W$, let
$W^{\underline k} \coloneqq \set{ \alpha \in W^k | \alpha_i \neq \alpha_j \text{ for distinct } i,j \in [k] }$
denote the set of $k$-tuples of distinct elements from $W$ (so that $|W^{\underline{k}}| = |W|^{\underline{k}}$).
$\bracbb{\cdot}$ is the Iverson bracket notation to denote an indicator function.

We will consider directed hypergraphs whose directed hyperedges are ordered tuples of vertices.
For a directed hypergraph $\calH$ or a random instance, we let $V(\calH) =[n]$ denote the ground set of its vertices, and $E(\calH)$ its collection of hyperedges (viewed as order-tuples).
All our hypergraphs are directed unless otherwise specified.
As a standard in related words, high probability stands for probability at least $1-o_n(1)$ over the randomness of input.

\subsection{CSP}
\begin{definition}[$k$-CSP]\label{def:csp}
  A $k$-CSP is given by a finite domain $D$ and a collection $\cR$ of $k$-ary relations over $D.$  A instance
  $\cI$ of a $k$-CSP $(D,\cR)$ is a pair $(V,\cC).$ $V$ is the set of variables and $\cC\coloneqq\set{\Paren{\sigma_i, R_i}}$
  is the multi-set of constraints where  $R\in \cR$ and $\sigma_i\,:[k]\to V$ is an injection.
\end{definition}
This model generalizes the model of CSPs with predicates and literals. For example, in $3$-SAT $D=\Set{0,1}$
and $\cR=\Set{R_{000},\ldots,R_{111}}\,,$ where $R_b\coloneqq\Set{0,1}^3\setminus \set{b}$ for $b\in\Set{0,1}^3.$

For a given instance $\cI=(V,\cC)$, an assignment is a vector $x\in D^n$ (or alternatively a function $V\to D$).
The value of $x$ is
\begin{align*}
  \Val_\cI(x) \coloneqq \Pr_{(\sigma,R)\sim \cC} [x\circ \sigma \in R].
\end{align*}
Here we used $(\sigma,R)\sim \cC$ to denote a uniformly random tuple in $\cC$.
The value of an instance is
\begin{align*}
  \opt(\cI)=\max_{x\in D^n}\Val_\cI(x)\,.
\end{align*}
A $k$-uniform multi-hypergraph $H$ is a pair $(V,E)$ where $V$ is the set of vertices and $E$ is the
multi-set of hyperedges of arity $k$. We treat each occurrence of each hyper-edge distinctly.

\begin{definition}[$\rho$-random instance]\label{def:random-instance}
  For a set of $n$ variables $V$, we consider the following random model, parametrized by a fixed
  distribution $\rho$ over $k$-ary relations over $D$:
  \begin{enumerate}
    \item Sample a $k$-uniform undirected hypergraph $H=(V,E)$ from the binomial model (proposed by Gilbert, but often
      called \Erdos--\Renyi). That is, every potential hyperedge $e \in \binom{V}{k}$ appears independently in $H$ with
      probability $p$, so that the expected number of hyperedges is $p \binom{|V|}{k}$.
    \item Independently for each undirected hyperedge $e\in E$ realized in $H$, choose a bijection $\sigma\,:[k]\to e$
      uniformly at random and draw a relation $R\subseteq D^k$ from $\rho$.
      Impose on $e$ the constraint with satisfying assignments $\set{a\in D^e\,|\, a\circ\sigma\in R}.$
  \end{enumerate}
\end{definition}

As indicated in \cref{def:random-instance}, we state our main results in the binomial hypergraph model.
Many of our results also hold the alternative model, where Step 1. \jnote{meant to capitalize Step?}~above is changed to so that we sample $m$ undirected
hyperedges from $\binom{V}{k}$ independently with replacement.
As mentioned in \cite{AOW15}, the precise details of the model don't matter, and we expect all our results to hold in
various well-studied models.
The above random model is from \cite{molloy2003models} and has been studied in different hierarchies
\cite{chan2013dichotomy,chan2024how,chan2025how,conneryd2026lower}.

\begin{definition}[$t$-wise independent CSP] \label{def:t-wise-independent-csp}
  A collection $\cR$ of $k$-ary relations over $D$ is $t$-wise independent if there are a distribution $\nu$ over $D$
  and a tuple $\Paren{\mu^R}_{R\in\cR}$ of $t$-wise $\nu$-independent distributions over $D^k$ such that
  $\supp\set{\mu^R}\subseteq R$ for every $R\in\cR$.
  We also say a $k$-CSP is $t$-wise independent if its collection $\cR$ of $k$-ary relations is.
\end{definition}

Our next definition generalizes (one minus) ``far from $t$-wise supporting'' \cite[Definition~3.14]{AOW15} from $t$-wise
uniform CSP to $t$-wise independent CSP.

\begin{definition} \label{def:t-wise-value}
  Given a CSP $(D, \mathcal R)$ and a distribution $\rho$ over $\mathcal R$, define the $t$-wise
  independent value to
  be the maximum expected fraction of constraints satisfied by any tuple $\Paren{\mu^R}_{R\in\cR}$ of $t$-wise
  $\nu$-independent distributions, when the constraints are chosen from $\rho$:
  \[ \opt_t(\rho) \coloneqq \max \Set{\E_{R \sim \rho} \Pr_{x \sim \mu^R} [x \in R] | \Paren{\mu^R}_{R\in\cR}
  \text{ is $t$-wise $\nu$-independent for some $\nu$}}\,. \]
\end{definition}

\paragraph{Fourier Analysis for CSP} \jnote{i moved the old prelim at spectral here}

\begin{definition}[Monomial restriction of a predicate]
We call a subset $S\subseteq [k]$ a monomial restriction. Additionally, we write $|S|$ to  be the weight of the restriction.
\end{definition}

\begin{definition}[Fourier expansion of a Boolean predicate]
For any $k$-CSP predicate $P:\{\pm 1\}^k \rightarrow \{0, 1\}$, we can write \[ 
P(x) =  \sum_{\substack{S \subseteq [k]}} c_S \cdot x_S\]
for some coefficients $\{c_S\}$ where we write $
x_S = \prod_{i\in S} x_i $.
\end{definition}

\begin{definition}[Global Polynomial for Monomial-Restriction]\label{def:monomial-global}
  Let $\calH$ be a directed hypergraph.
	For any subset $S \subseteq [k]$, we associate it with a polynomial \[ 
	\Psi_{S}(x) = \frac{1}{|E(\calH)|} \sum_{ \substack{ e \in E(\calH) \\ \tau_e: [k] \rightarrow [n]  }} x_{\tau(e)_S}
	\] 
	where we write $x_{\tau(e)_S} \coloneqq  \prod_{i\in S} x_{\tau_e(i)}  $ as the monomial restricted to the coordinates in $S$.
\end{definition}

\section{Sum-of-Squares Refutation}\label{sec:sos-refutations}

\subsection{Independence via Conditioning} \label{sec:independence-via-conditioning}

In this section, we show that conditioning can be used to make pseudo-distribution approximately independent over random
$k$-CSP instances.

In the following, $\Omega_k = \{z \in \set{0,1}^{[k] \times D} | \sum_{a \in D} z(i,a) = 1\}$ denotes boolean
version of assignments from $D^k$ (see \cref{sec:func-non-boolean}).
Given a function $f^\Omega: \Omega_k \to \R$ and an injective map $\sigma: [k] \to V$ (representing the scope of a
constraint), $f^{\Omega,\sigma}$ is the function applying $f^\Omega$ to the variables in the constraint (see
\cref{sec:tensor-func}).
Key to our refutation is the concentration of $t$-wise marginals of the induced distribution:

\begin{definition} \label{def:sos-axioms}
  The set of axioms for the SoS proof is
  \[ \Axioms \coloneqq \Set{\sum_{a \in D} x_{v,a}^2 = 1}_{v \in V}
  \cup \Set{ x_{v,a}x_{v,a'} = 0}_{v \in V, a,a' \in D, a \neq a'} \]
\end{definition}

\begin{definition} \label{def:concentration-marginals}
  Given $\eps > 0, k, t, d \in \N$, a $k$-CSP instance $(V, \cC)$ is $(t,\eps,d)$-concentrated if there is a degree-$d$
  SOS proof of
  \[
    \Axioms \quad \vdash_d \quad \Paren{\E_{\sigma \in \mathcal C}
    f^{\Omega,\sigma} (x) - \E_{\sigma \in V^{\underline k}} f^{\Omega,\sigma} (x)}^2 \leq \eps^2
  \]
  for every function $f^\Omega: \Omega_k \to \R$ of the form $f^\Omega(x) = \prod_{i \in Q} x_{i,b_i}$, where
  $Q \in {[k] \choose t}, b \in D^Q$.
\end{definition}

As we will show in \cref{lem:concentrated}, the above concentration property holds with high probability.

\begin{theorem} \label{thm:t-wise-conditioning}
  For any $\eps > 0$, any $t \in \N$, finite $D$, there is $K_{t,\delta,D} \in \N$ such that for any sufficiently large
  variable set $V$, there is a distribution $\cT_{t,\delta,D}$ over ${V \choose \leq K}$ such that if a $k$-CSP instance
  $(V, \cC)$ is $(t,\eps,d)$-concentrated, then every pseudo-distribution $\mu$ of degree $K_{t,\delta,D} + d$
  consistent with $\cP(V,D)$ satisfies
  \begin{equation} \label{eq:conditional-t-wise-independent}
    \E_{T,\alpha} \Norm{\E_{\sigma \in \cC} (\mu|\alpha)_{\sigma(Q)}
    - \Paren{\E_{v \in V} (\mu|\alpha)_v}^t }_{\ell_1} \leq 2\abs{D}^t \eps
  \end{equation}
  for every $Q \in {[k] \choose t}$.
  Here the outer expectation is over $T \sim \cT, \alpha \sim \alpha(T,\mu)$.
\end{theorem}

\begin{proof}
  For every $Q \in {[k] \choose t}$, $b \in D^t$, write
  \[ f \coloneqq \prod_{i \in Q} x_{i,b_i}, \quad A \coloneqq \E_{\sigma \in \mathcal C}
    f^{\Omega,\sigma} (x), \quad B \coloneqq \E_{\sigma \in V^{\underline k}} f^{\Omega,\sigma} (x). \]
  For every degree-$d$ pseudo-distribution $\eta$ consistent with $\cP(V,D)$,
  \[
    \Abs{\E_{\sigma \in \cC} \eta_{\sigma(Q)}(b) - \E_{S \sim V^{\underline t}} \eta(b)}
    = \Abs{\pE_\eta[A - B]} \leq \sqrt{\pE_\eta [(A-B)^2]} \leq \eps \,,
  \]
  where the first inequality is pseudo-expectation Cauchy--Schwarz, and the second is because $(V, \cC)$ is
  $(t,\eps,d)$-concentrated.
  Therefore
  \[
    \Norm{\E_{\sigma \in \cC} \eta_{\sigma(Q)} - \E_{S \sim V^{\underline t}} \eta_S}_{\ell_1} =
    \sum_{b \in D^t} \Abs{\E_{\sigma \in \cC} \eta_{\sigma(Q)}(b) - \E_{S \sim V^{\underline t}} \eta_S(b)}
    \leq \abs{D}^t \eps \,.
  \]
  Further, \cref{lem:conditional-product} implies
  \[
    \E_{T,\alpha} \Norm{\E_{S \sim V^{\underline t}} (\mu|\alpha)_S
    - \Paren{\E_{v \in V} (\mu|\alpha)_v}^t}_{\ell_1} \leq \delta
  \]
  for every degree-$K_{t,\delta,D}$ pesudo-distribution $\mu$ consistent with $\cP(V, D)$.
  The last two inequalities with $\eta = \mu|\alpha$ imply \cref{eq:conditional-t-wise-independent}.
\end{proof}

For $t \in N$, $n \geq 2t$, consider the modified Kneser graph $K'(n,t)$, whose vertex set is
$[n]^{\underline t}$ (i.e.~ordered $t$-tuples of distinct elements from $[n]$), and there is an edge between
$S, T \in [n]^{\underline t}$ if $S$ and $T$ are disjoint.
By contrast, the vertices in the usual Kneser graph $K(n,t)$ are unordered $t$-subsets of $[n]$.
Since every unordered $t$-subset gives rise to $t!$ ordered $t$-tuples, $K'(n,t)$ can be obtained from $K(n,t)$ by
replacing every vertex $v$ of $K(n,t)$ with a cloud $L_v$ of $t!$ vertices, and replacing every edge $(u,v)$ of $K(n,t)$
with a complete bipartite graph between $L_u$ and $L_v$.

\begin{lemma} \label{lem:conditional-product}
  For any $\delta > 0$, $t \in \N$, finite $D$, there is $K_{t,\delta,D} \in \N$ such that any large enough vertex
  set $V$ has a distribution $\cT_{t,\delta,D}$ over ${V \choose \leq K}$, such that for any pseudo-distribution
  $\mu$ of degree $K_{t,\delta,D}+t$ consistent with $\cP(V,D)$,
  \begin{equation} \label{eq:conditional-product}
    \E_{T,\alpha} \Norm{\E_{S \sim V^{\underline t}} (\mu|\alpha)_S
    - \Paren{\E_{v \in V} (\mu|\alpha)_v}^t}_{\ell_1} \leq \delta \,,
  \end{equation}
  where the outer expectation is over $T \sim \cT_{t,\delta,D}, \alpha \sim \alpha(T,\mu)$.
\end{lemma}

\begin{proof}
  We first prove the lemma for $t$ being a power of $2$ by induction.
  The base case $t = 1$ is trivial, because $\E_{S \sim V^{\underline 1}} (\mu|\alpha)_S = \E_{v \in V} (\mu|\alpha)_v$.
  In this case we can choose $K_{1,\delta,D} = 0$ and  $\cT_{1,\delta,D}$ to be always $\emptyset$.

  For any power of two $t > 1$, we want to show that for large enough $n$,
  \begin{equation} \label{eq:induction-conditional-product}
    \E_{\tau,T} \E_{T'} \E_\alpha \Norm{\E_{S \sim V^{\underline t}} (\mu|\alpha)_S -
    \Paren{\E_{v \in V} (\mu|\alpha)_v}^t }_{\ell_1}
    \leq \delta\,,
  \end{equation}
  where $\tau \sim [(\abs{D}/\delta)^{O(1)}], T \sim (V^{\underline{t/2}})^\tau$, independently $T' \sim
  \cT_{t/2,\delta/4,D}$, and also $\alpha \sim \alpha(T \cup T',\mu)$.
  Write $S \in V^{\underline t}$ as $S = XY$, where $X,Y \in V^{\underline{t/2}}$.
  Expand
  \[
    \E_{S \sim V^{\underline t}} (\mu|\alpha)_S - \Paren{\E_{v \in V} (\mu|\alpha)_v}^t
    = A + B + C,
  \]
  where
  \begin{align*}
    A &= \E_{S \sim V^{\underline t}} (\mu|\alpha)_S
    - \E_{XY \sim V^{\underline t}} (\mu|\alpha)_X(\mu|\alpha)_Y \\
    B &= \E_{XY \sim V^{\underline t}} (\mu|\alpha)_X(\mu|\alpha)_Y
    - \Paren{\E_{X \sim V^{\underline{t/2}}} (\mu|\alpha)_X}^2 \\
    C &= \Paren{\E_{X \sim V^{\underline{t/2}}} (\mu|\alpha)_X}^2
    - \Paren{\E_{v \in V} (\mu|\alpha)_v}^t\,.
  \end{align*}

  To bound $A$, let $K = (\abs{D}/\eps)^{O(1)}$.
  For any degree-$(K+t)$ pseudo-distribution $\eta$, in expectation over
  $\tau \sim [K], T \sim (V^{\underline{t/2}})^\tau, \alpha \sim \alpha(T,\eta)$,
  \begin{equation} \label{eq:kneser-conditioning}
    \begin{split}
      & \E_{\tau,T,\alpha} \Norm{\E_{S \sim V^{\underline t}} (\eta|\alpha)_S
      - \E_{XY \sim V^{\underline t}} (\eta|\alpha)_X(\eta|\alpha)_Y}_{\ell_1} \\
      \leq & \E_{\tau,T,\alpha} \E_{XY \sim V^{\underline t}} \Norm{(\eta|\alpha)_{XY}
      - (\eta|\alpha)_X(\eta|\alpha)_Y}_{\ell_1} \leq \frac\delta4
    \end{split}
  \end{equation}
  by \cref{thm:local-global-correlation} below (from \cite{BRS11}) applied to modified Kneser graph
  $K'(n,t/2)$, which is expanding by \cref{lem:kenser-expanding}.
  Therefore
  \[ \E_{\tau,T} \E_{T'} \E_\alpha \norm{A}_{\ell_1} = \E_{T'} \E_{\tau,T} \E_\alpha \norm{A}_{\ell_1}
    = \E_{\substack{T' \\ \alpha' \sim \alpha(T',\mu)}} \E_{\substack{\tau,T \\ \alpha \sim
  \alpha(T,\mu|\alpha')}} \norm{A}_{\ell_1} \leq \frac\delta4 \,, \]
  where the first equality uses the independence of $(\tau,T)$ and $T'$, and the second equality
  uses the fact that sampling $\alpha'' \sim \alpha(T \cup T',\mu)$ is the same as sampling $\alpha' \sim
  \alpha(T',\mu)$, $\alpha \sim \alpha(T,\mu|\alpha')$, and setting $\alpha'' = \alpha \cup \alpha'$.
  The last inequality is \cref{eq:kneser-conditioning} applied to $\eta = \mu|\alpha'$.
  To bound $B$, \cref{lem:independent-error} implies
  \[ \Norm{\E_{XY \sim V^{\underline t}} (\mu|\alpha)_X(\mu|\alpha)_Y
  - \Paren{\E_{X \sim V^{\underline{t/2}}} (\mu|\alpha)_X}^2}_{\ell_1} = o_n(1)\,, \]
  so $\E_{\tau,T} \E_{T'} \E_\alpha \norm{B}_{\ell_1} \leq \delta/4$ for large enough $n$.

  To bound $C$, induction hypothesis implies for any degree-$(K_{t/2,\delta/4,D} + t/2)$ peusdo-distribution $\eta$, in
  expectation over $T' \sim \cT_{t/2,\delta/4,D}, \alpha' \sim \alpha(T',\eta)$,
  \[
    \E_{T',\alpha'} \Norm{\E_{X \sim V^{\underline{t/2}}} (\eta|\alpha')_X
    - \Paren{\E_{v \in V} (\eta|\alpha')_v}^{t/2}}_{\ell_1} \leq \frac\delta4 \,.
  \]
  In particular when $\eta = \mu|\alpha$,   \cref{lem:product-close} implies
  \[ \E_{\tau,T} \E_{T'} \E_\alpha \norm{C}_{\ell_1}
    = \E_{\substack{\tau,T \\ \alpha \sim \alpha(T,\mu)}} \E_{\substack{T' \\ \alpha' \sim \alpha(T', \mu|\alpha)}}
  \norm{C}_{\ell_1} \leq \frac\delta2\,. \]

  \cref{eq:induction-conditional-product} now follows, implying \cref{eq:conditional-product} for $t$.
  Here $\cT_{t,\delta,D}$ is the distribution of the underlying set of $T \cup T'$, with $T' \sim \cT_{t/2,\delta/4,D}$,
  and independently $\tau \sim [(\abs{D}/\delta)^{O(1)}], T \sim (V^{\underline{t/2}})^\tau$.

  The general case of arbitrary $t$ will then follow by considering $u$ to be the smallest power of $2$ such that $u
  \geq t$, and using the fact that for any pseudo-distribution $\eta$ of degree $u$,
  \[
    \Norm{\E_{S \sim V^{\underline t}} \eta_S - \Paren{\E_{v \in V} \eta_v}^t}_{\ell_1}
    \leq \Norm{\E_{U \sim V^{\underline u}} \eta_U - \Paren{\E_{v \in V} \eta_v}^u}_{\ell_1} \,,
  \]
  which holds by \cref{lem:pushforward-dist} applied to the projection $(b \in D^u \mapsto b_{[t]})$, and also the fact
  that
  \[
    \E_{S \sim V^{\underline t}} \eta_S = \pi_{[t]} \Paren{\E_{U \sim V^{\underline u}} \eta_U}
    \quad\text{and}\quad
    \Paren{\E_{v \sim V} \eta_v}^t = \pi_{[t]} \Paren{\Paren{\E_{v \sim V} \eta_v}^u} \,. \qedhere
  \]
\end{proof}

Given a real symmetric matrix $A$ of dimension $n$, let $\lambda_1 (A) \geq \dots \geq \lambda_n (A)$ be its
eigenvalues listed in decreasing order.
For a graph $G = (V,E)$, let $\lambda_i(G)$ be the eigenvalues of its adjacency matrix.

\begin{lemma} \label{lem:kenser-expanding}
  For every $t \in \N$, $\lambda_2(K'(n,t)) = o_n(\lambda_1(K'(n,t)))$.
\end{lemma}

\begin{proof}
  The adjacency matrix of $K'(n,t)$ is $A \otimes J$, where $A$ is the adjacency matrix of the usual Kneser graph
  $K(n,t)$, and $J$ is the all-ones square matrix of size $(t/2)!$.

  This is because for every adjacent vertices $S$ and $T$ in $K'(n,t)$, their (unordered) underlying sets are adjacent
  in $K(n,t)$, the matrix $J$ is the adjacency matrix of the $(t/2)!$-clique with self-loops, and vertices of the clique
  specify orderings of the (ordered) tuples $S$ and $T$.

  From the eigenvalues of the Kneser graph \cite[Section~9.4]{godsil2001algebraic}, we have $\lambda_1(A)
  = {n-t/2 \choose t/2}$ and $\abs{\lambda_2(A)} \leq {n-t/2-1 \choose t/2-1} = o(\lambda_1(A))$.
  For $J$, $\lambda_1(J) = (t/2)!$ and all other eigenvalues are zeros.
  \cref{fact:eigenvalues-kronecker} implies $\lambda_1(K'(n,t)) = \lambda_1 (A) \lambda_1(J)$.
  Also $\lambda_2(K'(n,t)) = \lambda_2(A) \lambda_1(J)$ because $\lambda_2(A) > 0$ when $t \geq 4$, and $J$ is
  a $1$-by-$1$ matrix when $t = 2$.
  Therefore $\lambda_2(K'(n,t)) = o(\lambda_1(K'(n,t)))$.
\end{proof}

\begin{fact} \label{fact:eigenvalues-kronecker}
  Given real symmetric matrices $A$ and $B$ of dimension $n$ and $m$ respectively, the eigenvalues of $A \otimes B$
  are $\set{\lambda_i (A) \lambda_j (B)}_{1 \leq i \leq n, 1 \leq j \leq m}$.
\end{fact}

The following result is a special case of \cite[Proof of Theorem~5.6]{BRS11}.

\begin{theorem} \label{thm:local-global-correlation}
  Given a regular graph $G = (V, E)$ with $\lambda_2(G) \leq \Omega(\eps/\abs{D})^2 \lambda_1(G)$, there is
  $K = O(\abs{D}/\eps)^{O(1)}$ such that every degree $K+2$ pseudo-distribution $\mu$ consistent with
  $\cP(V,D)$ satisfies
  \[ \E_{\tau,T,\alpha} \E_{uv \sim E} \norm{(\mu|\alpha)_{uv} - (\mu|\alpha)_u (\mu|\alpha)_v}_{\ell_1} \leq \eps\,, \]
  where the outer expectation is over $\tau \sim [K], T \sim V^\tau, \alpha \sim \alpha(T,\mu)$.
\end{theorem}

\subsection{Almost $t$-wise Independence}

The goal of this section is to prove \cref{thm:almost-twise}, showing that a distribution whose $t$-wise
marginals are close to $\nu$-independent is fact close to some $t$-wise $\eta$-independent distribution.
Note that $\eta$ may be different from $\nu$.
Our proof generalizes \cite[Theorem~1.1]{odonnell18closness} from the uniform Boolean distribution to an
arbitrary finite distribution.
We also need to handle an additional subtlety to ensure every nonzero mass of $\nu$ is not too small.

Two distributions $\mu$ on $\eta$ on the same probability space are $\eps$-close if $\norm{\mu - \eta}_{\ell_1} \leq
\eps$.

We will consider distributions that are almost $t$-wise independent, as measured by the following metric:

\begin{definition}
  Given a distribution $\nu$ over $D$ and $t \in \mathbb{N}$, define for every distribution $\mu$ over $D^k$,
  \[ \delta^{[1,t]}_\nu [\mu] \coloneqq \sum_{1 \leq \abs{Q} \leq t} \Norm{\pi_Q (\mu) - \nu^Q}_{\ell_1}. \]
\end{definition}

Call a distribution $\mu$ $(\delta,t)$-wise $\nu$-independent if $\delta^{[1,t]}_\nu [\mu] \leq \delta$.
Also call a distribution $(\delta,t)$-wise independent if it is $(\delta,t)$-wise $\nu$-independent for some
distribution $\nu$.
Note that $\mu \mapsto \delta^{[1,t]}_\nu[\mu]$ is convex.

\begin{remark}
  Our notion of $(\delta,t)$-wise independent distribution differs from \cite{odonnell18closness}.
  The definition in \cite{odonnell18closness} is based on an upper bound on magnitudes of all low-degree Fourier
  coefficients, while our definition is a sum of $\ell_1$ distances over $t$-wise marginals.
  Our definition is more convenient for the global correlation rounding analysis.
\end{remark}

Given any finite distribution $\nu$, $p \geq 1$, define $\norm{f}_{L^p(\nu)} \coloneqq \Paren{\E_{a \in \nu}
[\abs{f(x)}^p]}^{1/p}$ for $f: \supp(\nu) \to \R$.

Given any distribution $\nu$ over a finite domain $D$, distribution $\mu$ with density $f$ wrt $\nu$,
distribution $\eta$ with density $g$ wrt $\nu$,
\begin{equation} \label{eq:one-l1}
  \norm{f - g}_{L^1(\nu)} = \sum_{a \in D} \nu (a) \abs{f(a) - g(a)} = \sum_{a \in D} \Abs{f(a) \nu (a) -
  g(a)\nu (a)} = \Norm{\mu - \eta}_{\ell_1}.
\end{equation}

Our next definition makes use of Efron--Stein decomposition (\cref{sec:efron-stein}).
Given a function $f \in L^2(\nu^k)$ and $t \leq k$, define its level-$1$-through-$t$-weight
\[ W^{[1,t]}_\nu [f] \coloneqq \sum_{1 \leq \abs{Q} \leq t} \norm{f^Q}^2_{L^2(\nu^k)}.\]
Note that $\1$ is the density of $\nu^k$ wrt $\nu^k$, and $\1^Q \equiv 0$ for any nonempty $Q \subseteq [k]$.
Closeness in $1$-norm implies closeness in Efron–Stein components:

\begin{lemma} \label{lem:es-1norm}
  For any distribution $\nu$ of minimum nonzero probability $\alpha$, any distribution $\mu$ having density
  $f \in L^2(\nu^k)$ wrt $\nu^k$,
  \[ W^{[1,t]}_\nu [f] \leq 2(4/\alpha)^t \delta^{[1,t]}_\nu [\mu].\]
\end{lemma}

\begin{proof}
  Let $B = 2/\sqrt\alpha$.
  Abbreviate $\norm{\cdot} \coloneqq \norm{\cdot}_{L^2(\nu^k)}$ and $\norm{\cdot}_1 \coloneqq \norm{\cdot}_{L^1(\nu^k)}$.
  For every nonempty $Q \in {[k] \choose \leq t}$,
  \begin{align*}
    \Norm{f^Q} &= \Norm{f^Q - \1^Q}^2 \leq \sum_{P \subseteq Q} \Norm{f^P - \1^P}^2
    \overset{(*)}= \Norm{f^{\subseteq Q} - \1^{\subseteq Q}}^2
    \leq B^{2t} \Norm{f^{\subseteq Q} - \1^{\subseteq Q}}^2_1 \\
    &= B^{2t} \Norm{f^{\subseteq Q}\vert_Q - \1^{\subseteq Q}\vert_Q}^2_1
    \overset{\eqref{eq:one-l1}}= B^{2t} \Norm{\pi_Q (\mu) - \nu^Q}^2_{\ell_1}
  \end{align*}
  where $(*)$ is Parseval's theorem (i.e.~orthogonality of the Efron--Stein decomposition), the last inequality is
  \cref{thm:2-1-norm}, and the last equality uses \cref{lem:proj-density}.
  Therefore
  \begin{align*}
    W^{[1,t]}_\nu [f] &= \sum_{1 \leq \abs{Q} \leq t} \Norm{f^Q - \1^Q}^2 \leq B^{2t} \sum_{1 \leq \abs{Q}
    \leq t} \Norm{\pi_Q (\mu) - \nu^Q}^2_{\ell_1} \\
    &\leq B^{2t} \cdot \sum_{1 \leq \abs{Q} \leq t} \Norm{\pi_Q (\mu) - \nu^Q}_{\ell_1} \cdot
    \max_{1 \leq \abs{Q} \leq t} \Norm{\pi_Q (\mu) - \nu^Q}_{\ell_1}
    \leq 2 B^{2t} \delta^{[1,t]}_\nu [\mu]\,,
  \end{align*}
  where the last inequality is because any two distributions has $\ell_1$-distance at most two.
\end{proof}

\begin{theorem} \label{thm:almost-twise}
  For any finite domain $D$, any $t \geq 1$, there is a polynomial $q$ in $k$ and $\abs{D}$ such that the
  following holds:
  For any distribution $\nu$ over $D$, there is a distribution $\eta$ over $D$, such that any $(\delta,t)$-wise
  $\nu$-independent distribution $\mu$ over $D^k$ is $q(k,\abs{D})\sqrt\delta$-close to some $t$-wise
  $\eta$-independent distribution.
\end{theorem}

\cref{thm:almost-twise} generalizes \cite[Theorem~1.1]{odonnell18closness} in that $\mu$ is almost $t$-wise
$\nu$-independent, where $\nu$ needs not be the uniform Boolean distribution.
The distance bound deteriorates with the minimum nonzero probability of $\nu$ (due to hypercontractivity inequality in
\cref{thm:2-1-norm}), so we have to first remove tiny probability masses from $\nu$ using \cref{lem:close-min-prob}.
Note that $\eta$ only depends on $\nu$ but not $\mu$, which will be crucial for our application.

\begin{proof}[{Proof of \cref{thm:almost-twise}}]
  Let $\alpha = 1/(k\abs{D}^2)$, $A = 2\abs{D}(k\alpha+\delta)$, $B = 2/\sqrt\alpha$, $w = B^{2t}\sqrt{2\delta}$.

  \cref{lem:close-min-prob} implies that $\mu$ is $A$-close to $\zeta$ so that $\zeta$ is $(t,\delta)$-wise
  $\eta$-independent, where $\eta$ has minimum nonzero probability at least $\alpha$.
  For the rest of the proof, all expectations are over $\eta^k$.
  Also, $\norm{\cdot} \coloneqq \norm{\cdot}_{L^2(\eta^k)}$, and $\norm{\cdot}_1 \coloneqq \norm{\cdot}_{L^1(\eta^k)}$.

  Let $f$ be the density of $\zeta$ wrt $\eta^k$.
  We will find another density $g$ wrt $\eta^k$ so that
  \[ f^Q = -w g^Q \qquad \forall 1 \leq \abs{Q} \leq t \,. \]
  Then $\frac{f+w g}{1+w}$ is the density wrt $\eta^k$ of a distribution which is
  $2w$-close to $\zeta$, because
  \[ \Norm{\frac{f+w g}{1+w} - f}_1 = \frac w{1+w} \norm{g - f}_1 \leq 2w. \]

  Such a density $g$ exists if and only if the following LP is feasible:
  \begin{align*}
    \E[g] & = 1 \\
    g^Q & = -\frac1w f^Q \quad \forall Q \subseteq [k], 1 \leq \abs{Q} \leq t \\
    g(a) & \geq 0 \quad \forall a \in D^k.
  \end{align*}
  The dual has variables $h: D^k \to \mathbb{R}$ with constraints
  \begin{align*}
    \E [h] & = 1 \\
    h^Q & \equiv 0 \quad \forall Q \in \textstyle{[k] \choose >t} \\
    h(a) & \geq 0 \quad \forall a \in D^k \\
    \frac1w \sum_{1 \leq \abs{Q} \leq t} \E \Brac{f^Q h^Q} & > 1.
  \end{align*}
  By Farka's lemma, the primal LP is feasible if and only if the dual is not.
  To verify dual infeasibility,
  \begin{align*}
    \sum_{1 \leq \abs{Q} \leq t} \E \Brac{f^Q h^Q} &\leq \sum_{1 \leq \abs{Q} \leq t} \Norm{f^Q} \Norm{h^Q}
    \leq \sqrt{\sum_{1 \leq \abs{Q} \leq t} \Norm{f^Q}^2} \sqrt{\sum_{1 \leq \abs{Q} \leq t} \Norm{h^Q}^2} \\
    & \leq \sqrt{W^{[1,t]}_\eta [f]} \norm{h} \overset{(*)}\leq B^t \sqrt{2\delta^{[1,t]}_\eta [\zeta]} \norm{h}
    \leq B^t\sqrt{2\delta} \norm{h}
  \end{align*}
  where the first two inequalities are Cauchy–Schwarz, and $(*)$ is \cref{lem:es-1norm}.
  Since $\norm{h} \leq B^t \norm{h}_1$ by \cref{thm:2-1-norm} and $\norm{h}_1 = 1$, the dual LP is infeasible.

  Altogether, $\mu$ is $(A+2w)$-close to some $t$-wise $\eta$-independent distribution, and $A+2w =
  O(2^t k^t \abs{D}^{2t} \sqrt\delta)$.
\end{proof}

\begin{theorem} \label{thm:2-1-norm}
  Given any distribution $\nu$ with minimum nonzero probability $\alpha$, for every degree-$d$ polynomial $f
  \in L^2(\nu^k)$,
  \[ \norm{f}_{L^2(\nu^k)} \leq (2/\sqrt\alpha)^d \norm{f}_{L^1(\nu^k)}. \]
\end{theorem}

\begin{proof}
  If $\alpha > 1/2$, then in fact $\alpha = 1$ and $\abs{\supp(\nu)} = 1$, and the theorem follows.
  Assume from now on $\alpha \leq 1/2$.
  \cite[Theorem~2.7]{AustrinHastad2011} implies that for every $p \geq 2$, every degree-$d$ polynomial $f$
  \[ \norm{f}_{L^p(\nu^k)} \leq C_p^{d/2} \norm{f}_{L^2(\nu^k)}, \]
  where $C_p = C_p(\alpha)$ comes from \cite{wolff07hypercontractivity} and depends on $p$ and $\alpha$.
  \cite[Theorem~2.9]{AustrinHastad2011} further implies
  \[ \norm{f}_{L^2(\nu^k)} \leq \Paren{C_p^{p/(p-2)}}^{d/2} \norm{f}_{L^1(\nu^k)}. \]
  \cite[Fact~2.8]{AustrinHastad2011} implies $C_3 \leq (4/\alpha)^{1/3}$, so the theorem follows.
\end{proof}

Given any $x \in D^k$, $a,b \in D$, let $x_{a \to b} \in D^k$ be obtained from $x$ by replacing every
occurrence of $a$ with $b$:
\[
  x_{a \to b} (i) \coloneqq
  \begin{cases}
    x(i) & \text{if } x(i) \neq a \\
    b & \text{if } x(i) = a
  \end{cases} \qquad \forall i \in [k] \,.
\]

Given any distribution $\mu$ over $D^k$, any $a,b \in D$, denote by $\mu_{a \to b}$ the distribution over $D^k$ obtained
by sampling $x \sim \mu$ and then replacing every occurrence of $a$ with $b$ in $x$:
\[ \mu_{a \to b}(c) \coloneqq \sum_{x \in D^k, x_{a \to b} = c} \mu(x) \qquad \forall c \in D^k\,. \]

Given $a \in D$, let $D^k_{\ni a} \coloneqq \set{x \in D^k | x(i) = a \text{ for some } i \in [k]}$ be the set of strings
containing $a$.

\begin{lemma} \label{lem:dist-replace}
  Given any distribution $\mu$ over $D^k$, any $a,b \in D$, $\norm{\mu - \mu_{a\to b}}_{\ell_1} \leq
  2\mu\Paren{D^k_{\ni a}}$.
\end{lemma}

\begin{proof}
  Let $W$ be the set of $x \in D^k$ such that $\mu_{a \to b}(x) < \mu(x)$.
  If $x \in D^k$ does not contain $a$, then $x_{a \to b} = x$, so $\mu_{a \to b}(x) \geq \mu(x)$.
  Therefore $W \subseteq D^k_{\ni a}$, and thus $\Norm{\mu - \mu_{a\to b}}_{\ell_1} = 2\Normtv{\mu -
  \mu_{a\to b}} = 2 \mu(W) \leq 2 \mu\Paren{D^k_{\ni a}}$.
\end{proof}

\begin{lemma} \label{lem:proj-replace}
  Given any distribution $\mu$ over $D^k$, any $a,b \in D$, any $Q \subseteq [k]$, $\pi_Q (\mu_{a \to b}) = (\pi_Q
  \mu)_{a \to b}$.
\end{lemma}

\begin{proof}
  Let $\overline Q \coloneqq [k]\setminus Q$.
  For any $x \in D^Q$,
  \begin{align*}
    \pi_Q (\mu_{a \to b}) (x)
    &= \sum_{y \in D^{\overline Q}} \mu_{a\to b}(x, y) = \sum_{\substack{c \in D^Q,y,d \in D^{\overline Q}\\ (c,d)_{a\to
    b} = (x,y)}} \mu(c,d) \\
    &= \sum_{c \in D^Q,c_{a \to b} = x} \sum_{d\in D^{\overline Q}} \mu(c,d)
    = \sum_{c \in D^Q,c_{a \to b} = x} (\pi_Q \mu)(c) = (\pi_Q \mu)_{a \to b}(x) \,. \qedhere
  \end{align*}
\end{proof}

\begin{lemma} \label{lem:power-replace}
  Given any distribution $\nu$ over $D$, any $a,b \in D$, any finite set $Q$, $(\nu^Q)_{a \to b} = (\nu_{a \to b})^Q$.
\end{lemma}

\begin{proof}
  For any $x \in D^Q$,
  \begin{align*}
    (\nu^Q)_{a \to b} (x) &= \sum_{y \in D^Q, y_{a \to b} = x} \nu^Q(y)
    = \sum_{y \in D^Q, y_{a \to b} = x} \prod_{i \in Q} \nu(y_i)
    = \prod_{i \in Q} \sum_{z \in D, z_{a \to b} = x_i} \nu(z) \\
    &= \prod_{i \in Q} \nu_{a \to b}(x_i) = (\nu_{a \to b})^Q (x) \,. \qedhere
  \end{align*}
\end{proof}

\begin{lemma} \label{lem:close-min-prob}
  For any $0 < \alpha \leq 1/2$, any distribution $\nu$ over a finite domain $D$, there is a distribution
  $\eta$ over $D$ of minimum nonzero probability at least $\alpha$, such that for any $1 \leq t \leq k$, any
  $(\delta,t)$-wise $\nu$-independent distribution $\mu$ over $D^k$ is $2\abs{D}(k\alpha+\delta)$-close to
  some $(\delta,t)$-wise $\eta$-independent distribution $\zeta$.
\end{lemma}

\begin{proof}
  Suppose there is $a \in D$ such that $0 < \nu(a) < \alpha$.
  Since $\nu$ has total probability mass $1$, there is $b \in \supp(\nu) \setminus \set a$.
  Then for any $Q \subseteq [k]$,
  \[
    \Norm{\pi_Q (\mu_{a \to b}) - (\nu_{a \to b})^Q}_{\ell_1}
    = \Norm{(\pi_Q \mu)_{a \to b} - (\nu^Q)_{a \to b}}_{\ell_1}
    \leq \Norm{\pi_Q \mu - \nu^Q}_{\ell_1}
  \]
  where the equality is \cref{lem:proj-replace,lem:power-replace}, and the inequality is \cref{lem:pushforward-dist}.
  This implies
  \[ \delta^{[1,t]}_{\nu_{a \to b}} [\mu_{a \to b}] \leq \delta^{[1,t]}_\nu [\mu]\,. \]
  Since $\mu$ is $(\delta,t)$-wise $\nu$-independent, $\mu_{a \to b}$ is $(\delta,t)$-wise $\nu_{a \to
  b}$-independent.
  Also \cref{lem:dist-replace} implies
  \begin{align*}
    \Norm{\mu_{a \to b} - \mu}_{\ell_1}
    &\leq 2\mu\Paren{D^k_{\ni a}} \leq 2 \sum_{i \in [k]} (\pi_{\set i} \mu)(a) \leq 2\sum_{i \in [k]}
    \Paren{\nu(a) + \Norm{\pi_{\set i} \mu - \nu}_{\ell_1}} \\
    &< 2k \alpha + 2\delta^{[1,t]}_\nu [\mu] \leq 2k\alpha + 2\delta.
  \end{align*}

  In other words, as long as $0 < \nu(a) < \alpha$ for some $a \in D$, pick any $b \in \supp(\nu) \setminus \set a$,
  transform $\mu' \coloneqq \mu_{a \to b}$.
  Then $\norm{\mu' - \mu}_{\ell_1} < 2k\alpha + 2\delta$, and $\mu'$ is $(\delta,t)$-wise $\nu'$-independent,
  where $\nu' \coloneqq \nu_{a \to b}$.

  Keep shifting probability mass of $\nu$ as above until the minimum nonzero probability of $\nu$ is at
  least $\alpha$.
  Denote by $\zeta$ the final $\mu$ and by $\eta$ the final $\nu$ after all the transformations.
  Since no more than $\abs{D}$ many $a$'s can satisfy $0 < \nu(a) < \alpha$ initially, there are at most
  $\abs{D}$ such
  transformations.
  Therefore $\Norm{\mu - \zeta}_{\ell_1} < 2\abs{D}(k\alpha + \delta)$, and $\zeta$ is $(\delta,t)$-wise
  $\eta$-independent, so that $\eta$ has minimum nonzero probability at least $\alpha$.
\end{proof}

\subsection{Refutation via Independence}\label{sec:refute-via-independence}

In this section we prove the following statement about refutations under independence.
Recall $t$-wise independent value of a CSP from \cref{def:t-wise-value}.

\begin{theorem}\label{thm:upper-bound}
  For any $k$-CSP $(D, \cR)$, any $2 \leq t \leq k$, any $\eps > 0$,
  there are $K = K(k, \abs{D}, \eps), C = C(k,\abs{D},t,\eps), \delta > 0$ such that whenever
  \[ m \geq
    \begin{cases}
      C n^{t/2} & \text{if $t$ is even and $\ell = t/2$} \\
      \displaystyle C n^{t/2} \frac{\ell}{\ell^{t/2}} \log n &
      \text{if $t$ is even and $\displaystyle \frac t2 \leq \ell \leq \delta n$} \\
      \displaystyle C n^{t/2} \frac{\ell}{\ell^{t/2}} \sqrt{\log n} & \text{if $t$ is odd and $t-1 \leq \ell \leq
      \displaystyle \frac{\delta n}{\log n}$}
    \end{cases},
  \]
  except with probability $o_n (1)$ over a binomial $\rho$-random instance $\cI$ with
  constraint density $p$, every pseudo-distribution $\mu$ of degree $K+2(\ell+1)$ consistent with $\cP(V,D)$ satisfies
  \[ \pE_\mu \Val_\cI(x) \leq \opt_t(\rho) + \eps. \]
\end{theorem}

\cref{thm:upper-bound} implies \cref{thm:main-upper-bound,thm:main-upper-bound-higher}.
\cref{thm:upper-bound} upper bounds the primal SoS value to be nearly $\opt_t(\rho)$.
Since the SoS program includes the constraints $\Set{x^2_{va} \leq 1}_{v \in V, a \in D}$, there is no duality gap
\cite[Theorem~1]{josz2016strong}.
Therefore the theorem indirectly yields a sum-of-squares upper bound.

Throughout the section let  $2\leq k\leq n$, let $(D,\cR)$ be a $k$-CSP and let  $\rho$ be a distribution over  $\cR$.
For an instance $\cI=(V,\cC)$ and a relation $R\in \cR$  define $\cC(R)$ to be the multi-set of injective functions such
that $(\sigma,R)\in \cC.$
In other words, $\cC(R)$ consists of the scopes of constraints of type $R$ in $\cI$.
Note that by construction there exists a bijection between elements in $\cC(R)$ and constraints in $\cC$ containing the
relation $R.$

\begin{definition}
  Given a collection $\cC$ of $k$-ary constraints from an instance and an assignment $x \in D^V$, the \emph{induced
  distribution}, denoted $\cD_{\cC,x}$ is the probability distribution on $D^k$ defined by
  \[ \cD_{\cC,x}(a) \coloneqq \Pr_{(\sigma,R) \in \cC} [x_\sigma = a] \qquad \forall a \in D^k\,. \]
\end{definition}

In the above definition, the $k$-ary relation $R$ of the constraint is irrelevant.
\cite[Definition~3.2]{AOW15} also considered an induced distribution in their sum-of-squares
refutations, but their definition differs and only works for CSPs given by a predicate with uniformly random literals.


Also given a $k$-ary constraint multiset $\cC$ and a pesudo-distribution $\mu$ of degree at least $k$, let
\[ \cD_{\cC,\mu} \coloneqq \pE_{\mu} \cD_{\cC,x} = \E_{(\sigma,R) \in \cC} \mu_\sigma \]
be the induced distribution under $\mu$.
Here $\mu_\sigma$ denotes the distribution of $\mu$ on the $k$-tuple $\sigma$ of variables.

The following lemma will be proved in \cref{sec:deviation-bound}.

\begin{restatable}{lemma}{concentrated} \label{lem:concentrated}
  For any $2 \leq t \leq k$, any finite domain $D$, any $\eps > 0$, there is $C = (k, \abs{D}, t, \eps) > 0$ and
  $\delta > 0$ such that whenever
  \[ m \geq
    \begin{cases}
      C n^{t/2} & \text{if $t$ is even and $\ell = t/2$} \\
      \displaystyle C n^{t/2} \frac{\ell}{\ell^{t/2}} \log n &
      \text{if $t$ is even and $\displaystyle \frac t2 \leq \ell \leq \delta n$} \\
      \displaystyle C n^{t/2} \frac{\ell}{\ell^{t/2}} \sqrt{\log n} & \text{if $t$ is odd and $t-1 \leq \ell \leq
      \displaystyle \frac{\delta n}{\log n}$}
    \end{cases},
  \]
  except with probability $o_n (1)$, a directed $k$-uniform hypergraph $\cH$ from \cref{def:random-instance} having $n$
  vertices and $m$ expected directed hyperedges is $(t,\eps,2(\ell+1))$-concentrated.
\end{restatable}

We are now ready to prove \cref{thm:upper-bound}.

\begin{proof}
  Let $\delta = \eps/(2^{\abs{D}^k} 3)$ and $\gamma = \eps^2 / (9H({k \choose \leq t} 2\abs{D}^t)^2)$ for some
  $H = H(t,k,\abs{D}) > 0$ defined below.
  Call $R \in \cR$ $\delta$-heavy if $\rho(R) \geq \delta$.
  \cref{lem:concentrated} and a union bound imply that with probability $1-o_n (1)$ over a random
  instance $\cI$, every subinstance $(V, \cC(R))$ is $(t,\gamma,2t)$-concentrated for each $\delta$-heavy $R \in \cR$.
  \cref{thm:t-wise-conditioning} implies that there are $K > 0$ and a distribution $\cT$ over ${V \choose
  \leq K}$ such
  that for every pseudo-distribution $\mu$ of degree $K+2t$ consistent with $\cP(V,D)$,
  \[
    \E_{T,\alpha} \Norm{\pi_Q \cD_{\cC(R),\mu\vert\alpha} - \Paren{\E_{v \in V} (\mu|\alpha)_v}^t }_{\ell_1}
    \leq 2\abs{D}^t \gamma
  \]
  for every $Q \in {[k] \choose t}$, $\delta$-heavy $R \in \cR$.
  Here the outer expectation is over $T \sim \cT, \alpha \sim \alpha(T,\mu)$.
  Let $\nu_\alpha \coloneqq \E_{v \in V} (\mu|\alpha)_v$.
  \cref{lem:pushforward-dist} applied to the projection $(b \in D^Q \mapsto b_P)$ implies that for any
  $P \subseteq Q \in {[k] \choose t}$,
  \[
    \Norm{\pi_P \cD_{\cC(R),\mu\vert\alpha} - (\nu_\alpha)^P }_{\ell_1} \leq
    \Norm{\pi_Q \cD_{\cC(R),\mu\vert\alpha} - (\nu_\alpha)^Q }_{\ell_1}\,.
  \]
  Summing the previous two inequalities over $1 \leq \abs{P} \leq t$, we have for every $\delta$-heavy $R \in \cR$,
  \begin{equation} \label{eq:small-subset-cond-distance}
    \E_{T,\alpha} \delta^{[1,t]}_{\nu_\alpha}\Brac{\cD_{\cC(R),\mu\vert\alpha}} \leq \textstyle{k \choose \leq t}
    2\abs{D}^t\gamma \,.
  \end{equation}
  \cref{thm:almost-twise} implies there is $H = H(t, k, \abs{D}) > 0$ such that every (not necessarily $\delta$-heavy)
  $R \in \cR$ has a $t$-wise $\eta^\alpha$-independent distribution $\zeta_{R,\alpha}$ satisfying
  \begin{equation} \label{eq:close-t-wise-indep}
    \Norm{\cD_{\cC(R),\mu\vert\alpha} -
    \zeta_{R,\alpha}}_{\ell_1} \leq H \sqrt{\delta^{[1,t]}_{\nu_\alpha} \Brac{\cD_{\cC(R),\mu\vert\alpha}}}.
  \end{equation}
  Further, $\eta^\alpha$ depends only on $\nu_\alpha$ (and not on $R$).
  In particular for every $\delta$-heavy $R \in \cR$,
  \begin{align}
    \E_{T,\alpha} \Norm{\cD_{\cC(R),\mu\vert\alpha} - \zeta_{R,\alpha}}_{\ell_1}
    &\overset{\eqref{eq:close-t-wise-indep}}\leq \E_{T,\alpha}
    H \sqrt{\delta^{[1,t]}_{\nu_\alpha} \Brac{\cD_{\cC(R),\mu\vert\alpha}}}
    \overset{(*)}\leq H\sqrt{\E_{T,\alpha} \delta^{[1,t]}_{\nu_\alpha} \Brac{\cD_{\cC(R),\mu\vert\alpha}}} \nonumber\\
    &\overset{\eqref{eq:small-subset-cond-distance}}\leq H\sqrt{\textstyle{k \choose \leq t}2\abs{D}^t\gamma}
    \label{eq:cond-close-t-wise-indep},
  \end{align}
  where $(*)$ is Cauchy--Schwarz.

  Denote by $\tilde\rho$ the empirical distribution of relations $R$ from $\cC$:
  \[ \tilde\rho(R') \coloneqq \Pr_{(\sigma,R) \in \cC} [R = R']\, \qquad \forall R' \in \mathcal R. \]
  Then
  \begin{align*}
    \pE_\mu \val_{\cI} (x)
    &= \E_{(\sigma,R) \sim \cC} \pE_{x \sim \mu} \bracbb{x_\sigma \in R}
    = \E_{R \sim \tilde\rho} \E_{\sigma \sim \cC(R)} \pE_\mu \bracbb{x_\sigma \in R} \\
    &\leq \norm{\tilde\rho - \rho}_{\ell_1} + \E_{R \sim \rho} \E_{\sigma \sim \cC(R)} \pE_\mu \bracbb{x_\sigma \in R}
  \end{align*}
  because $\Abs{\E_{\sigma \sim \cC(R)} \pE_\mu \bracbb{x_\sigma \in R}} \leq 1$ for every $R \in \cR$.
  To bound the first term, Chernoff and union bounds imply $\norm{\tilde\rho - \rho}_{\ell_1} \leq \gamma/3$ with
  probability $1-o_n(1)$.
  To bound the second term,
  \begin{align*}
    \E_{R \sim \rho} \E_{\sigma \sim \cC(R)} \pE_\mu \bracbb{x_\sigma \in R}
    &= \E_{R \sim \rho} \E_{\sigma \sim \cC(R)} \E_{T,\alpha} \pE_{\mu\vert\alpha} \bracbb{x_\sigma \in R}
    = \E_{R \sim \rho} \E_{T,\alpha} \pE_{y \sim \cD_{\cC(R),\mu\vert\alpha}} \bracbb{y \in R}\,, \\
    &\leq \E_{R \sim \rho} \E_{T,\alpha} \Paren{\pE_{y \sim \zeta_{R,\alpha}} \bracbb{y \in R}
    + \Normtv{\cD_{\cC(R),\mu\vert\alpha} - \zeta_{R,\alpha}}} \\
    &\leq \opt_t(\rho) + \E_{R \sim \rho} \E_{T,\alpha} \Normtv{\cD_{\cC(R),\mu\vert\alpha} -
    \zeta_{R,\alpha}}
  \end{align*}
  because $\Paren{\zeta_{R,\alpha}}_{R \in \cR}$ is $t$-wise $\eta^\alpha$-independent for every $\alpha$.
  Finally, let $\cR' \subseteq \cR$ be the set of $\delta$-heavy $R \in \cR$.
  Then
  \begin{align*}
    & \E_{R \sim \rho} \E_{T,\alpha} \Normtv{\cD_{\cC(R),\mu\vert\alpha} - \zeta_{R,\alpha}} \\
    \leq & \Pr_{R \sim \rho}[R \in \cR'] \E_{R \sim \rho} \Brac{ \E_{T,\alpha}
    \Normtv{\cD_{\cC(R),\mu\vert\alpha} - \zeta_{R,\alpha}} \;\middle|\; R \in \cR'}
    + \Pr_{R \sim \rho}[R \not\in \cR'].
  \end{align*}
  The first term is at most $H\sqrt{{k \choose \leq t}2\abs{D}^t\gamma} = \eps/3$ by
  \cref{eq:cond-close-t-wise-indep}.
  The second term equals $\rho(\cR \setminus \cR') \leq \eps/3$ since there are at most $2^{\abs{D}^k}$ many
  $R \in \cR$ that are not $\delta$-heavy.
\end{proof}

\begin{theorem} \label{thm:upper-bound-replacement}
  For any $k$-CSP $(D, \cR)$, any $2 \leq t \leq k$ such that $t$ is even, any $\gamma > 0$,
  there are $K = K(k, \abs{D}, \gamma)$ and $C = C(k,\abs{D},t,\gamma)$ such that for any
  $m \geq C n^{t/2}$, except with probability $o_n (1)$ over a $\rho$-random instance $\cI$ having $n$ variables with
  $m$ independent constraints with replacement, every pseudo-distribution $\mu$ of degree $K$ consistent
  with $\cP(V,D)$
  satisfies
  \[ \pE_\mu \Val_\cI(x) \leq \opt_t(\rho) + \gamma. \]
\end{theorem}

\begin{proof}
  Let $\delta$ and $\eps$ be defined as in the proof of \cref{thm:upper-bound}.
  Denote by $\tilde\rho$ the empirical distribution of relations $R$ from $\cC$:
  \[ \tilde\rho(R') \coloneqq \Pr_{(\sigma,R) \in \cC} [R = R']\, \qquad \forall R' \in \mathcal R. \]
  Call $R \in \cR$ $\delta$-heavy if $\tilde\rho(R) \geq \delta$; note that the definition of $\delta$-heavy differs
  from the one in the proof of \cref{thm:upper-bound}.
  Conditioned on $\tilde\rho$, \cref{thm:concentration-random-instance-replacement} and a union bound imply that with
  probability $1-o_n(1)$ over a random instance $\cI$, every subinstance $(V, \cC(R))$ is
  $(t,\eps,2t)$-concentrated for
  each $\delta$-heavy $R \in \cR$.
  The rest of the proof is identical to that of \cref{thm:upper-bound}.
\end{proof}

\begin{remark}
  Our SoS refutation uses the almost $t$-wise independence result (\cref{thm:almost-twise}).
  Instead of \cref{thm:almost-twise}, it is also possible to use dual dominating polynomial for spectral refutation in
  subsequent sections.
  However, the $\eps$ in dual dominating polynomial has a much worse dependence on $k$, being doubly exponential in $k$
  (\cref{lem:lp-bounded-dual-norm}).
  For some applications where $k$ is a slowly growing function of $n$ (not covered in this paper),
  using \cref{thm:almost-twise} yields better SoS upper bounds.
\end{remark}

\section{Spectral Refutation from Partial Certification} \label{sec: partial-certification-sect}

In this section, we establish our main result for spectral refutation based on partial certification, and our main goal of this section is to outline that these certification results can be reduced to matrix norm bounds formally proven in the subsequent section.

\subsection{Our Results for Spectral Refutation}

For simplicity, we showcase our primary result for a single predicate (with fixed literals) restricted to boolean predicates. This is also the cornerstone for the extensions through Kikuchi hierarchy and the application to general domains.

\paragraph{Basic Results for Spectral Refutation}
\begin{restatable}{theorem}{SpectralRefutation}
\label{thm:spectral-refutation}(Polynomial-Time Refutation for Boolean Predicate (Special Case of \cref{thm: alphabet-ref}) )
	For a random instance $\calH$  of Boolean $k$-CSP predicates $P$ on a random hypergraph $\calH$ from~\cref{def:random-instance},  for any $\eps>0$ and $t > 1\in \N$ ,  we have a spectral refutation that runs in time $\tilde{O}_k(n^{2t})$ to certify that the normalized objective value is at most \[
	\max_{x\in \{\pm 1\}^n} \Psi_P( x) \leq  \max_{\substack{\mu:  \text{$t$-wise independent} } } \opt_t(P) +\eps +o_n(1)
	 \,,\]
	where $\opt_t(P)$  is the predicate-value achieved by the optimal $t$-wise $\nu$-independent distribution, 
	  provided    $m \geq \Omega_{k}( \frac{1}{\epsilon^2}\cdot    n^{t/2}  \log n)$. \end{restatable}

We emphasize that, when specialized to Boolean CSPs, our result avoids any $n^{1/\poly(\epsilon)}$ dependence in the runtime, improving upon prior polynomial-time SoS refutation algorithms based on global correlation rounding.

\paragraph{Optimal Three-Way Tradeoff via the Kikuchi Hierarchy} 
 
 Moreover, the simplicity of the spectral algorithm from above gives us a convenient way to go beyond the polynomial time threshold of $O(n^{t/2})$. We show that the same family of algorithms, augmented by higher-level Kikuchi hierarchy, achieves a tight, smooth tradeoff in subexponential time that matches the three-way tradeoff provided by \cite{AOW15, RRS17, KMOW17}. 

\begin{restatable}{theorem}{SubexpRefutation}
\label{thm:smooth-spectral-refutation}(Refutation for Boolean Predicate in Subexponential Time)
	For a random instance $\Psi$ of a single Boolean predicate $P$ (as defined in \cref{def:random-instance}) , for any $\eps>0$, for any $\ell <O(n)$, there is an algorithm that runs in time $O_{k}( n^{O(\ell)})$ to certify that the normalized objective value \[
	\max_{x\in \{\pm 1\}^n} \Val_\Psi(x) \leq   \opt_t(P) + \eps	+o_n(1) \,,\]
	  provided   \[m \geq \Omega_k (\frac{1}{\epsilon^2}\cdot   n^{t/2}  \cdot \frac{\ell }{\ell^{t/2}} \cdot   \log n) \,,\]  or more tightly for odd $k$ with $\ell =O(\frac{n }{\log n})$ ,  \[ 
m \geq \Omega_k (\frac{1}{\epsilon^2}\cdot   n^{t/2}  \cdot \frac{\ell }{\ell^{t/2}} \cdot   \sqrt{\log n})	  \,.\]  
\end{restatable}

\begin{remark}
	We make no effort to optimize our dependence on $k$, $t$, and later on $|D|$ and $|\calP|$. Our bounds throughout this work hide any absolute-constant dependence on these parameters.
	\end{remark}
	Our techniques also extend naturally to instances involving multiple
predicates as well as to non-boolean CSPs with essentially the same tradeoff, up to an absolute-constant-factor dependence on the domain size.


\LargeAlphabetRefutation*
 \begin{remark}
	When specialized to the study of random CSPs with random signings, our result recovers the best known dependence of clause density in $ \ell$ and $\log n$, and gives improved dependence in $\epsilon$ for odd-$t$ that matches the dependence for even-$t$ .
	\end{remark}

\paragraph{Roadmap of this Section}
We begin with an overview of partial certification in the next subsection, highlighting the main technical ideas in the Boolean domain. This overview emphasizes that all of our results ultimately rely on the spectral norm bounds established later in \cref{sec:norm-bound}. 

We then give a summary for the single-predicate Boolean case in \cref{sec:boolean-summary} with its full algorithm, followed by the extension to multiple predicates. Finally, we present the proof of our main result- \cref{thm: alphabet-ref}- in \cref{proof:main-spectral}. 

\subsection{Overview of Partial Certification}
As described in the technical overview, \emph{partial certification} is the key ingredient in our spectral refutation algorithm. It shows that, for a polynomial arising from a random instance, one can efficiently certify that its value concentrates around its mean over the biased hypercube. In this overview section, we will focus on the Boolean domain to illustrate our main insights.

To establish concentration for general polynomials, we first introduce a lemma that focuses on \emph{monomials with random support}, and show that certifiable concentration for such monomials readily extends to certifiable concentration for arbitrary polynomials.

\begin{lemma}[Partial Certification for Monomials] \label{lem:certification-monomial}
	For any $\ell =O(n)$ and $\eps>0$, for any subset $S$ of weight at most $t$, let $\Psi_S(x)$ be the global polynomial induced by the monomial-restriction $S$ as in~\cref{def:monomial-global}, 
	 there is a spectral certification that runs in time $n^{O(\ell)}$ for
	  \[ 
\Psi_S(x)  \in \left[ \left(\frac{\langle \1_n, x\rangle}{n} \right)^{|S|} \pm (\eps +o_n(1))  \right] \,,
\]
 for all $x\in \{\pm 1\}^n$ provided  $m$ satisfies the density constraint in  \cref{thm:smooth-spectral-refutation}.
 

%
\end{lemma}

Before we prove this main lemma, we illustrate how this comes in handy for establishing partial certification of the dual polynomial $\Psi_{Q_\nu}$, that ultimately serves as an upper bound for a CSP predicate.

\begin{restatable}{lemma}{FormalPartialCertification} \label{lem:formal-partial-certification}
(Partial Certification in the $\nu$-Biased Hypercube [Formal Statement of \cref{lem:partial-certification}] ) For any $\eps>0$ and any $\ell \leq O(n)$,  	 there is an algorithm that runs in time $n^{O(\ell)}$ certifies \[ 
\max_{\substack{x\in\{\pm 1\}^n \\ \E_{i\in [n]} [x_i] = \nu } } \Psi_{Q_\nu}(x) \leq \sum_{S\subseteq  [k], |S|\leq t} c_S\cdot \nu^{|S|}  + o_n(1) = \val_{t}(Q_\nu) + \eps +o_n(1)\,.  \]
provided the instance is sufficiently dense above the anticipated spectral threshold (see the formal condition in \cref{thm:smooth-spectral-refutation})  . 
\end{restatable}
%

\begin{proof}[Proof of \cref{lem:formal-partial-certification} ]

Expanding the given polynomial in the monomial basis, we have
\begin{align*}
	  \Psi_{Q_\nu}(x)  &=  \frac{1}{|E(\calH)| } \cdot \sum_{e\in E(\calH)} Q_\nu (x_{\tau(e)})  
	  =  \frac{1}{|E(\calH)| }  \left(  \sum_{e\in E(\calH)}   \sum_{\substack{S\\ \text{monomial-restriction}\\ |S|\leq t}}  c_{\nu, S}   \cdot x_{\tau(e)_S} \right)
	  \\&=
	   \sum_{\substack{S\\ \text{monomial-restriction}\\ |S|\leq t}}  c_{\nu,S}\cdot     \left( \frac{1} {|E(\calH)| } \cdot   \sum_{e\in E(\calH)}  x_{\tau(e)_S} \right) = \sum_{\substack{S\\ \text{monomial-restriction}\\ |S|\leq t}} c_{\nu,S} \cdot \Psi_S(x)
\end{align*}

Therefore, it is convenient to note that as long as we can certify for all monomal-restrictions $S\in \binom{k}{t} $,  \[ 
\Psi_S(x)  \in \left[ \left(\frac{\langle \1_n, x\rangle}{n} \right)^{|S|} \pm (\eps'+ o_n(1))  \right] 
\]for any $x\in \{\pm 1\}^n$, we would then certify \begin{align*}
\max_{x: \langle \1_n, x\rangle =  \nu \cdot n}  \Psi_{Q_\nu}(x)&  \leq \sum_{S: |S|\leq t} c_{\nu,S} \cdot \nu^{|S|} +  \|c_\nu\|_1 (\eps'+   o_n(1)) 	\\&
=\sum_{S: |S|\leq t} c_{\nu,S} \cdot \nu^{|S|}  + \eps+ o_n(1)
\end{align*}
 as desired 
 for $m$ satisfying the prescribed density constraint from 
 \jnote{Tommaso, i think its ok to keep the "prescribed density" here cuz otherwise it might be too long?...}
when picking $\eps' = \frac{\eps }{|c_\nu|_1}$ and using a bound on the size of the coefficients  $\|c_\nu\|_1 \coloneqq \sum_{S} |c_{\nu,S}| \leq O_{k}(1)$ by \cref{lem:lp-bounded-dual-norm}. 
%

Finally, noting that $\sum_{S: |S|\leq t} c_{\nu,S} \cdot \nu^{|S|} = \val_{t}(Q_\nu) $ is precisely the objective value in our dual LP, which by strong duality of the LP, is the optimal value $\val_{P,t}(\nu)$ by the $t$-wise $\nu$-independent distribution.
\end{proof}	

That said, it remains to complete the proof of \cref{lem:certification-monomial}. Towards that end, we introduce two additional claims.

\begin{claim} \label{claim:edge-concentration}
 Let $\calH$ be the random instance from ~\cref{def:random-instance},  recall that $|E(\calH)|$ is the number of  clauses in an instance $\calH$,
	With high probability,\[ 
	\frac{m}{|E(\calH)| } = 1+o_n(1)
	\]
\end{claim}
This follows by standard concentration of random variables provided $m=\Omega(n)$ satisfied by our regime throughout this work. Next, we introduce the following claim that bounds the deviation terms.
\begin{restatable}{claim}{KikuchiClaim}
\label{claim:smooth-certification}(Smooth Certification for Deviation via Matrix Norm Bounds) For any $\epsilon>0$, and  any $\ell \leq C n$ for some constant $C>0$, assuming $m\geq \frac{n}{\epsilon^2}$,	 ~\cref{alg:single-predicate-certification} runs in time $n^{O(\ell)}$  to certify 	\[
\max_{x\in \{\pm 1\}^n } \left|   ( \Psi_{S}-  \widetilde{\Psi}_{S})( x) \right| < \epsilon + o_n(1)
 \]
with high probability for random instance $\calH$ provided  $m$ satisfies the density constraint in \cref{thm:smooth-spectral-refutation}.
\end{restatable}
We are now ready to wrap up the proof to \cref{lem:certification-monomial}.
\begin{proof}[Proof to \cref{lem:certification-monomial}]
Recall that $|E(\calH)|$ is the number of edges in the hypergraph sample.  To get started, we center $\Psi_{S}(x)$ around its mean $\widetilde{\Psi}_{S}(x)$ by rewriting the following, \[ \Psi_{S}(x) = \widetilde{\Psi}_{S}(x) +  \underbrace{\left( \Psi_{S}(x)  -  \widetilde{\Psi}_{S}(x) \right)}_{\textsf{Deviation}_S(x) \coloneqq } \]
where the centered polynomial is defined as 
\[   \widetilde{\Psi}_{S}( x) \coloneqq  \frac{1}{|E(\calH)|} \cdot  \sum_{T \in  [n]^{|S|}}( \sum_{\substack{ R \in [n]^k \\ R_S = T }} p)  \cdot x_T  = \frac{1}{|E(\calH)|}\sum_{T\in [n]_{|S|} } \E[D_S(T)] \cdot x_T  \]
where $D_{S}(T)$ denotes the $S$-degree of $T$ -- the number of predicates whose $S$-restriction is $T$. Since each tuple has the same expected degree $\frac{m }{n^{|S|}}  $ as each clause exists i.i.d. with probability $p = \frac{m }{n^{|S|}} $, the above simplifies as
\begin{align*}	
 \widetilde{\Psi}_{S}(x) =  \frac{1}{|E(\calH)|} \cdot \frac{m}{n^{|S|}}  \cdot \sum_{T\subseteq [n]^{|S|}} x_T  &= \left( \frac{m }{|E(\calH)|}\right)  \cdot  \left(\frac{\langle 1_n, x\rangle}{n}\right)^{|S|}\,.
\end{align*}


Finally, the main lemma is then immediate by the following two claims, as for all $x \in \{\pm 1\}^n $, we have  \begin{align*}
	\Psi_S(x)  &= \tilde{\Psi}_S(x) + \textsf{Deviation}_S(x)\\
	&= \left( \frac{m }{|E(\calH) | }\right)  \cdot  \left(\frac{\langle 1_n, x\rangle}{n}\right)^{|S|} + \textsf{Deviation}_S(x) \\
		&=    \left(\frac{\langle 1_n, x\rangle}{n}\right)^{|S|} +  \eps +o_n(1) 	\end{align*}
where we plug in ~\cref{claim:edge-concentration} and \cref{claim:smooth-certification}.
\end{proof}
\begin{claim} \label{claim:edge-concentration}
Recall that $|E(\calH)|$ is the number of  clauses in an instance $\calH$.
	With high probability,\[ 
	\frac{m}{|E(\calH)| } = 1+o_n(1)
	\]
\end{claim}



We dedicate the next subsection to the proof of this claim. In particular, we show next that this boils down to controlling the spectral norm of certain random matrix, and subsequently establish the desired spectral norm bounds in \cref{sec:norm-bound}.

\subsection{Deviation Certificates via Spectral Norm Bounds}

In this subsection, we reduce partial certification to a spectral norm bound for a centered
random tensor, and prove~\cref{claim:smooth-certification}.
Concretely, for each monomial restriction $S \subseteq [k]$, we show that the
deviation
$
|\Psi_S(x) - \widehat{\Psi}_S(x)|
$
is uniformly small over all assignments $x \in \{\pm1\}^n$.  This deviation is then controlled by reducing the injective norm to a spectral norm bound via matrix flattening. As usual, the construction of
this matrix differs slightly depending on whether $|S|$ is even or odd.

	 
It is more convenient to work with the unscaled objective value, and	
consider $C_S(x) \coloneqq |E(\calH)| \cdot   (\Psi_S - \widehat{\Psi})(x) $ where we ignore the normalization factor in $\Psi$. Hence the target bound here would be $o(m)$ for the bound to be non-trivial.  
 When viewed as a polynomial, let $C_S[T]$ denote the coefficient of $x_T$. We can see
 \[ 
	C_S[T] = \sum_{\substack{ R \in  [n]^k \\  R_S = T  } } \left(\underbrace{  \1[ R \in E(\calH)] - p}_{ \coloneqq G_R } \right)\,,	
\] 
itself as a polynomial in the underlying (centered) random variables $\{G_R\}_{R\in [n]^k}$. 
As a sanity check, for the case $|S|= k$, the summation is merely a singleton term, and this is simply the normalized adjacency-tensor for hypergraph $\calH$.
		
From this point onwards, we will view $C_S$ as an order-$|S|$ tensor.
Since $C_S(x)$ is a polynomial of degree $|S|$, it is natural to
represent it by its coefficient tensor indexed by $|S|$-tuples of vertices.
	 
	 For concreteness, we showcase the derivation for the basic level of Kikuchi hierarchy, i.e. $\ell = |S|/2$ for even $|S|$, and $\ell =|S|-1$ for odd $S$ due to the Cauchy-Schwarz trick. We illustrate the full derivation for higher-level of Kikuchi matrices towards the end of this section. 
	 
	
\paragraph{Deviation  for even $S$} \label{def:basic-even-S-matrix} 	We start with the simpler case of even $S$. We now define a matrix whose quadratic form exactly recovers the polynomial $C_S(x)$ (up to  some rescaling factor),  \begin{definition}
		[$M_S$ for even-$S$]
		For even-$S$, let $M_S$ be a matrix with rows and columns indexed by all subsets of size $|S|/2$ and have non-zero entries $I,J \in \binom{n}{ |S|/2}$ given by 
	\[ M_S[I,J] \coloneqq   C_S[I\cup J]\,,\]
	\end{definition}
%
Recall that we define $x^{\circledcirc r}$ is the vector indexed by subset of size $r$, and $x^{\circledcirc r}[T] = \prod_{i\in T} x_i$, we have the following observation.
\begin{definition}[$x^{\circledcirc r}$ vectors]
	For any $x\in \{\pm 1\}^n$, and $r\in N$, $x^{\circledcirc r}$ is the $\binom{n}{r}$-dimensional vector s.t.  for any subset $S\subseteq \binom{n}{r}$, $x^{\circledcirc r}[S] \coloneqq \prod_{i\in S} x_i$.
\end{definition}

\begin{observation} \label{obs:basic-even-equivalence}
     For any $x\in \{\pm 1\}^n $,
    \[ 
    (x^{\circledcirc |S|/2})^T M_S \cdot  x^{\circledcirc |S|/2} = C_S(x)  \cdot \binom{|S|}{|S|/2} \,.
    \]

\end{observation}
Once the matrix is defined, we can then apply spectral norm bound of $\|M_S\|$ to obtain an upper bound for $|C_S(x)|$. We defer the formal verification to after the introduction of the corresponding matrix for the odd case.
	
\snote{I removed deviation for odd $S$ for basic matrix, as discussed in meeting}

\paragraph{Wrapping Up for the Basic Cases}

We now prove~\cref{claim:smooth-certification} for the basic matrices, corresponding specifically to $\ell=O(|S|)$. 
\begin{proof}[Proof to the special case of $\ell=O(|S|)$ in \cref{claim:smooth-certification}]	
We split the proof into odd and even as well. For the even case, the matrix defined in \cref{def:basic-even-S-matrix} gives \[ 
 (x^{\circledcirc |S|/2})^T M_S \cdot  x^{\circledcirc |S|/2} = C_S(x)  \cdot \binom{|S|}{|S|/2}
\]
via~\cref{obs:basic-even-equivalence}. Next, the deviation bound for even-$S$ follow immediately follow from a spectral norm bound of $M_S$ as we have
\[ 
\binom{|S|}{|S|/2} \cdot  \max_{x\in \{\pm 1\}^n} C_S(x)  =   \max_{x\in \{\pm 1\}^n } (x^{\circledcirc |S|/2})^T M_S \cdot  x^{\circledcirc |S|/2} \leq \|M_S\|_{sp} \cdot \binom{n}{|S| /2 } \,.
\]
The lower tail follows analogously by noting $\lambda_n \geq - \|M_S\|_{sp}$. 
 Plugging in the norm bounds from~\cref{lemma:spec-norm-bounds}, which state that w.h.p., \[ 
\|M_S\|_{sp} \leq O(\sqrt{\frac{m}{n^{|S|/2}}}  \cdot \sqrt{\log n})\,,
\]
 we have 
\[
\max_{x} | C_S(x)| \leq O_{|S|}(1) \cdot \sqrt{\frac{m}{n^{|S|/2} }  } \cdot \sqrt{\log n} \cdot n^{|S|} =  o(m)\,,	 
\]
subject to the density constraint from~\cref{lemma:spec-norm-bounds} of
\[ m \geq \Omega_{|S|}(n^{|S|/2} \cdot \log n )\,. \qedhere \]
\end{proof}
\snote{removed the proof for the basic odd case}

\paragraph{Kikuchi Extension for Smoothed Certification }
Before delving into the proof of our spectral norm bounds, we introduce a variant of tensor injective norm certification augmented by the Kikuchi hierarchy. This formulation yields a smooth tradeoff between runtime and clause density, and in particular enables refutation below the polynomial-time spectral threshold of $\tilde{\Omega}(n^{|S|/2})$. 

In this section, we focus on the even-order case, as the odd-order setting is well known to be technically cumbersome with the Cauchy Schwarz trick. We defer the treatment of the odd case to a later section, where we handle the non-Boolean domain in full generality. 

\begin{definition}[Kikuchi Matrix for Even $|S|$] \label{def:even-S-matrix}
For any even $|S|$, and $\ell \in \N$, and a given tensor $C_S$ of order-$|S|$, we define  $M_{S,\ell}$ to be the matrix of dimension $\binom{n}{\ell} \times \binom{n}{\ell}$ indexed by all subsets of size $\ell$ with entries $I, J\in \binom{n}{\ell}$ such that \[ 
M_{S,\ell}[I,J] = C_S[I\Delta J]
\]
for any $I\Delta J = |S|$, and $0$ otherwise.
\end{definition}

%

With that matrix defined, we can relate $C(x)$ to the corresponding quadratic form of the matrix.
Recall that for even $|S|$,  we have the polynomial \[C_S(x) = \sum_{\substack{ R \in [n]^k    } }\left( \1[ R \in E(\calH)] - p\right) x_{R_S}  \,,\]
hence  \jnote{is this still correct?}\snote{It's correct.}
\[ 
 \widetilde{x}^{\otimes \ell} M_{S,\ell} \cdot ( \widetilde{x}^{\otimes \ell})^T = \binom{n}{\ell - |S|/2  } \cdot \binom{|S|}{|S|/2} \cdot   C_S(x) 
 \,.\]
%

Finally, combining the above calculation and our norm bound in the next section completes the proof to our claim. We defer the complete calculation to ~\cref{sec:appendix-verification}  that verifies  $C_S(x) = o(m)$ assuming the conditions from \cref{thm:smooth-spectral-refutation}.
\jnote{im leaning towards sacking the section 5.5 - maybe keep thm 5.6 for even only}
\snote{I have updated Section 5.5.1 (Certifiable Concentration for Indicator Polynomials) to be consistent with the notations in Section 5.6 (Kikuchi for Asymmetric Tensors). The Generalized LP and Full Algorithm for general domain from Section 5.5 must be kept or moved to another place.}
\subsection{Summary of Algorithm for Boolean CSP} \jnote{i changed the name here} \label{sec:boolean-summary}
We formally present our final spectral refutation algorithm, specialized to a single predicate on Boolean domain, in this section. This formulation highlights how the algorithm ultimately reduces to certifying the injective norm of a tensor. We defer the proof of the required concentration bounds and the full proof, as they follow as a special case of our main theorem for the general (non-Boolean) domain. That said, we include a brief summary specialized to the Boolean domain to showcase the core ideas underlying our approach.

\jnote{this above para is new}

%
%
\begin{algorithm}[h]
\caption{Certification Algorithm for Fixed Boolean Predicates}
\label{alg:single-predicate-certification}
\begin{center}
\begin{minipage}{0.92\linewidth}

\textbf{Input:} A $k$-CSP instance $\calH$ on $n$ Boolean variables with single type of predicate $P$, and a parameter $\eps$.\\[4pt]
\textbf{Output:} Certified upper bound on the normalized objective value.\ \\
\textbf{Step 1 (Spectral Certification).} \\
\begin{itemize}
    \item For any subset restriction $S\subseteq [k]$, construct the matrix that flattens the polynomial $\Psi(x_S)$ defined from ~\cref{def:even-S-matrix} , and compute (or approximate) its spectral norm.
    \item If any spectral norm exceeds the desired threshold, output a trivial bound of $1$.
\end{itemize}

\textbf{Step 2 (Bias Enumeration).} \\
For each possible bias $\nu$ (there are $2n$ possibilities in the Boolean domain):
\begin{itemize}
    \item Solve the corresponding linear program to obtain a dominating polynomial $Q_\nu(\cdot)$ depending on $\nu$. Let $val_t(Q_\nu )$ be defined as \[ 
    \val_t(Q_\nu) = \sum_{|S|\leq t} c_S \cdot \nu^{|S|} 
    \]
    where we write $Q_\nu(x) = \sum_{S: |S|\leq t} c_S \cdot  x_S  $ for $x\in \{\pm 1\}^k$.
\end{itemize}

\textbf{Step 3 (Final Certification).} \\
Output
$
\max_{\nu} \val_t(Q_\nu) + \varepsilon
$
as the certified upper bound on the instance’s objective value.


\end{minipage}
\end{center}
\end{algorithm}


\SubexpRefutation*

Before turning to the general domain, we provide a roadmap of the proof to orient the reader.

Firstly, apply spectral certification to bound the deviations of any monomial-restriction $\Psi_S$ for $|S|\leq t$. This is simply captured by eigenvalue computations for $\min(O(2^k), O(k^t)) $-many   $O(n^{\ell } \times {n^\ell })$-size matrices (with each entry being a degree-$1$ polynomial of the hypergraph input). As a result, this shows that for any $\nu \in [-1,1]$,  we have a certifiable upper bound of \[ 
\max_{\substack{ x\in \{\pm 1\}^n \\ \langle \1_n ,x\rangle = \nu \cdot n } } \Psi_P(x) \leq \max_{\substack{ x\in \{\pm 1\}^n \\ \langle \1_n ,x\rangle = \nu \cdot n } } \Psi_{Q_\nu}(x)  =  \val_t(Q_\nu) +  \frac{\eps }{
\|Q_\nu \|_1  }  
\]
by~\cref{lem:certification-monomial} and our coefficient bound of $\|Q_\nu\|_1 \leq O_k(1)$, provided \[ 
m\geq \Omega_k( \frac{1}{\epsilon^2} \cdot n^{t/2} \cdot \frac{\ell }{\ell^{k/2}} \log n)
\]
(as well as the corresponding strengthened condition for odd $k$ when $\ell =O(\frac{n}{\log n} )$. Note that for any fixed bias $\nu$, \[ 
\sum_{S:|S|\leq t} c_S \cdot \nu^{|S|} = \val_t(Q_\nu)= \val_{P,t}(\nu) = \max_{\substack{\mu: \text{$t$-wise $\nu$-independent} }}  \E_{x\sim \mu}[P(x)] \,,
\]
where the first and last equality follow by definition, and the second equality by the strong duality of our LP.

Finally, by enumerating over all possible biases $\nu$,
we have completed the proof of our main theorem.

\subsection{Extension to Non-Boolean Domain} \label{sec:non-bool}
With the main ideas developed in the simpler setting of Boolean CSPs, we are now ready to state our main theorem in full generality.
\LargeAlphabetRefutation*

To prove the theorem, we highlight two key extensions:
\begin{enumerate}
	\item We generalize the LP from~\cref{eq:dual-lp-boolean}, which was formulated for a fixed marginal (bias), to the large-alphabet setting;
	\item We introduce a new notion of Kikuchi matrices tailored to the asymmetric polynomials that naturally arise in this setting, which we refer to as \emph{indicator polynomials};
	\item We replace the direct enumeration for the optimal bias from the Boolean argument by an $\epsilon$-net argument to avoid an exponential blow-up in $|D|$. \jnote{check me here}
\end{enumerate}

\paragraph{Generalized LP for Fixed Marginal (Bias)}

Given a $k$-ary predicate $P : D^k \to \{0,1\}$ on a finite domain $D$, and a distribution $\nu$ over $D$,
the following LP looks for a $t$-wise $\nu$-independent distribution $\mu$ (over $D^k$) maximizing the acceptance
probability for $P$:

\begin{mdframed}[frametitle =Primal LP for General Domain with Fixed Marginal $\nu$]	
\begin{align*}
\max \quad & \sum_{x\in D^k} P(x) \cdot  \mu(x) \\
\text{s.t.}\quad
& \sum_{x_W=b}\mu(x) = \nu_W(b)\qquad \forall W\in \binom{[k]}{t},\ b\in D^W \nonumber\\
& \mu(x)\ge 0 \qquad \forall x\in D^k. \nonumber
\end{align*}
\end{mdframed}

We do not explicitly enforce the total probability mass to be $1$, but that is implied by the constraints
$\pi_W(\mu)=\nu_W$. The following LP is dual to the above LP. The dual LP has variables $c(b)$ for subsets $W\in \binom{[k]}{t}$, $b\in D^W$.

\begin{mdframed}[frametitle =Dual LP for General Domain with Fixed Marginal $\nu$]	 
\begin{align*}
\numberthis
\label{eq:dual-lp-alphabet}
\min \quad & \sum_{W\in \binom{[k]}{t},\ b\in D^W} \nu_W(b)\,c_W(b) \\
\text{s.t.}\quad
& \sum_{W\in \binom{[k]}{t}} c_W(x_W) \ge P(x)\qquad \forall x\in D^k.
\end{align*}
\end{mdframed}

In other words, the dual LP looks for a degree-$t$ polynomial  for any $x\in D^k$ (with coefficients $c(b)$ for any $b\in D^t$) such that
\[
Q_\nu  (x)=\sum_{W\in \binom{[k]}{t},\ b\in D^W} c_W(b)\,\Iv{x_W=b} \geq P(x)
\]
point-wise dominates $P$ while minimizing the dual objective
\[
\val_t(Q_\nu) \coloneqq  \sum_{W\in \binom{[k]}{t},\ b\in D^W} \nu_W(b)\,c(b)={ \E_{x\sim \nu^t} Q_\nu(x) }\,.
\]
Here we use $\Iv{\cdot}$ to denotes the Iverson bracket notation, representing an indicator function.

\begin{remark}
	As we will be working with the polynomial of indicator functions in this section, we use two different symbols for two indicators as following,
	\begin{enumerate}
		\item $\Iv{x_W=b}$ is reserved for the indicator of assignment marginals;
		\item $\1[e\in E(\calH)]$ is reserved for the indicator of whether an edge $e$ is sampled in the instantiated random instance $\calH$.
	\end{enumerate}
\end{remark}

The above primal LP can be summarized as:
\[
\max_{\mu}\ \E_{\mu} P \quad \text{s.t. $t$-wise $\nu$-independent $\mu$.}
\]
The above dual LP can be summarized as:
\[
\min_{Q}\ \E_{\mu'} Q \quad \text{s.t. } Q\in \mathrm{span}\{\Iv{x_W=b}: W\in \binom{[k]}{t},\ b\in D^W\},\ Q\ge P,
\]
for any $t$-wise $\nu$-independent $\mu'$.

For every $\nu$, the primal LP is feasible ($\mu=\nu^k$ is valid) and has value at most $1$, so strong duality implies
primal and dual values always equal.
%


Given a homogeneous degree-$t$ polynomial
\[
Q(x)=\sum_{W\in \binom{[k]}{t},\ b\in D^W} c_W(b)\, x_{W,b},
\]
define its $\ell_1$-norm on coefficients
\[
\norm{Q}_1 \stackrel{\mathrm{def}}{=} \sum_{W\in \binom{[k]}{t},\ b\in D^W} \abs{c_W(b)}.
\]



\paragraph{Finding the Optimal Marginal for Large Alphabet}
Recall that for the boolean case, we have a nice characterization of solution bias as the Hamming weight of the assignment, and it is immediate that there are $2n$-bias parameter that needs to be check in total. This is less clear for the general domain. However, one can consider the following bias vector in replace of the total bias in the boolean case,
\begin{definition}[Marginal Vector]
	For a general domain $D$, and an assignment $x\in D^n$, we define the bias vector of an assignment as $b_x \in [0,1]^{|D|} $  \[
	b_x[a] \coloneqq (\# x_i = a)/n
	\,.\]
	In other words, this is simply the normalized frequency of $a\in D$ in assignment $x\in D^n$. \end{definition}
Observe that this definition extends immediately to a distribution of solutions, and to search over all bias in the general domain, it suffices for us to take an $\epsilon$-net over the bias grid of $\{0,1\}^{D}$.

\paragraph{Final Algorithm for General Domain}
We are now ready to describe the final algorithm for extension to general domain by combining with an $\epsilon$-net.  Again, it suffices for us to consider the single predicate case as the reduction from boolean applies verbatim to the multiple predicate case.

\begin{enumerate}
	\item \textbf{Spectral Certification:} Apply Spectral Certification of the resulting degree-$t$ polynomial via the corresponding matrix for any basis polynomial now given by any subset-restriction \jnote{ok to call it subset-restriction?}  $w$ for $w\subseteq [k]^{\leq t}$. This runs in time $O(n^{O(\ell)} \cdot k^t) =O_{k, |D|} (n^{O(\ell)})$. Output "Failure" if any matrix spectral norm bound fails, i.e., exceeds anticipated bound from the analysis. 
	\item   \textbf{Marginal Enumeration}: Take $\calS$ to be an $\eps$-net  over the space of marginal vectors $[0,1]^{D}$, for each bias-point $\nu \in \calS$ in the net, apply the dual LP to find the best dominating polynomial for each fixed marginal, and hence the certifiable bound for $\val_t(Q_\nu)$. Finally, output the \[  \max_{\nu\in \calS}  \val_t(Q_\nu) + \eps  \] as an upper bound for the instance's value.
	 \end{enumerate}
We defer the presentation of the full algorithm to appendix (see~\cref{alg:multiple-predicate-general-domain}) for completeness.

\paragraph{Proof to Main Theorem for General Domain}

\begin{proof}[Proof of \cref{thm: alphabet-ref}] \label{proof:main-spectral}
This argument largely mirrors the Boolean case showcase earlier, with three modifications: 
(i) we incorporate the matrix norm bounds for the odd case; 
(ii) we invoke the specialized partial certification lemma appropriate for the modified matrix construction over general domains; and 
(iii) we include an $\epsilon$-net argument, with the error budget split into three parts.

\paragraph{Reduction to heavy predicates.}
For each predicate $P_\sigma$, let $\Psi_\sigma$ denote the subinstance consisting of all constraints with predicate $P_\sigma$. Then
\[
\Val_\Psi(x)
=
\E_{\sigma \sim E(\Psi)}\!\big[\Val_{\Psi_\sigma}(x)\big].
\]

Set
\[
\delta \coloneqq \frac{\epsilon}{6 \cdot 2^{|D|^k}}.
\]
We call a predicate $P_\sigma$ \emph{$\delta$-heavy} if $\rho(\sigma)\ge \delta$, and \emph{light} otherwise.

By concentration, with probability $1-o_n(1)$ the empirical frequency $\tilde\rho(\sigma)$ satisfies
\[
|\tilde\rho(\sigma)-\rho(\sigma)| \le o(1)
\]
for every $\sigma$. Hence the total empirical mass of light predicates is at most
\[
\sum_{\sigma:\,\rho(\sigma)<\delta}\tilde\rho(\sigma)
\le
2^{|D|^k}\delta + o(1)
=
\epsilon/6 + o(1)
\le
\epsilon/3
\]
for all sufficiently large $n$.

\paragraph{Certification for each heavy predicate.}
We now apply the single-predicate refutation to each heavy predicate. For each $\delta$-heavy predicate $P_\sigma$, the subinstance $\Psi_\sigma$ contains
\[
m_\sigma = \tilde\rho(\sigma)\, m \ge (\delta-o(1))m
\]
constraints. Therefore, provided $m$ satisfies the density requirement in \cref{lemma:spec-norm-bounds}  with a sufficiently large constant slack depending on $\delta$, each heavy sub-instance also satisfies the required density, which is 
 $m$ satisfies the following density constraint: for some large enough constant $C = C(\rho, \delta)$
		 \[
    m \geq \begin{cases}
      \displaystyle C n^{t/2} \frac{\ell}{\ell^{t/2}} \log n & \text{if $t$ is even} \\
      \displaystyle C n^{t/2} \frac{\ell}{\ell^{t/2}} \sqrt{\log n} & \text{if $t$ is odd}
    \end{cases},
  \]
  for any $\ell\in \N$ s.t. $t/2 \leq \ell \leq \tilde{O}(n)$.

Fix a heavy predicate and a bias $\nu \in \calS$. Applying spectral certification over all subsets $S\subseteq [k]$ with $|S|\le t$ and all assignments $b\in D^S$, we obtain that for every $\nu$-biased assignment $x\in D^n$,
\[
\Psi_{Q_\nu}(x) \le \val_t(Q_\nu) + \frac{\epsilon}{3}.
\]
Indeed,
\[
\Psi_{Q_\nu}(x)-\val_t(Q_\nu)
=
\sum_{\substack{S\subseteq [k]\\ |S|\le t}} \sum_{b\in D^S} c_S(b)
\left(
\Pr_{\sigma\in\Psi}[x_{\sigma(S)}=b] - \Pr_{y\sim \nu^S}[y=b]
\right),
\]
whose magnitude is at most
\[
\sum_{S,b} |c_S(b)|
\left|
\Pr_{\sigma\in\Psi}[x_{\sigma(S)}=b] - \Pr_{y\sim \nu^S}[y=b]
\right|.
\]
The desired bound now follows from \cref{lem:alphabet-partial-certification}  \jnote{re-route to the new assymetric lemma} and \cref{lem:lp-bounded-dual-norm}, after setting the target slack there to be $\epsilon/3$.

\paragraph{The $\epsilon$-net step.}
Choose the net $\calS$ finely enough so that
\[
\max_{\nu}\val_t(Q_\nu)
\le
\max_{\nu\in\calS}\val_t(Q_\nu)+\frac{\epsilon}{3}.
\]
By LP duality, the left-hand side is exactly the target benchmark
\[
\max_{\mu:\,\text{$t$-wise independent}} \E_{x\sim\mu}[P(x)].
\]

\paragraph{Conclusion.}
Combining the above bounds, the contribution of light predicates is at most $\epsilon/3$, the certification error on heavy predicates is at most $\epsilon/3$, and the discretization over the bias contributes another $\epsilon/3$. Therefore, for every assignment $x\in D^n$,
\[
\Val_\Psi(x)
\le
\max_{\mu:\,\text{$t$-wise independent}} \E_{x\sim\mu}[P(x)]
+
\epsilon
+
o(1),
\]
which completes the proof.
\end{proof}

\subsubsection{Certifiable Concentration for Indicator Polynomials}
\snote{I have updated the notations in this subsection to be consistent with Section 5.6 (Kikuchi Matrices for Indicators)}

Throughout this section, we will write $S\subseteq [k]$ as a subset corresponding to the indicator-basis that arises from the dual of $Q_\nu$, and write $\beta:S \rightarrow D $ as a \emph{local} assignment of variables in $S$. 
Also fixes a bijection  $\phi$ between $[|S|]$ and $S$, giving an ordering of variables in $S$.

Given an ordered-tuple $T \in V^{|S|}$, we will write ``$x_T = \beta$'' as a shorthand for ``$x_{T(i)} = \beta(\phi(i))$ for every $i \in [|S|]$''.

\snote{I changed the naming below: $C_{S,\beta}$ is the deviation polynomial, not a deviation tensor.
Later we will define the deviation tensor $C_S$.}
\begin{definition}[Deviation polynomial for $(S,\beta)$] \label{def:deviation-poly}
Let $\calH$ be a directed $k$-uniform hypergraph.
  For any $S\subseteq [k]$ and $\beta \in D^S$, define the \emph{deviation polynomial} $C_{S,\beta}$ (in indeterminates
  $\set{x_{v,\alpha}}_{v \in V, \alpha \in D}$ that represent the indicators $\Iv{x_v = \alpha}$)
\[
C_{S,\beta} (x) \coloneqq \Pr_{\sigma\in E(\calH)}\big[x_{\sigma(S) } = \beta\big]
\;-\;
\Pr_{y\sim \nu^S}\big[y=\beta\big].
\]

\end{definition}

Equivalently, writing $E(\calH)$ for the directed hyperedge set of $\calH$, for any $x\in D^n$ that is $\nu$-biased,  we may express $C_{S,\beta}(x)$ as

\[
C_{S,\beta}(x)
=
\frac{1}{|E(\calH)|}\sum_{e\in E(\calH)} 
\Iv{x_{e(S)} = \beta }
-
\frac{1}{n^{|S|}}\sum_{\gamma  \in [n]^{|S|}}
\Iv{x_{\gamma}=\beta }.
\]

We can first sum over ordered-tuples $\gamma$, and then sum over clauses $R$ whose $S$-restriction yield $\gamma$ (writing $R_S = \gamma$), obtaining
\[
C_{S,\beta}(x)
=
\frac{1}{|E(\calH)|}
\sum_{\gamma\in [n]^{|S|}}
\left(
\sum_{\substack{R\in [n]^k \\ R_S=\gamma}}
\big(\mathbf{1}[R\in E(\calH)] - p\big)
\right)
\Iv{x_{S }=\beta(S)} +o_n(1),
\]
where $p = m/n^k$ is the expected constraint density, and we apply the concentration of $\frac{|E(\calH)|}{m} = 1\pm o_n(1) $ from~\cref{claim:edge-concentration}. 

Recall that we define our random variables for constraint on order-tuple $R \in [n]^k $ as 
\[
G_R \coloneqq \mathbf{1}[R\in E(\calH)] - p,
\]
we may equivalently write
\[
C_{S,\beta}(x)
=
\frac{1}{|E(\calH)|}
\sum_{\gamma\in [n]^{|S|} }
\sum_{\substack{R\in  [n]^k \\ R_S=\gamma}}
G_R \cdot \Iv{x_{\gamma}=\beta(S) } + o_n(1).
\]
\begin{remark}
	It is important to notice that $\gamma$ in the above sum is an ordered-tuple as opposed to a subset. This is the reason our subsequent matrix is indexed by ordered-tuples as opposed to by subsets.
\end{remark}

For convenience, we also define the tensor $C_S \in \R^{V^{|S|}}$ with entries \[ 
C_S[\gamma] \coloneqq \frac{1}{|E(\calH)|}  \sum_{R\in [n]^k, R_S=\gamma } G_R .
\]
This  allows us to write the deviation polynomial as a polynomial of the indicators. \[ 
C_{S,\beta}(x) = \sum_{\gamma\in [n]^{|S|}} C_S[\gamma]  \cdot \Iv{x_{\gamma} = \beta }\,.
\]

%



\begin{remark}
  We will use $C_{S,\beta}(x)$ to refer to a polynomial of $x$, and equivalently, a polynomial of indicators of $x$.
  We will use $C_S[\cdot]$ to refer to a tensor entry defining the coefficient of a specific indicator monomial in $x$.
\end{remark}
We now introduce our lemma for certification of  polynomials of indicators.
\begin{restatable}{lemma}{alphabetpartialcertification}\label{lem:alphabet-partial-certification}
For any subset $S \subseteq [k]$, any $\ell = O(n)$, and any $\eps>0$, there is an
algorithm running in time $n^{O(\ell)}$ that certifies
\[ \max_{x\in D^n} |C_{S,\beta}(x)| \le \eps + o_n(1), \]
provided $m$ satisfies the density constraint in \cref{thm: alphabet-ref}. 
\end{restatable}


We defer the complete proof of this lemma to the end of this section , while in the subsequent section we establish the certification question continues to reduce to studying spectral norm bound of random matrices. From this point onwards, we consider an arbitrarily fixed marginal $\beta(S)\in D^S$. Our analysis will only depend on the underlying subset-restriction $S$ while independent of the particular marginal restriction $\beta(S)$. Hence, we will suppress the dependence of $b$ throughout the section unless necessary.

We now give a self-contained analysis for the basic-level and even $S$ to showcase the use of polynomial of indicators. We will then generalize it to higher-level via a new Kikuchi matrix specialized to our setting, and then extend that to odd $S$.
 
\paragraph{Certification for (basic) even-$S$ }

Given an order-$|S| = r$ tensor $C_S$ over index set $[n]$ where $r$ is even, consider its flatten matrix $M_w\coloneqq M_{w, \ell = r/2}$ that has rows and columns indexed by ordered-tuples $[n]^{r/2}$,
so that $M_{S}[\alpha, \beta]\coloneqq C_{S}[\alpha_1, \al_2, \dots, \al_{r/2}, \beta_1,\dots, \beta_{r/2}]$ for row $\alpha\in [n]^{r/2}$ and column $\beta\in [n]^{r/2}$.

 For any assignment $x \in D^n$, consider vectors $y$ and $z$ (indexed by $\gamma\in [n]^{r/2}$) to be the indicators of whether $a_\gamma$ agrees with the first and second halves of $b$:
\[
y[\gamma]  \coloneqq \Iv{x_\gamma=b_{1,...,r/2}} = 
\qquad\text{and}\qquad
z[\gamma] \coloneqq \Iv{x_\gamma=b_{r/2+1, r}}.
\]
writing $b_{1,...,r/2}$ as the restriction to the first-$r/2$ variables, and similarly $b_{r/2+1, r}$ for the second half.
Then for every $\alpha,\beta\in [n]^{r/2}$,  $y[\alpha] z[\beta]=\Iv{x_{\alpha\beta}=b}  $. Thus
\[
C_S(x) 
= \sum_{\gamma  \in [n]^{r}  } M_S[\al, \beta] \,\Iv{x_\gamma=b}
= \sum_{\alpha,\beta\in [n]^{r/2}} M_w[\alpha, \beta] \cdot  y_\alpha \cdot  z_\beta
= y^\top M z \leq n^{r/2} \cdot \|M\|_{sp}
\]
which has magnitude at most $\norm{y}\,\norm{M}\,\norm{z}$ by Cauchy--Schwarz and definition of spectral/operator norm. Also,
$\norm{y}\le n^{t/4}$ and $\norm{z}\le n^{t/4}$.

Plugging the norm bound, we recover a non-trivial refutation bound when \[
n^{r/2} \cdot \|M_{w} \|_{sp} \leq m \cdot \eps \,,
 \]
 i.e., when \[ 
 m\geq \Omega_{k,|D|} (\frac{1}{\eps^2}\cdot  n^{r/2} \log n  )\,.
 \]

\paragraph{Certification for Odd $S$}

We now demonstrate the Cauchy--Schwarz trick for odd tensor:

\begin{align*}
  C_\beta(x)^2 &= \Paren{\sum_{T \in [n]^{|S|}} C[T] \Iv{x_T = \beta}}^2 \\
  &= \Paren{\sum_{t \in [n]} \Iv{x_t = \beta(\phi(\abs{S}))} \sum_{\alpha \in [n]^{|S|-1}} C[\alpha,t] \Iv{x_\alpha =
  \beta|_{S'} }}^2 \\
  &\leq \Paren{\sum_{t \in [n]} \Iv{x_t = \beta(\phi(\abs{S}))}^2 }
  \Paren{\sum_{t \in [n]} \Paren{\sum_{\alpha \in [n]^{|S|-1}} C[\alpha,t] \Iv{x_\alpha = \beta|_{S'}} }^2 } \,.
  \numberthis \label{eq:cauchy-schwarz-trick}
\end{align*}
The first factor is at most $n$.
The second factor equals
\begin{equation} \label{eq:chauchy-schwarz-trick2}
  C^*_\beta(x) \coloneqq \sum_{t \in [n]} \sum_{\alpha,\gamma \in [n]^{|S|-1}} C[\alpha,t]C[\gamma,t]
  \Iv{x_\alpha = \beta|_{S'}}
  \Iv{x_\gamma = \beta|_{S'}} \,.
\end{equation}
\cref{eq:chauchy-schwarz-trick2} is the sum of the squared terms
\[ C^{sq}_\beta(x) \coloneqq \sum_{t \in [n]} \sum_{\alpha \in [n]^{|S|-1}}
  C[\alpha,t]^2 \Iv{x_\alpha = \beta|_{S'}}^2 \]
and the cross terms
\[ \tilde C_\beta(x) \coloneqq \sum_{t \in [n]} \sum_{\substack{\alpha,\gamma \in [n]^{|S|-1} \\ \alpha \neq \gamma}}
C[\alpha,t] C[\gamma,t] \Iv{x_\alpha = \beta|_{S'}} \Iv{x_\gamma = \beta|_{S'}} \,.
\]

\subsection{Kikuchi Matrices for Indicators}

In prior works, Kikuchi matrices are usually applied to problems with $\{\pm 1\}$ domain and to symmetric tensors. To generalize it to a larger domain, we instead show that it can be applied to the $\{0,1\}$ domain (representation of an arbitrary domain) and to asymmetric tensors.
To distinguish from the usual notion of Kikuchi matrices that we study in the previous section, we will also call the Kikuchi matrices in this section the \emph{Kikuchi matrices for indicators}. We start by some definitions.

\subsubsection{Even Tensor}

Fix a subset $S \subseteq [k]$ of even size.
Write $\beta:S \rightarrow D $ as a \emph{local} assignment of variables in $S$.

\begin{examples}[Example of $S$-partite indices] For $v_a,v_b,v_c,v_d\in [n]$ and alphabets $\al_1, \al_2, \al_3$, and $k=4$, 
	$I=\{(v_a, \al_1, 1 ),(v_b,\al_2, 2), (v_c,\al_1, 2), (v_d, \al_2, 3)\}$ is an index for $\ell=4$. 
\end{examples}

\begin{definition}[Even Kikuchi index set] \label{def:even-kikuchi-index}
  The rows and columns of the $(S, \ell)$-Kikuchi matrix are indexed by ``$S$-partite level-$\ell$ indicators'' $I,J \in
  \binom{V \times D \times S}{\ell}$.
  In other words, an $S$-partite level-$\ell$ indicator is an $\ell$-size subset $I$ of tuples
  $(v,\alpha,j) \in V \times D \times S$.
\end{definition}

For convenience, we extend the operation of symmetric difference to these indices that will come in handy when we would like to specify a corresponding constraint. For two indices $I,J$ of level-$\ell$ defined above, their symmetric difference $I\Delta J$ is now a subset of tuples of the form $(v,\alpha,j) \in [n]\times D \times S $.

\begin{examples}[Example of symmetric difference of $S$-partite indices]  Consider $k=4, S=[4]$.
	For $v_a,v_b,v_c,v_d, v_e,v_f\in [n]$, consider indices
  \[ I=\{(v_a, D_1, 1),(v_b, D_2, 2), (v_c,D_1, 1), (v_d,D_2, 2)\}\]  and \[ J =\{(v_e,D_1, 3),(v_f,D_2, 4), (v_c,D_1, 1), (v_d,D_2, 2) \}\,.\] Then $I\Delta J = \{(a,1), (b,2), (e,3), (f,4) \}$ is well-behaved in the sense that it encodes the vertex tuple $(a,b,e,f)$ (defined below), and $I\Delta J$ corresponds to the tensor entry $C_S[a,b,e,f]$.
\end{examples}

The above example motivates the following definition that characterizes the pairs of indices that allow us to deduce an entry in the deviation tensor $C_S$.

\begin{definition}[Well-behaved index pair]
  A pair $I,J \in \binom{V \times D \times S}{\ell}$ is well-behaved for $S$ if:
  \begin{enumerate}
    \item $|I \Delta J| = |S|$; and
    \item Restricted to the $S$-labels of $I \Delta J$, each element of $S$ appears exactly once.
  \end{enumerate}
\end{definition}

\begin{definition} \label{def:V-D-pair}
  Fix an arbitrary bijection $\phi$ between $[|S|]$ and $S$.
  Consider any pair $I,J$ of indices well-behaved for $S$.
  Define $V(I\Delta J) \in V^{|S|}$ to be the $|S|$-tuple of vertices, so that its $i$-th entry equals the unique
  vertex $v$ such that $(v,\alpha,\phi(i)) \in I\Delta J$ for some $\alpha \in D$.
  Define $D(I\Delta J) \in D^S$ to be the function from $S$ to domain values, so that its $j$-th entry (for $j \in S$)
  equals the unique $\alpha$ such that $(v,\alpha,j) \in I\Delta J$ for some $v \in V$.
\end{definition}

Given a tensor $C \in \R^{V^{|S|}}$, local assignment $\beta \in D^S$ and $\ell \in \N$, define the polynomial
(in indeterminates $\set{x_{v,\alpha}}_{v \in V, \alpha \in D}$ that are intended to represent the indicator functions
$\Iv{x_v = \alpha}$)
\begin{equation}
  C_\beta (x) \coloneqq \sum_{\substack{T \in V^{|S|}}} C[T] \Iv{x_T = \beta} \,,
\end{equation}
  where ``$x_T = \beta$'' is a shorthand for
\[ x_{T(i)} = \beta(\phi(i)) \quad \text{and} \quad \text{ for every } i \in [|S|]. \]

Given a tensor $C \in \R^{V^{|S|}}$ and local assignment $\beta \in D^S$, the $(S,\ell)$-Kikuchi matrix $M_\beta$ has
entries
\[ M_\beta[I,J] \coloneqq C[V(I\Delta J)] \Iv{D(I\Delta J) = \beta} \]
if $I,J$ are well-behaved, and $M[I,J] \coloneqq 0$ otherwise.

Given $I \subseteq V \times D \times S$, write ``$x_I = D(I)$'' to mean ``$x_v = \alpha$ for every
$(v,\alpha,j) \in I$''.

\begin{definition}[$x^{\circledcirc \ell}$ Tensor Lifts for Indicator Vectors]\label{def:tensor-lift-vector}
  For any assignment $x\in D^V$, let $x^{\circledcirc \ell}$ denote a vector indexed by $S$-partite indices with
  entries
  \[
    x^{\circledcirc \ell }[I] \coloneqq  \Iv{x_I = D(I)} = \prod_{(v,\alpha,j) \in I} \Iv{x_v = \alpha}  \,. \]
  for any for any $S$-partite indices $I$.
\end{definition}

\paragraph{Why Norm-Bound Again?}
With the above definition, we now show the certification question of the deviation polynomial again boils down to the understanding of the associated random matrices.

\begin{lemma}[Even Kikuchi quadratic form] \label{lem:large-alphabet-kikuchi-even}
For $x \in D^V, \beta \in D^L$,
\[
  (x^{\odot \ell})^\top M_\beta\,x^{\odot \ell} = \binom{(|V|-1)|S|}{\ell-|S|/2} \cdot \binom{|S|}{|S|/2}
  \cdot C_\beta(x) \,.
\]
\end{lemma}

\begin{proof}
  Expand the LHS as a sum over well-behaved pairs $I,J \in \binom{V \times D \times L}{\ell}$ of
  \begin{align*}
    & C[V(I\Delta J)] \Iv{D(I\Delta J) = \beta} \Iv{x_{I \cup J} = D(I \cup J)} \\
    = \;& C[V(I\Delta J)] \Iv{x_{V(I\Delta J)} = \beta} \Iv{x_{I \cup J} = D(I \cup J)} \,,
  \end{align*}
  Denoting $T = V(I\Delta J)$, the previous sum is equivalent to the sum over $T \in V^{|S|}$ of
  \[ C[T] \Iv{x_{V(I\Delta J)} = \beta} \cdot N_x(T) \]
  where $N_x(T)$ is the number of well-behaved $I,J$ such that $V(I\Delta J) = T$ and $x_{I \cup J} = D(I \cup J)$.

  For each $T \in V^{|S|}$, we count the number of such $I,J$.
  \begin{enumerate}
		\item Choices for $D$: Each $D$ label in $I\cup J$ is fixed by the assignment $x$, and there is a unique choice.
		\item Choices for $V \times S$ in $I \cap J$:
      Out of the possible labels in $V \times S$, exactly those labels already chosen in $I \Delta J$ are forbidden.
      That's because if $I\Delta J$ has some tuple $(v,\alpha,j)$, then $I \cap J$ cannot have any tuple
      $(v,\alpha',j)$ sharing the same $V \times S$ label.
      We pick $\ell - |S|/2$ such $V \times S$ labels for $I \cap J$, hence $\binom{|V||S| - |S|}{\ell-|S|/2} =
      \binom{(|V|-1)|S|}{\ell-|S|/2}$ choices.
    \item Choices for $V \times S$ in $I \Delta J$: To split the labels in $I \Delta J$ evenly among $I$ and $J$,
      there are $\binom{|S|}{|S|/2}$ choices.
	\end{enumerate}
\end{proof}

\subsubsection{Odd Tensor}

Fix a subset $S \subseteq [k]$ of odd size, and an arbitrary bijection $\phi$ from $[|S|]$ to $S$.

\begin{definition}[Odd Kikuchi index set]\label{def:odd-kikuchi-indices}
  Let $S' \coloneqq S \setminus \set{\phi(\abs{S})}$ be $S$ with the last element removed (under the $\phi$-ordering),
  and $L \coloneqq S' \times [2]$ represent two copies of $S'$.
  The rows and columns of the $(S, \ell)$-Kikuchi matrix are indexed by ``$L$-partite level-$\ell$ indicators'' $I,J \in
  \binom{V \times D \times S' \times [2]}{\ell}$.
  In other words, an $L$-partite level-$\ell$ indicator is an $\ell$-size subset $I$ of tuples
  $(v,\alpha,j,z) \in V \times D \times S' \times [2]$.
\end{definition}

\begin{definition}[Well-behaved index pair]
  A pair $I,J \in \binom{V \times D \times S' \times [2]}{\ell}$ is well-behaved for $S$ if:
  \begin{enumerate}
    \item $|I \Delta J| = |L| = 2(|S|-1)$; and
    \item Restricted to the $L$-labels of $I \Delta J$, each element of $L$ appears exactly once.
    \item Among the $|S|-1$ tuples in $I \setminus J$, exactly $(|S|-1)/2$ tuples have their $L$ labels in
      $S' \times \set{1}$ while the other $(|S|-1)/2$ tuples have their lables in $S' \times \set{2}$.
      (This also implies the same statement when $I \setminus J$ is replaced with $J \setminus I$.)
  \end{enumerate}
\end{definition}
Note that for odd $S$, the definition of well-behaved pairs has an additional requirement:
Each copy of $S'$ in $I \Delta J$ is split evenly among $I$ and $J$.
This addition requirement is useful for the moment calculation for the Cauchy--Schwarz trick.

Given a well-behaved pair $I,J$, define $V(I\Delta J) \in V^{[|S|-1] \times [2]}$ to be the tuple of vertices indexed by
$(i,z) \in [|S|-1] \times [2]$, so that the $(i,z)$-coordinate of $V(I\Delta J)$ equals the unique vertex $v \in V$ such
that $(v,\alpha,\phi(i),z) \in I \Delta J$ for some $\alpha \in D$.
Further define $D(I\Delta J) \in D^L$ to be the function mapping $(j,z) \in (S \setminus \set{\phi(\abs{S})}) \times [2]
= L$ to the unique domain value $\alpha \in D$ such that $(v,\alpha,j,z) \in I\Delta J$ for some $v \in V$.

Given a tensor $C \in \R^{V^{|S|}}$, define the tensor $\tilde C \in \R^{V^{2(|S|-1)}}$ such that for $\alpha,\gamma
\in V^{|S|-1}$,
\[ \tilde C[\alpha,\gamma] \coloneqq \sum_{t \in V} C[\alpha, t] C[\gamma, t] \]
if $\alpha \neq \gamma$, and $\tilde C[\alpha,\gamma] \coloneqq 0$ otherwise.
Given $T \in V^{[|S|-1] \times [2]}$, we also write $\tilde C[T]$ to mean $\tilde C[T_1,T_2]$, where $T_1(i) = T(i,1)$
and $T_2(i) = T(i,2)$ for $i \in [|S|-1]$.

Additionally, given $\beta \in D^S$, define the polynomial (in indeterminates
$\set{x_{v,\alpha}}_{v \in V, \alpha \in D}$ that are intended to represent the indicator functions $\Iv{x_v = \alpha}$)
\begin{equation} \label{eq:cross-poly}
  \tilde C_\beta (x) \coloneqq \sum_{\substack{\alpha,\gamma \in V^{|S|-1}}}
  \tilde C[\alpha,\gamma] \Iv{x_{\alpha\gamma} = \beta} \,,
\end{equation}
  where ``$x_{\alpha\gamma} = \beta$'' is a shorthand for
\[ x_{\alpha(i)} = \beta(\phi(i)) \quad \text{and} \quad
  x_{\gamma(i)} = \beta(\phi(i)) \quad \text{ for every } i \in [|S|-1]. \]
Similarly, given $R \in D^L$, write ``$x_R = \beta$'' as a shorthand for
\[ x_{R(i,z)} = \beta(j) \quad \text { for every } j \in S', z \in [2]. \]

\begin{definition}[$(S,\ell)$-Kikuchi matrix $M_\beta$] \label{def: kikuchi-beta}
	Given a tensor $C \in \R^{V^{|S|}}$ and $\beta \in D^S$, we define 
 the $(S,\ell)$-Kikuchi matrix $ M_{S,\beta, \ell}$ has entries
  \[ M_{S,\beta,\ell}[I,J] \coloneqq \tilde C[V(I\Delta J)] \Iv{D(I\Delta J) = \beta} \]
  if $I,J$ are well-behaved, and $M_{S,\beta,\ell}[I,J] \coloneqq 0$ otherwise.

When the dependence on $S,\ell$ is clear, we also write it as $M_\beta$ for convenience.
\end{definition}

\begin{lemma}[Odd Kikuchi quadratic form] \label{lem:nonbool-odd}
For $x \in D^V, \beta \in D^S$,
\[
  (x^{\odot \ell})^\top M_\beta\,x^{\odot \ell} = \binom{2(|V|-1)(|S|-1)}{\ell-|S|+1} \cdot \binom{|S|-1}{(|S|-1)/2}^2
  \cdot \tilde C_\beta(x) \,.
\]
\end{lemma}

\begin{proof}
  Expand the LHS as a sum over well-behaved pairs $I,J \in \binom{V \times D \times L}{\ell}$ of
  \begin{align*}
    & \tilde C[V(I\Delta J)] \Iv{D(I\Delta J) = \beta} \Iv{x_{I \cup J} = D(I \cup J)} \\
    = \;& \tilde C[V(I\Delta J)] \Iv{x_{V(I\Delta J)} = \beta} \Iv{x_{I \cup J} = D(I \cup J)} \,,
  \end{align*}
  where ``$x_T = \beta$'' is a shorthand for ``$x_{T(i,z)} = \beta(\phi(i))$ for every $i \in [|S|-1], z \in [2]$''.
  Denoting $\alpha(i) = V(I\Delta J)(i,1), \gamma(i) = V(I\Delta J)(i,2)$, the previous sum is equivalent to the sum
  over $\alpha, \gamma \in V^{|S|-1}$ of
  \[ \tilde C[\alpha,\gamma] \Iv{x_{V(I\Delta J)} = \beta} \cdot N_x(\alpha,\gamma) \]
  where $N_x(\alpha,\gamma)$ is the number of well-behaved $I,J$ such that $V_1(I\Delta J) = \alpha$,
  $V_2(I\Delta J) = \gamma$ and $x_{I \cup J} = D(I \cup J)$.

  For each $\alpha,\gamma \in V^{|S|-1}$, we count the number of such $I,J$.
  \begin{enumerate}
		\item Choices for $D$: Each $D$ label in $I\cup J$ is fixed by the assignment $x$, and there is a unique choice.
		\item Choices for $V \times L$ in $I \cap J$: Out of the possible labels in $V \times L$, exactly those labels
      already chosen in $I \Delta J$ are forbidden.
      That's because if $I\Delta J$ has some tuple $(v,\alpha,j,z)$, then $I \cap J$ cannot have any tuple
      $(v,\alpha',j,z)$ sharing the same $V \times L$ label.
      We pick $\ell - (|S|-1)$ such $V \times L$ labels for $I \cap J$, hence $\binom{|V||L| - 2(|S|-1)}{\ell-|S|+1} =
      \binom{2(|V|-1)(|S|-1)}{\ell-|S|+1}$ choices.
    \item Choices for $V \times L$ in $I \Delta J$: To split each copy of $S'$ in $I \Delta J$ evenly among $I$ and $J$,
      there are $\binom{|S|-1}{(|S|-1)/2}$ choices for each copy of $S'$.
	\end{enumerate}
\end{proof}

\subsubsection{Deviation Bound} \label{sec:deviation-bound}

\begin{lemma} \label{lem:deviation-bound}
  For any subset $S \subseteq [k]$ of size at least $2$, any $\eps > 0$, there is $\delta > 0$ such that provided
  \[
    m \gtrsim_{k,|S|,|D|,\eps} \begin{cases}
      \displaystyle n^{|S|/2}\cdot \frac{\ell}{\ell^{|S|/2}}\cdot \log n
      & \text{if $|S|$ even and } \displaystyle \frac{|S|}2 \leq \ell \leq \delta n  \\
      \displaystyle n^{|S|/2}\cdot \frac{\ell}{\ell^{|S|/2}}\cdot \sqrt{\log n}
      & \text{if $|S|$ odd and } \displaystyle |S|-1 \leq \ell \leq \frac{\delta n}{\log n}
    \end{cases}
  \,, \]
  then for any $\beta \in D^S$, whp
  \[ \max_{x\in D^V} |C_{\beta}(x)| \le \eps. \]
\end{lemma}

\begin{proof}
  We split our proofs into even and odd $S$ again.
  \paragraph{Even $|S|$.}
  From \cref{lem:large-alphabet-kikuchi-even}, we have
  \[
    \binom{(n-1)|S|}{\ell - |S|/2} \cdot  \binom{|S|}{|S|/2} \cdot   C_\beta(x)
    = (x^{\circledcirc \ell })^T \cdot  M_\beta \cdot x^{\circledcirc \ell }\,.
  \]
  Rearranging and applying spectral theorem, we have
  \begin{equation} \label{eq:nonbool-even}
    |C_\beta(x)| \leq \|x^{\circledcirc \ell} \|_2^2 \cdot \|M_\beta\|_{sp} \cdot \frac{1}{\binom{(n-1)|S|}{\ell - |S|/2  }}
    \cdot \frac{1}{\binom{|S|}{|S|/2}} \,;
  \end{equation}
  Plugging the norm bounds and the following observations,
  \begin{enumerate}
    \item For any $x\in D^n$, let $x^{\circledcirc \ell}$ be the vector defined from \cref{def:tensor-lift-vector}.
      We have  \[ \| x^{\circledcirc \ell} \|_2^2 = \binom{n|S|}{\ell} \]
    \item For any $|S|/2 \leq \ell \leq n/2$,
      \[  \frac{ \binom{n|S|}{\ell} } {\binom{(n-1)|S|}{\ell - |S|/2  } }
      \leq \Paren{\frac{n|S|}{(n-1)|S|-\ell}}^\ell \frac{(n|S|)^{|S|/2}}{\ell^{\underline{|S|/2}}}
      \lesssim_{|S|} \Paren{\frac{n}{\ell}}^{|S|/2}   \]
    \item Our even-$S$ norm bound is from~\cref{lemma:spec-norm-bounds}
      \[ \|M_{\beta}\|_{sp} \lesssim_{k,|S|} \sqrt{ \frac{m}{n^{|S|/2}} } \cdot \ell^{|S|/4}  \cdot \sqrt{\ell  \log n} \]
  \end{enumerate}

  We have
  \begin{align*}
    \eqref{eq:nonbool-even}  & \lesssim_{k,|S|} \frac{\binom{n|S|}{\ell } \cdot
    \sqrt{  \frac{m}{n^{|S|/2}} \cdot \ell^{|S|/2+1} \cdot \log n  }}{\binom{(n-1)|S|}{\ell - |S|/2}} \\
    &\lesssim_{k,|S|} \Paren{\frac{n}{\ell}}^{|S|/2} \cdot \ell^{|S|/4 + 1}\cdot  \sqrt{m} \cdot n^{-|S|/4}  \sqrt{\log n}
  \end{align*}

  Picking the above to be less than $\eps m$ and solve for $m$ gives us
  \[ m \gtrsim_{k,|S|,\eps}  n^{|S|/2} \cdot \frac{1}{\ell^{|S|/2 - 1 } } \cdot \log n \,. \]

  \paragraph{Odd $|S|$.}
  Recall that for odd-case from \cref{eq:cauchy-schwarz-trick},
  \[ C_\beta(x)^2 \leq  n \cdot  C_\beta^*(x) \,. \]
  We have
  \[ C_\beta^*(x) = \tilde{C}_{\beta}(x) + C^{sq}_\beta(x) \,. \]
  From \cref{claim:sq-term-negligible-appendix}, we have
  \[ |C^{sq}_\beta(x)| \cdot n \leq \eps^2 m^2/2 \,. \]
  It suffices for us to verify
  \[ |\tilde{C}_{\beta}(x)| \cdot n \leq \eps^2 m^2/2 \,, \]
  \cref{lem:nonbool-odd} implies
  \begin{equation}
    |\tilde C_\beta(x)| \leq \|x^{\circledcirc \ell} \|_2^2 \cdot \|M_\beta\|_{sp} \cdot
    \frac{1}{\binom{2(n-1)(|S|-1)}{\ell - |S|+1  }} \cdot \frac{1}{\binom{|S|-1}{(|S|-1)/2}^2} \,;
  \end{equation}
  The calculation from this point is now identical to the previous even-$S$ case, except that the norm bound given by
  \cref{lemma:spec-norm-bounds} is
  \[
    \|M_{\beta}\|_{sp}
    \lesssim_{k,|S|,|D|} \sqrt{ \frac{m^2}{n^{|S| }} } \cdot \sqrt{\ell}^{|S|-1} \sqrt{\ell  \log n} \,.
  \]
  In other words, it is sufficient to verify \[ 
  |\tilde C_\beta(x)|\leq  \frac{\|x^{\circledcirc \ell} \|_2^2 \cdot \|M_{\beta}\|_{sp} }
    { \binom{2(n-1)(|S|-1)}{ \ell - |S| + 1 } \cdot \binom{|S|-1}{(|S|-1)/2}^2 } \leq \frac{\eps^2 m^2}{2n} \,.
  \]
  Solving out $m$, it suffices for us to take
  \[ m \gtrsim_{k,|S|,|D|,\eps} n^{|S|/2} \cdot \frac{1}{\ell^{|S|/2 - 1 } } \cdot \sqrt{ \log n} \,. \qedhere \]
\end{proof}

We now proof \cref{lem:concentrated}, which is mostly the SoS version of \cref{lem:deviation-bound}.

\concentrated* \jnote{is it supposed to be here?}\snote{The proof \cref{lem:concentrated} should come after
\cref{lem:deviation-bound}, because the former is mostly the SoS version of the latter.}

\begin{proof}
  When $t$ is even and $\ell = t/2$, the result follows from \cref{thm:concentration-random-instance}.

  When $t$ is even and $t/2 \leq \ell \leq \delta n$, our goal is the turn even case in the proof of
  \cref{lem:deviation-bound} into SoS.
  Recall that our polynomials have indeterminates $\set{x_{v,a}}_{v \in V, a \in D}$.
  Given any set $L$ and any $L$-partite index set $I \in \binom{V \times D \times L}{\ell}$, define the monomial
  $x_I \coloneqq \prod_{(v,a,j) \in I} x_{v,a}$.
  Consider the polynomial
  \[ \|x^{\circledast\ell}\|_2^2 \coloneqq \sum_{I \in \binom{V \times D \times L}{\ell}} x_I^2\,. \]
  For any $Q \in {[k] \choose t}, b \in D^Q$, using \cref{lem:large-alphabet-kikuchi-even},

  \cref{claim:spectral-norm-sos} implies \cref{eq:nonbool-even} has an SoS version
  \[ \pm C_b(x) \leq \|x^{\circledast \ell}\|_2^2 \cdot \norm{M_b}_{sp} \cdot
  \frac{1}{\binom{(n-1)t}{\ell - t/2  }} \cdot \frac{1}{\binom{t}{t/2}} \,. \]
  The SoS axioms yield the bound $\|x^{\circledast \ell}\|_2^2 \leq {nt \choose \ell}$ by
  \cref{claim:indicator-sos-bound}.
  And we conclude as in \cref{lem:deviation-bound} that
  \[ \Axioms \quad \vdash_{2\ell} \quad \pm C_b(x) \leq \eps \quad \vdash_{2\ell} \quad C_b(x)^2 \leq \eps^2 \]
  provided
  \[ m \gtrsim_{k,t,\eps}  n^{t/2} \cdot \frac{1}{\ell^{t/2 - 1 } } \cdot \log n \,. \]

  When $t$ is odd and $t-1 \leq \ell \leq \delta n/\log n$, our goal is to turn the odd case in the proof of
  \cref{lem:deviation-bound} into SoS.
  Consider any $Q \in {[k] \choose t}, b \in D^Q$, and any bijection $\phi$ between $[t]$ and $Q$.
  The Cauchy-Schwarz trick in \cref{eq:cauchy-schwarz-trick} has a degree-$2(\ell+1)$ SoS proof.
  Expand $C^*_b(x) = C^{sq}_b(x) + \tilde C_b(x)$.
  There is a degree-$2\ell$ SoS proof of
  \begin{align*}
    C^{sq}_b(x) &= \sum_{v \in [n]} \sum_{\alpha \in [n]^{t-1}} C[\alpha,v]^2
    \prod_{i \in [t]} x^2_{\alpha(i),b(\phi(i))} \\
    &\leq \sum_{v \in [n]} \sum_{\alpha \in [n]^{t-1}} C[\alpha,v]^2 \leq \eps^2 m^2/2n \,,
  \end{align*}
  where the first inequality is due to $x_{v,a}^2 \leq 1$ for $v \in V, a \in D$ by the axioms, and the last inequality
  holds whp by \cref{claim:sq-term-negligible-appendix}.

  \cref{claim:spectral-norm-sos} implies \cref{eq:nonbool-even} has an SoS version
  \[ \tilde C_\beta(x) \leq \|x^{\circledast \ell} \|_2^2 \cdot \|M_\beta\|_{sp} \cdot
    \frac{1}{\binom{2(n-1)(t-1)}{\ell - t+1  }} \cdot \frac{1}{\binom{t-1}{(t-1)/2}^2} \,; \]
  Using the spectral norm bound in \cref{lemma:spec-norm-bounds} and the norm bound in \cref{claim:indicator-sos-bound},
  conclude as in the proof of \cref{lem:deviation-bound} that
  \[ \Axioms \quad \vdash_{2(\ell+1)} \quad C_b(x)^2 \leq \eps^2 \]
  provided
  \[ m \gtrsim_{k,|D|,t,\eps} n^{t/2} \cdot \frac{1}{\ell^{t/2 - 1 } } \cdot \log n \,. \qedhere \]
\end{proof}

\begin{claim} \label{claim:spectral-norm-sos}
  Let $W$ be a finite set, and $M$ be a symmetric $W$-by-$W$ matrix.
  There is a degree-$2$ SoS proof of the inequality (in indeterminates $\set{x_i}_{i \in W}$)
  \[ \sum_{i,j \in W} M_{i j} x_i x_j \leq \norm{M}_{sp} \sum_{i \in W} x_i^2\,. \]
\end{claim}

\begin{proof}
  Let $\lambda = \norm{M}_{sp}$ and $I$ be the $W$-by-$W$ identity matrix.
  $\lambda I - M$ is positive semidefinite, and the spectral theorem implies
  $\lambda I - M = \sum_k \lambda_k u^{(k)} (u^{(k)})^\top$, where $\lambda_k \geq 0$ are the eigenvalues and
  $u^{(k)} \in \R^W$ are eigenvectors of $\lambda I - M$.
  Therefore
  \[ \sum_{i,j \in W} M_{i j} x_i x_j = \lambda \sum_{i \in W} x_i^2
  - \sum_k \lambda_k \Paren{\sum_{i\in W} u^{(k)}_i x_i}^2 \leq \lambda \sum_{i \in W} x_i^2 \,. \qedhere \]
\end{proof}

\begin{claim} \label{claim:indicator-sos-bound}
  Consider any set $L$ and the polynomial
  \[ \|x^{\circledast\ell}\|_2^2 \coloneqq \sum_{I \in \binom{V \times D \times L}{\ell}} x_I^2\,,
  \quad \text{where } x_I \coloneqq \prod_{(v,a,j) \in I} x_{v,a} \,. \]
  From the axioms in \cref{def:sos-axioms}, there is an SoS proof
  \[ \Axioms \quad \vdash_{2\ell} \quad \|x^{\circledast \ell}\|_2^2 = {|V||L| \choose \ell} \,. \]
\end{claim}

\begin{proof}
  Whenever $I \in \binom{V \times D \times L}{\ell}$ contains two tuples $(v,a,j)$ and $(v,a',j')$ sharing the
  same vertex $v \in V$ but with distinct $a,a' \in D$,
  \[ \Set{ x_{v,a}x_{v,a'} = 0}_{v \in V, a,a' \in D, a \neq a'} \quad \vdash_{2\ell} \quad x_I = 0 \,. \]
  Therefore, the only index sets $I$ that have non-zero term in $\|x^{\circledast\ell}\|_2^2 $ are those without
  two triples $(v,a,j)$ and $(v,a',j')$ with distinct $a,a' \in D$.
  Denote by $\cI$ the collection of such $I$.
  Summing over such $I$ is equivalent to first summing over $J \in \binom{V \times L}{\ell}$ and then over $a \in D$ for
  each $(v,j) \in J$:
  \[ \|x^{\circledast \ell}\|_2^2 = \sum_{I \in \cI} x_I^2
  = \sum_{J \in \binom{V \times L}{\ell}} \prod_{(v,j) \in J} \sum_{a \in D} x_{v,a}^2 \,. \]
  Under the additional axioms $\Set{\textstyle \sum_{a \in D} x_{v,a}^2 = 1}_{v \in V}$, the inner sum
  $\sum_{a \in D} x_{v,a}^2$ equals $1$.
  Therefore so is each product over $(v,j) \in J$.
  There are $\binom{|V||L|}{\ell}$ choices for $J \in \binom{V \times L}{\ell}$.
\end{proof}

\section{Spectral Norm Bounds for Polynomials of Random Support} \label{sec:norm-bound} 
It  remains to bound the spectral norms of the matrices constructed above.
However, before we delve into the proof, one question that we need to answer, especially to the experts in the study of random CSP, is the following.
\paragraph{Why Can't We Use Existing Results?} 
 Admittedly,  at a high level,
the matrices in our hands are highly reminiscent of matrices that have appeared throughout the random CSP literature.
However, existing results crucially rely on uniformly random signings, and we are not aware
of any theorem that can be directly applied in our present setting of fixed literals and
normalized adjacency matrices. In particular, standard trace-method arguments that rely on
evenness properties of contributing walks no longer apply. On the other hand, it  is conceivable that the even-case can be obtained by standard black-box matrix concentration bounds. However, we resort to the trace method calculations in the end due to notorious odd case that we will now elaborate. 
 
 \paragraph{Hypergraph Decomposition No More}
 While this similarity of our matrix is closely related to the connection between a randomly signed and a
centered adjacency matrices, there in fact underlines a conceptual challenge. Prior analyses for odd $k$, despite being stated for arbitrary hypergraphs with random signings, rely heavily on decompositions and row pruninngs of the underlying
hypergraph even for random hypergraphs. Conceptually, such a
decomposition is not immediate, or even well-defined, in the setting of normalized adjacency matrices, prohibiting us from appealing to these results as a black box. Notably, our argument works directly
with the polynomial tensor arising from the Cauchy-Schwarz lift in its natural
flattening, and bypasses hypergraph decomposition with our encoding argument for walks that arise in the trace method calculations.

In the next section, we carry out this analysis using a factor-assignment scheme that
bounds the contribution of each step in a trace walk individually. This approach allows us
to control the spectral norm through local combinatorial and probabilistic bounds, without
keeping track of the full global structure of the walk. We now briefly recap this framework.

\paragraph{Spectral Norm Bounds from Local Bounds}
Traditional trace moment analyses typically rely on global structural properties of the walks
contributing to the trace, while we apply the machinery developed in~\cite{JPRTX, KPX24,kothari2025smoothtradeofftensorpca}. We refer the reader to \cite{kothari2025smoothtradeofftensorpca} for a self-contained introduction, while we briefly recap the strategy as the following.

The crux of these approaches is to break correlations along long walks via a
\emph{factor-assignment scheme} that distributes the global combinatorial contributions of a
walk across its individual steps. This allows the analysis to proceed in a local, step-by-step
manner, without explicitly tracking the intricate dependencies that span the entire walk.

Designing such a factor-assignment scheme is technically delicate. It must satisfy two
competing requirements:
\begin{enumerate}
  \item \textbf{Completeness}: the scheme must faithfully account for all global contributions to
  the trace moment;
  \item \textbf{Sharpness}: the local bound assigned to each step must match the best possible
  global upper bound.
\end{enumerate}

We now formalize the outcome of a factor-assignment scheme.

\begin{definition}[Factor-assignment scheme and step-bound function $B_q(\cdot)$] \label{def:block-value-trace}
For a matrix $M$ and an integer $q \ge 1$, consider a fixed factor-assignment scheme. We call
$B_q(M)$ a valid \emph{step-bound function} for the scheme if
\[
\mathbb{E}\!\left[\operatorname{Tr}\!\left( (M\cdot M^T)^{q}\right)\right]
\;\le\;
\mathsf{MatrixDimension} \cdot B_q(M)^{2q},
\]
and moreover, the scheme assigns to each individual step of any contributing walk a factor
of at most $B_q(M)$. When the dependence on $q$ and $M$ is clear from context, it is omitted.
\end{definition}

Once such a scheme is specified, deriving a spectral norm bound becomes relatively
straightforward. The task reduces to verifying that each step of the walk satisfies the
corresponding local bound $B_q(\cdot)$, which is exactly the candidate global spectral norm
bound one aims to prove. This local viewpoint extends naturally to more complex matrix
structures and has proven effective in a variety of graph-based settings. As an up-shot of the scheme, once taking $q$ large enough (usually $\Omega(\log \mathsf{MatrixDimension}  )$), the local bound of $B_q(\cdot)$ also becomes an high probability bound. 

\paragraph{Organization of this section} 	We introduce our norm bounds in the next subsection while defer the formal proof to our norm bound statements to the end of this section~\cref{sec:norm-bound-wrap} . We introduce some preliminaries needed for our machinery, and then establish the bound for the even-case as a warm-up for our machinery in ~\cref{sec:even-bound}, and subsequently the odd-case in ~\cref{sec:odd-bound}. 

\subsection{Norm Bound Statements.}

\begin{restatable}[Spectral Norm Bounds]{lemma}{SpectralNormBounds}
\label{lemma:spec-norm-bounds}
For a random instance $\calH$, for any $\ell \in \mathbb{N}$, and for any $\beta: L \rightarrow D$, let $M_{S,\ell}$ denote the $M_\beta $ matrix for $(S,\ell)$-Kikuchi matrix \jnote{check me}, and for $\ell < Cn$ for some constant
$C>0$, the following holds with probability at least $1-o_n(1)$:
\[
\|M_{S,\ell}\|_{\mathrm{sp}}
\;\le\;
O_{k,|S|}\!\left(
\sqrt{\frac{m}{n^{|S|/2}}}\;
\ell^{|S|/4}\;
\sqrt{\ell\log n}
\right),
\]
and
\[
\|M_{S,\ell}\|_{\mathrm{sp}}
\;\le\;
O_{k,|S|}\!\left(
\frac{m}{n^{|S|/2}}\;
\ell^{(|S|-1)/2}\;
\sqrt{\ell\log n}
\right),
\]
provided that
\[
m
\;\ge\;
\Omega_{k,|S|}\!\left(
n^{|S|/2}\cdot
\frac{\ell}{\ell^{|S|/2}}\cdot
\log n
\right).
\]

Moreover, when $|S|$ is odd, the same bound holds with an additional
$\sqrt{\log n}$ saving provided
\[
m
\;\ge\;
\Omega_{k,|S|}\!\left(
n^{|S|/2}\cdot
\frac{\ell}{\ell^{|S|/2}}\cdot
\sqrt{\log n}
\right),
\]
as long as
\[
\ell \le C \frac{n}{\log n}
\]
for some constant $C>0$.
\end{restatable}

	\subsection{Preliminaries for Factor Assignment Scheme}


For starters, let's reintroduce the definition of trace-walks which would immediately apply to trace walks of Kikuchi matrices for odd-$r$ as well to incorporate vertices at each step of the walk getting summed over while outside the left/right boundaries. For readability purpose, we suggest the reader to focus on the even-$r$ case in the first pass.

For concreteness, we recall our definitions of indices for Kikuchi matrices 
\begin{definition}[Even Kikuchi index set][Restatement of \cref{def:even-kikuchi-index}]
  The rows and columns of the $(S, \ell)$-Kikuchi matrix are indexed by ``$S$-partite level-$\ell$ indicators'' $I,J \in
  \binom{V \times D \times S}{\ell}$.
  In other words, an $S$-partite level-$\ell$ indicator is an $\ell$-size subset $I$ of tuples
  $(v,\alpha,j) \in V \times D \times S$.
\end{definition}

And analogous for the indices for odd given as 
\begin{definition}[Odd Kikuchi index set][Restatement of~\cref{def:odd-kikuchi-indices}]
  Let $S' \coloneqq S \setminus \set{\phi(\abs{S})}$ be $S$ with the last element removed (under the $\phi$-ordering),
  and $L \coloneqq S' \times [2]$ represent two copies of $S'$.
  The rows and columns of the $(S, \ell)$-Kikuchi matrix are indexed by ``$L$-partite level-$\ell$ indicators'' $I,J \in
  \binom{V \times D \times S' \times [2]}{\ell}$.
  In other words, an $L$-partite level-$\ell$ indicator is an $\ell$-size subset $I$ of tuples
  $(v,\alpha,j,z) \in V \times D \times S' \times [2]$.
\end{definition}

And recall our definition of $(S,\ell)$-Kikuchi matrices and any assignment $\beta$ viewed as a conditioning of the variables,

\begin{definition}[$(S,\ell)$-Kikuchi matrix $M_\beta$][Restatement of \cref{def: kikuchi-beta}]
	Given a tensor $C \in \R^{V^{|S|}}$ and $\beta \in D^L$, we define 
 the $(S,\ell)$-Kikuchi matrix $ M_{S,\beta, \ell}$ has entries
\begin{align*}
    M_\beta[I,J] &\coloneqq \tilde C[V(I\Delta J)] \Iv{D(I\Delta J) = \beta} 
\end{align*} 
if $I,J$ are well-behaved, and $M[I,J] \coloneqq 0$ otherwise.
\jnote{add even v. odd}
\end{definition}
For the purpose of bounding the norm, it is most informative to track the dependence on $n$, as we do not attempt to optimize factors involving $S$, $|D|$, $k$, or other absolute constants. Throughout this section, we treat $\beta$ as fixed and write $M_{S,\ell} = M_{S,\beta,\ell}$, since the dependence on $\beta$ is not central to our analysis.
Throughout the trace method analysis, we view the entries of $M_\beta$ as degree-1 (even $S$) and degree-2 (odd $S$) polynomials of random variables in the underlying input via \cref{def:deviation-poly} that ultimately defines $\tilde{C}$.

\jnote{added}

\begin{definition}[Generalized Definitions of Trace-walk] For any $r\coloneqq |S|>2$,
	A trace-walk for Kikuchi matrix of level-$\ell$ at length-$2q$ is a sequence of sets $(S_1, W_1, S_2, W_2, \dots,S_{2q}, W_{2q},S_{2q+1})$ such that   \begin{enumerate}
		\item Walk-boundary sets: for any $i\in [2q]$, $S_i$ is an index set as defined in \cref{def:even-kikuchi-index} for even $r$, or an index set as defined in \cref{def:odd-kikuchi-indices} for odd $r$. 
		\item  For even-$r$,  $|V(S_i \Delta S_{i+1})| = r$ for any $i\in [2q]$  and $W_i = \emptyset$ for any $i\in [2q]$;
		\item For odd-$r$, $|V(S_i \Delta S_{i+1})| = 2r-1$, and $|W_i|=1$ for any $i\in [2q]$; 
		\item (Closed-walk) $S_1 = S_{2q+1}$;
		\item  (At least $2$-steps) Each hyper-edge (underlying random variable) appears in at least two steps;
	\end{enumerate}
	Equivalently, we view it as a walk of length-$(2q)$ such that
	\begin{enumerate}
		\item  at each step-$t \in [2q]$, we move from the left boundary $U_t = S_t$ to the right boundary $V_{t} = S_{t+1} = U_{t+1}$ with potential intermediate-set $W_t$;
		\item For even-$r$, each step traverses exactly $1$ hyper-edge (viewed as an ordered-tuple) that contains $V(S_t\Delta S_{t+1}) $;
		\item for the odd-$r$, each step traverses along $2$ hyper-edges incident to  $U_A \cup V_A \cup W_t$ and $U_B\cup V_B\cup W_t$ where we assume that $U_A, U_B$ are equal partitions for $S_{t}\setminus S_{t+1}$ and $V_A, V_B$ for $S_{t+1}\setminus S_t$ as prescribed by the intermediate set.  \jnote{adjust as needed}
        \item Each hyper-edge appears in at least two steps.
	\end{enumerate} 
\end{definition}
See below for a diagram illustration of a step in the trace-walk for $k=|S|=3$ (equivalently, an entry of the corresponding the Kikuchi Matrix). In addition to the labels in $[n]$, each circle vertex additionally comes with a label in $[D]$ and in $S'\times [2]$.
\begin{figure}[h]\label{fig:p2bar}
    \centering
    \includegraphics[height=3cm]{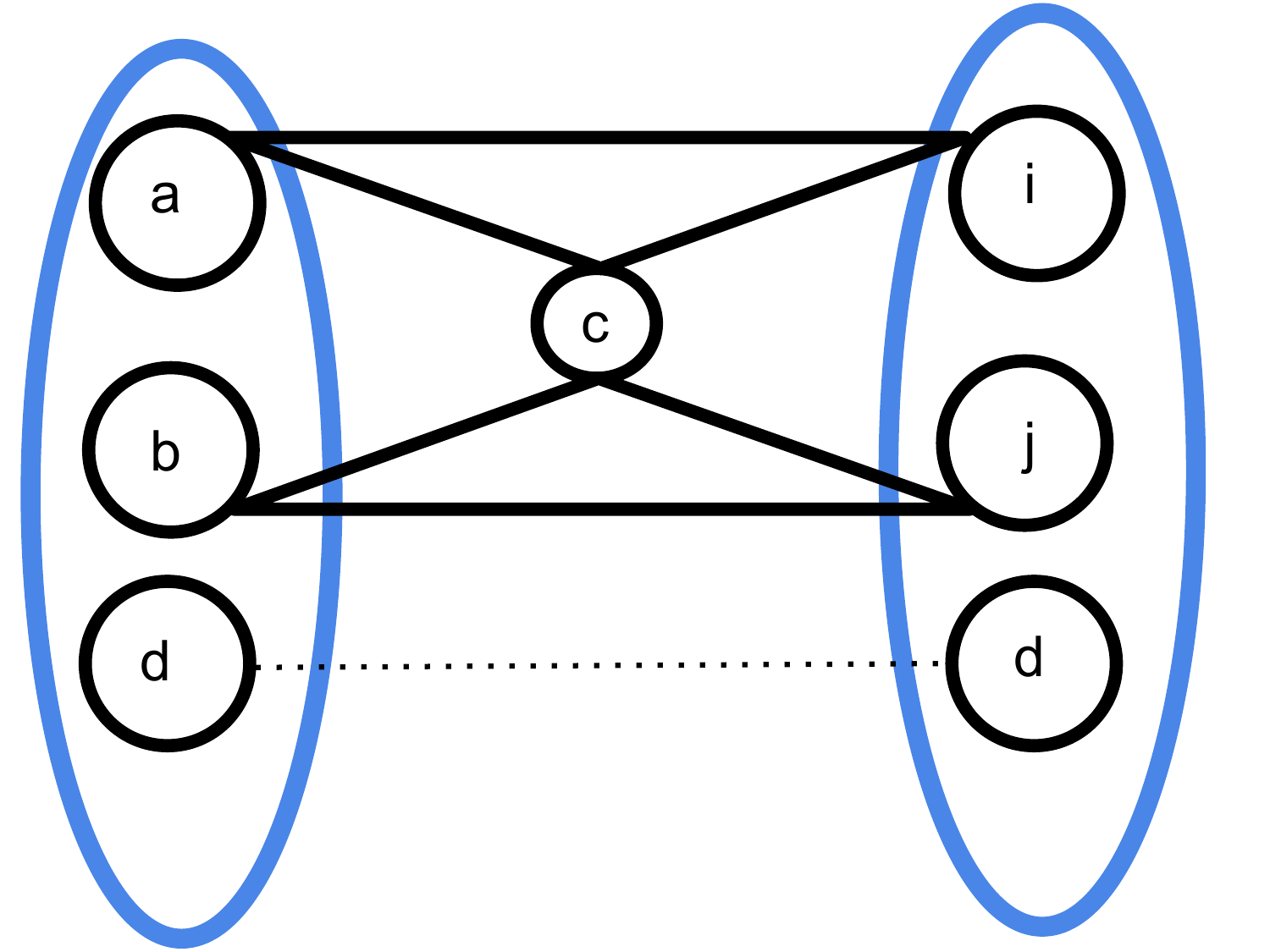}
    \caption{Example of Kikuchi-Enry for $|S|=k=3$}
    \caption*{$M_S[(a,b,d), (i,j,d)] = \sum_{c\not\in \{a,b,d,i,j\} }  G_{(a,c,i)} \cdot G_{b,c,j)}$}
        \label{fig:P2'}
\end{figure}

With the generalized definition of trace-walk, we may now relate the expected trace of the Kikuchi matrix to the weighted sum of trace-walks as defined above.
\begin{proposition} For any $r\geq 3, \ell \in \N $ and any $q\in \N$, 
	\[ 
	\E[\Tr(M_{S,\ell}\cdot   (M_{S,\ell})^T)^{q}]  \leq \sum_{P:\text{trace-walk of length }2q} \val(P) \,.
	\]
    where we recall that \[ 
    \val(P) \coloneqq \E_G \left[ \prod_{i\in [2q]} M_{S,\ell} [S_i, S_{i+1}] \right] = \prod_{e = (I,J)\in E(P)} \E_G[(G_e)^{\mul_P(e)} ]
    \]
    \jnote{add definition for $g_e$}
\end{proposition}

We consider the following factor-assignment scheme. For starters, we focus on the "combinatorial" factor that arises from counting of the walk, and a crucial observation to get us started is to split the cost of specifying random variables (viewed as hyper-edges) among vertices, in particular, among the incoming active vertices. 

	\begin{definition}[Hyperedge Status at Each Step]
		For each step-$t$ in a trace walk, for each hyperedge $g_e$ that appears in the matrix entry corresponding to the step-$t$ (the single edge in the even case, and two edges in the odd case), we assign them the following labels, \begin{enumerate}
			\item $F$-Step: a hyperedge that is making its first appearance in the walk;
			\item $L$-Step: a hyperedge that is making its final appearance in the walk;
			\item $H$-Step: a hyperedge that is making its middle appearance (neither first nor final) in the walk.		\end{enumerate}
	\end{definition}

\paragraph{Bounding the Edge-Values}
	For each (possible) hyperedge $e \in [n]^k$, we associate with a random variable $
	G_e = \1[e\in E(\calH)] - p$ for $p = \frac{m}{n^k}$
	 ,  and note that \[ 
	\E[(G_e)^{t}] \leq p 
	\]
	for $t\geq 2$, and \[ 
	\E[G_e ] = 0\,.
	\]
	It suffices for us to consider the following edge-value assignment scheme.
	\begin{definition}[Edge-Assignment Scheme]
	For any step-$t$ in the trace-walk, and any edge $e$ that appears in step-$t$, we define 
	\begin{enumerate}
		\item $F$-step: assign a factor of $B_{ev}(e,t) = p$; 
		\item $H/L$-step: assign a factor of $B_{ev}(e, t)  = 1$.
	\end{enumerate}
	
	\end{definition}
	 	 \begin{proposition}[Verification of Global Value Bound from Steps]
    For any walk $P$, recall that $\val(P) \coloneqq \E_G \left[ \prod_{i\in [2q]} M_{S,\ell} [S_i, S_{i+1}] \right] = \prod_{e = (I,J)\in E(P)} \E_G[(G_e)^{\mul_P(e)} ]$, we have \[ 
    \val(P) \leq   \prod_{e\in E(P) } p      \]
    where $E(P)$ is collection of (hyper)-edges traversed by the walk $P$, and $mul_P(e)$ denotes the multiplicity/number of times $e$ is traversed in $P$.
\end{proposition}

\begin{proof}
Since the product factorizes across random variables, it suffices for us to verify \[ 
\E[\val(P)] =  \E[(G_e)^{\mul_P(e)} ] ]\leq \prod_{e\in E(P) } \E[G_e^{\mul_P(e)}]
\,.\]
By mean zeroness of the random variable, each edge appears at least twice in the walk, and in which case we have  \[ \E_G[(G_e)^{\mul_P(e)} ] = (1-p)^{\mul_P(e)} \cdot p + (1-p) \cdot (-p )^{\mul_P(e)} \leq p    \]
for any $\mul_P(e)\geq 2$. 
\end{proof}

We are now ready to apply the above machinery, and bound the block-value of each step.
	
	\subsection{Warm-up: Block-Value Bound for Even $S$} \label{sec:even-bound}
		To get started, we first showcase the application of this scheme for even $S$. Note that similar bound can also be obtained via standard matrix concentration inequalities while we apply our machinery to prepare for our exposition to the odd case.

We consider the following factor assignment scheme for combinatorial factor and edge-value combined.
	\paragraph{Factor Assignment for  Hyperedge}
\begin{enumerate}
	\item We focus on factors of $n$ first. For an $F$-step, i.e. an hyperedge appearing for the first time, it suffices for us to identify all its incidental vertices to specify the edge. This can be done via  identifying each out-going active vertex's label in $\ell$, that is a factor of $\ell^{|S|/2 } $ in total; moreover, for each vertex outside the current boundary,  a label in $[n]$ is sufficient , that is a total of $n^{k-|S|/2 }$;
	\item For an $H$ or $L$-step, a cost of $O(q)$ is sufficient to identify this edge as this edge has made a prior appearance in the walk.		
	\item For each hyperedge that appears at least twice, we assign a total factor at most $p$ through its edge value; 
	\item For factors of $S$, it suffices for us to assign the entire factor to each time the walk enters a vertex, and another half of the factor departs from a vertex - that is a factor of at most $|S|^{|S|/2}$. There is no factor of $D$ as we consider fixed $\beta$. \jnote{check for $D$ factor}	\item  The culminates in the following factor assignment scheme, assign a factor of $ 
	\ell^{|S|/2 } \cdot n^{k-|S|/2 } \cdot p  \cdot |S|^{|S|/2}
	$
	to the $F$-step, and a factor of $2q$ for each $H$ or $L$-step. \label{claim:factor-assignment-even}
	\end{enumerate}

In short, the above factor assignment scheme simply states that we assign a factor of $n^{k-|S|/2 } \cdot p $ (i.e. degree-bound) for each $F$-step, and a trivial $q$ bound for each non-$F$ step. To conclude norm bounds, we use the following average-idea from the trace-method calculation in Theorem A.7 of \cite{JPRTX}, which states the following, 
\begin{displayquote}
	For a given self-returning walk (i.e. a circle), we have freedom in encoding the walk in either direction (clockwise or counterclock-wise). And it suffices for us to pick a minimal-cost direction for the encoding, which is at most the average-cost of either direction.
\end{displayquote}

In particular, with this observation, we can bound the factor of a given step by considering traversal in either direction, and take the (geometric) average. This gives us a convenient way to avoid the imbalance in the above assignment scheme because each $F$-step in one direction becomes an $L$-step in the other direction (and vice versa), while an $H$-step remains as an $H$-step . Note that in the even-case, this can also be readily achieved by noting each hyperedge necessarily appear in an $F$ and an $L$-step that can be paired together. That said, we continue with this observation that will be more useful later for the odd case. 

We now apply the above idea to obtain a block-value bound for each step. For convenience, we will work in the clockwise direction , let $B(\cdot)$ denote the final block value bound obtained after averaging of either direction, and let $B_{\rightarrow}(\cdot)$ and $B_{\leftarrow}(\cdot)$ denote the block-value bound in either direction of $\rightarrow$ (clockwise) and $\leftarrow$ (counter-clockwise), from the factor assignment scheme without any averaging.

\begin{remark}
	We will apply the aforementioned scheme directly to obtain bounds of $B_{\leftarrow}(\cdot)$ and $B_{\rightarrow}(\cdot)$, and then use these two bounds to conclude a bound for $B(\cdot)$ as our final averaged-bound.
\end{remark}

\paragraph{Block-Value Bound}
Throughout this section, set $q = C \cdot \ell\log n$ for some constant $C>1$. We now establish the following lemma that bounds the block-value for even $|S|$.

\begin{lemma} \label{lem: block-value-bound-even}
	For even $|S|$, for any $\ell$, and $q = C\cdot \ell \log n$ for some constant $C>1$, it suffices for us to take \[ 
	B(M_{S,\ell}) \leq O\left(   \sqrt{\frac{m}{n^{|S|/2}} \cdot \ell ^{|S|/2+1 } \log n}  \right)
	\]
	to satisfy the requirement in \cref{def:block-value-trace}, provided $m=\Omega( n^{|S|/2} \cdot \ell^{1-|S|/2}\log n)$.
\end{lemma}
\begin{proof}

Fixing a traversal direction (say clockwise $\rightarrow$), note the following,
\begin{enumerate}
	\item For any $F$-step, we bound it by \[ 
	B(F) \leq \sqrt{B_{\rightarrow}(F) \cdot B_{\leftarrow}(L)   }  \,,
	\]
	and since $B_{\rightarrow}(F) = 	n^{k-|S|/2 } \cdot p 
$ and   $B_{\leftarrow}(L)= q $ via \cref{claim:factor-assignment-even}, we have
\begin{align*}
	B(F) &\leq \sqrt{\ell^{(|S|/2 }\cdot  	n^{k-|S|/2 } \cdot p  \cdot q} =  \sqrt{	\ell^{|S|/2 }\cdot  n^{k-|S|/2 } \cdot \frac{m}{n^k} \cdot \ell \log n }\\ & \leq  \sqrt{\frac{m}{n^{|S|/2}} \cdot \ell ^{|S|/2+1 } \log n}\,.
\end{align*}

	\item For any $H$-step, since it is an $H$-step in either direction, we simply have \[ 
	B(H) = 2q  
	\]
	from the $H$-step assignment in \cref{claim:factor-assignment-even}. This is upper bounded by $O\left(   \sqrt{\frac{m}{n^{|S|/2}} \cdot \ell ^{|S|/2+1 } \log n}  \right)
$ provided \[ 
m\geq \Omega( n^{|S|/2} \cdot \ell^{1-|S|/2}\log n)
\]
	\item For any $L$-step, the bound is identical from the bound for $F$ as \[
	B(L)\leq \sqrt{B_{\rightarrow}(L) \cdot B_{\leftarrow}(L)}=  \sqrt{\frac{m}{n^{|S|/2}} \cdot \ell ^{|S|/2+1 } \log n}\,.
	 \]
\end{enumerate}

Combing the above, for any even $|S|$, we have a block-value bound of \[ 
B(M_{S,\ell} ) \leq B(F) + B(L) + B(H) = O\left(   \sqrt{\frac{m}{n^{|S|/2}} \cdot \ell ^{|S|/2+1 } \log n}  \right)
\,.\]
This completes the proof for the even case. 

\end{proof}

	\subsection{Block-Value Bound for Odd $S$} \label{sec:odd-bound}
	
	The distinction from the even-case is that each step, or equivalently, each matrix entry is now a sum over products over \emph{two} neighboring hyperedges, and therefore each block step receives instead a label in $\{F, H, L\} \times \{F, H, L\} $. It is not too hard to verify that the case of $F\times F$ is largely identical to the even-case, while special care is warranted in the steps involving an $H$ or $L$-step, particularly the $H\times H$ step when both hyperedges are making a middle appearance, for illustrative purposes.

	\paragraph{Encoding for $L\times L$ and $H\times H$}
	Let's now focus on the case of both incidental edges in the upcoming step have made the prior appearances. For  concreteness, we show case the issue with an $H\times H$-step as its block-value is more transparent without any averaging.
	
	  Naively, we may use a factor of $2q$ to specify each edge as in the even-case for $L$ and $H$-step, giving a total of $O(q^2)$. 
	However, interestingly, such a bound is not tight enough for us to obtain a final bound as setting \[
	B(H\times H) = O(q^2 )\leq  \text{Target-Norm-Bound for Odd}
	 \]
	 does not give us the desired bound on $m$ when setting $q= \Theta( \ell \log n)$ as needed. However, with some reverse-engineering, one can verify that a bound of $B(H\times H)\leq O(q) $ would instead suffice as solving for \[ 
	 q \leq  \text{Target-Norm-Bound for Odd}  
	 \]
	 gives us the desired condition of $ m \geq \Omega_{|S|}( n^{|S|/2} \cdot \frac{\ell}{\ell^{|S|/2}}   \log n) $.
	 
	 This prompts us to show that a cost of $O(q)$ is sufficient in this case, and this is achieved by the following claim.
	 \begin{claim}
	A cost of $O(q)$ is sufficient in total to identify both hyperedges in that have made their prior appearance, i.e. edges in $H\times H$ and $L\times L$-steps.
		 \end{claim}
	\begin{proof}
	As illustrated above, the claim is non-trivial because we are encoding two edges at a time, as opposed to a single one. Let's start by observing that both edges to be encoded must be adjacent to each other in the graph spanned by edges traversed in the walk. Moreover, we know the graph has at most $O(q)$ (hyper-)edges.
	
	The crux of the argument is that to specify both edges: it suffices for us to identify one of the edges first, and then identify its neighboring (hyper-)edge. Moreover, we can further consider a partition of hyperedges based on its degree, and encode them via identifying the degree, and then the particular vertex in the corresponding set of vertices with the same degree. Formally, we encode as the following,\begin{enumerate}
		\item Let $A$ be the first edge, and $B$ the second edge;
		\item We will first reveal $d(A)$: this is a number sufficient to decode edge $B$ if $A$ is revealed;
		\item To reveal $A$, we observe that there are at most $O(q)/d(A)$ choices for edge-choice of $A$, by a Markov argument as the total degree is $O(q)$ in a graph of $O(q)$ edges.
				\end{enumerate}		
		 Combining the above, to encode both edges, we incur a total a cost of \[d(A) \cdot O(q)/d(A)= O(q)\] as desired.
 
%
	\end{proof}
%

Next, we consider the following factor assignment scheme that is a refinement from that for the even-case, with the key difference in step-$2$ that we split the factor arising from the vertex in the intersection of both hyper-edges. 
\paragraph{Factor Assignment for Hyper-edegs}
\begin{enumerate}
	\item Let's start by focusing on factors of $n$. For each hyper-edge appearing for the first time, assign each out-going active vertex's label incident to the underlying edge, that is a factor of $\ell^{(|S|-1)/2 } $ to each hyper-edge appearing as an $F$-step;
	\item For each vertex outside the current boundary, assign the vertex's label in $n$ to the underlying edge unless it is in the intersection of two edges. Recall that our Kikuchi matrix for odd-$k$ has at least one vertex outside the boundary that is in the intersection of two hyper-edges, there are at most $k- (|S|-1)/2 - 1$ vertices since each hyper-edge have at least $(|S|-1)/2$ vertices on the current boundary.
	\item  For each vertex in the intersection, split the cost of $n$, and assign to each hyper-edge a cost of $\sqrt{n}$, and crucially notice this factor of $n$ can be replaced by a factor of $k =O_{{|S|}}(1)$ if the other incident edge is not an $F$-edge, and identified  ;
	\item For each hyper-edge that appears at least twice, we assign a total factor at most $p$ through its edge value; 
	\item Finally, the factor of $S|$ is identical to the even-case. It suffices for us to assign the entire factor to each time the walk enters a vertex, and another half of the factor departs from a vertex - that is a factor of at most $|S|^{2(|S|-1)}$. There is no factor of $D$ as we consider fixed $\beta$. \jnote{double check- go thru with Siu On}
	\item Combining the above, we have a combined factor at most \[
		n^{k- (|S|-1)/2 - 1} \cdot \sqrt{n} \cdot \ell^{(|S|-1)/2 } \cdot p \cdot |S|^{(|S|-1)/2}  \,,
	 \]
	 for each hyper-edge if it first appears in an $F\times F$-step, and a factor of \[ 
	 k \cdot  n^{k- (|S|-1)/2 - 1}   \cdot \ell^{(|S|-1)/2 } \cdot p \cdot |S|^{2(|S|-1)} \,.
	 \] 
	 for each hyper-edge if it first appears in an $F\times L$ or $F\times H$-step. 
	\item For each $H\times H$, and $L\times L$-step,  we assign a \emph{single} factor of $O(q)$ . \label{claim: single-q}
	\end{enumerate}
\jnote{double check the D,S}

\paragraph{Analysis of Block-Value Bound for Odd $S$}

Now we are ready to analyze the block-value of each step-combination for our walk.

\begin{lemma}[Block-Value Bound for Odd $S$]
	For any odd $S$, and any $\ell$, it suffices for us to take \[ 
	B(M_{S,\ell})  = O_{k,|S|} ( \frac{m}{n^{|S|/2}} \cdot \sqrt{\ell}^{|S|} \cdot \sqrt{\log n} ) \,.\]
	provided $\ell \leq O(\frac{n}{\log n})$ and $m \geq \Omega_{k}( n^{|S|/2} \cdot \frac{\ell}{\ell^{|S|/2}}   \sqrt{\log n})$, or $\ell =O(n)$ and $m \geq \Omega_{k}( n^{|S|/2} \cdot \frac{\ell}{\ell^{|S|/2}}   {\log n})$. \end{lemma}
	
	\begin{proof}
		
 Again fix an arbitrary traversal direction, say $\rightarrow$, and fix $\ell = C \ell \log n$ for some constant $C>1$. For starters, note that
our block-value is assigned irrespective of the order of the hyperedge-labels, i.e. $B(F\times H) = B(H\times F)$, hence it suffices for us to consider all combinations up to ordering. 
 Note the following,
\begin{enumerate}
	\item For an $F\times F$-step, note that they both become an $L$-step in the reverse traversal. 
	Firstly, notice that we have \[ 
	B_{\rightarrow}(F\times F) = 		(n^{k- (|S|-1)/2 - 1} \cdot \sqrt{n} \cdot \ell^{(|S|-1)/2 } \cdot p \cdot |S|^{2(|S|-1)} 
)^2\,,
	\]
	and by~\cref{claim: single-q},  \[ 
	B_{\leftarrow}(L\times L) = q
	\,.\]
	Combining the above yields the averaged-bound of ,
	 \begin{align*}
		B(F\times F) &\leq \sqrt{B_{\rightarrow}(F\times F) \cdot B_{\leftarrow} (L\times L) }\\
		&= \sqrt{  
		\left(		n^{k- (|S|-1)/2 - 1} \cdot \sqrt{n} \cdot \ell^{(|S|-1)/2 } \cdot p
 \right)^2  \cdot q \cdot |S|^{2(|S|-1)}  }  \\
		&\leq O_{k,|S|} ( \frac{m}{n^{|S|/2}} \cdot \sqrt{\ell}^{|S|} \cdot \sqrt{\log n}  \cdot ) \,.
	\end{align*} 
	By symmetry, the same bound applies for $B(L\times L)$.
	\item For an $H\times H$-step: a factor of $q$ is sufficient to identify both edges, and there are no reassigned factors, hence we have \[ 
	B(H\times H) = q \,,
	\]
	and verify that \[
	q \leq B(F\times F)
	 \]
	 provided \[ m \geq \Omega_{k}(n^{|S|/2} \cdot \frac{\ell }{\ell^{|S|/2}}\cdot \sqrt{\log n} )\,.\]
	\item For an $F\times L$-step, averaging over both direction is the same, we have \[ 
	 B(F\times L) = \sqrt{ B_{\rightarrow} (F\times L) \cdot B_{\leftarrow}(F\times L) }  =  B_{\rightarrow}(F\times L)	\]
	 To bound $B_{\rightarrow}(F\times L)$, observe that the $F$-edge is assigned a factor of $k \cdot n^{k- (|S|-1)/2 - 1}   \cdot \ell^{(|S|-1)/2 } \cdot p$, and the $L$-edge a factor of $2q$, giving a total of 
	\begin{align*}
		B(F\times L)& =  O_{k}(1) \cdot q \cdot n^{k- (|S|-1)/2 - 1}   \cdot \ell^{(|S|-1)/2 } \cdot p\cdot |S|^{2(|S|-1)} 
	\end{align*}
	
	Simplifying, this is at most our target norm bound obtained in the previous two cases of \[ 	 O_{k,|S|} ( \frac{m}{n^{|S|/2}} \cdot \sqrt{\ell}^{|S|} \cdot \sqrt{\log n} ) \,.\]
	provided $\ell \leq O_{|S|}(\frac{n}{\log n})$.
	
	\item For an $F\times H$-step in one direction, notice it becomes an $L\times H$-step in the reverse direction.  Observe that \[ 
	B_{\rightarrow}(F\times  H)  = k\cdot   n^{k- (|S|-1)/2 - 1} \cdot \cdot \ell^{(|S|-1)/2 }\cdot p  |S|^{2(|S|-1)} \cdot |D|^{2(|S|-1)}
	\,,\] 
	and \[ 
	B_{\rightarrow}(L\times  H)   \leq 2q\,.
	\]
	Therefore, we have \[	B(F\times H) \leq \sqrt{B _{\rightarrow} (F\times H)\cdot B_{\leftarrow}(L\times H)  }	\,. \]
	
	and verify \[
	B(F\times H)= o_n(1) \cdot B(F\times F)\,.
		 \]
		 and note that by symmetry, the same bound holds for $B(L\times H)$.
	\end{enumerate}
Combing the above, for any odd $|S|$, we have a block-value bound of \[ 
B(M_{S,\ell} ) \leq O_{k}(1) \cdot B(F\times F) =  O_{k,|S|} ( \frac{m}{n^{|S|/2}} \cdot \sqrt{\ell}^{|S|} \cdot \sqrt{\log n} ) \,.\]
	provided $\ell \leq O(\frac{n}{\log n})$ and $m \geq \Omega_{k}( n^{|S|/2} \cdot \frac{\ell}{\ell^{|S|/2}}   \sqrt{\log n})$, or $\ell =O(n)$ and $m \geq \Omega_{k}( n^{|S|/2} \cdot \frac{\ell}{\ell^{|S|/2}}   {\log n})$. This completes the proof for the odd case. 

	\end{proof}

\subsection{Wrapping Up}
\label{sec:norm-bound-wrap}
We now complete the proof to our main spectral norm bounds from the block-value bounds established above for even and odd $S$.

\SpectralNormBounds*

\begin{proof}
It suffices for us to show the corresponding block-value bounds for $M_{S,\ell}$ in either odd and even case. For even case, from \cref{lem: block-value-bound-even}, we have a bound of \[ 
B(M_{S,\ell}) \leq O_{k,|S|}\left(   \sqrt{\frac{m}{n^{|S|/2}} \cdot \ell ^{|S|/2+1 } \log n}  \right)
\]
provided $m$ satisfies the density constraint from above. 
Similarly, for odd case, we have a general bound of \[ 
B(M_{S,\ell}) \leq O_{k}(1) \cdot B(F\times F) = O_{k,|S|}( \frac{m}{n^{|S|/2}} \cdot \sqrt{\ell}^{|S|} \cdot \sqrt{\log n} ) 
\]
for any regime of $\ell$ provided $m$ satisfies the above density constraint, as well as an improved bound of 
\[ B(M_{S,\ell} ) \leq O_{k}(1) \cdot B(F\times F) =  O_{k,|S|} ( \frac{m}{n^{|S|/2}} \cdot \sqrt{\ell}^{|S|} \cdot \sqrt{\log n} ) \,.\]
	provided $\ell \leq O(\frac{n}{\log n})$ and $m \geq \Omega_{k}( n^{|S|/2} \cdot \frac{\ell}{\ell^{|S|/2}}   \sqrt{\log n})$,

By the set-up of our block-value bounds in \cref{def:block-value-trace}, we obtain a bound of the high-trace, \[ 
\mathbb{E}\!\left[\operatorname{Tr}\!\left( (M\cdot M^T)^{q}\right)\right]
\;\le\;
\mathsf{MatrixDimension} \cdot B_q(M)^{2q}\,.
\]

 Finally,  plugging $q = C\ell \log n$ for some constant $C>1$,  these block value bounds translate to high-probability norm bounds via~~\cref{claim:block-value-to-norm}. This completes our proof of norm bounds.	
\end{proof}

\section{Lower Bounds}\label{sec:lower-bounds}

In this section, we exploit the recent frameworks of \cite{chan2024how,chan2025how} to show LP, SoS, and related lower
bounds for $t$-wise independent CSPs.
Our work is the first to use these frameworks to prove new LP and SoS lower bounds (beyond $t$-wise uniformity).
Thanks to these modular frameworks, our proofs are relatively short.
Such short proofs would be difficult with the monolithic framework of \cite{KMOW17} that has been standard for
$t$-wise uniform CSPs.
Our results do not follow from the techniques from the pseudo-calibration framework \cite{BHKKMP19} used in
many other SoS lower bounds (and conjectured to give a universal recipe).

\subsection{Linear Program} \label{sec:lp-lower}

The first step towards SoS lower bound is LP lower bound.
In this subsection, we prove LP lower bound $t$-wise independent CSPs, generalizing \cite{benabbas12sdp}.

\begin{lemma}
  Let $t \geq 1$.
  Any $t$-wise $\nu$-independent distribution $\mu$ is associated with a probability density function $f$
  with respect to $\nu^k$, so that
  \begin{equation} \label{eq:density}
    \mu (b) = f (b) \nu^k (b) \quad \text{for } b \in D^k.
  \end{equation}
\end{lemma}

\begin{proof}
  Define $f$ to be the Radon–Nikodym derivative $f \coloneqq d \mu / d \nu^k$.
  For $f$ to be well defined, we need $\supp(\mu) \subseteq \supp(\nu^k)$.
  This holds because $\supp(\pi_i (\mu)) = \supp(\nu)$ for every $i \in [k]$, so $\supp(\mu) \subseteq
  \prod_{i \in [k]} \supp(\pi_i (\mu)) = \supp(\nu)^k = \supp(\nu^k)$.
\end{proof}

Since the CSP is $t$-wise independent, \cref{eq:density} implies that for any $R \in \R, Q \in {[k] \choose t}$, the
distribution $\mu_R$ can be expressed as $\mu_R = f_R \nu^k$, and
\begin{equation} \label{eq:proj-density}
  \pi_Q (\mu_R) = \nu^Q \quad \Longleftrightarrow \quad \sum_{\substack{b \in D^k \\ b_Q = a}} f_R (b) \nu^k (b)
  = \nu^Q (a) \quad \text{for } a \in D^Q.
\end{equation}

When we say that a distribution $\mu$ over a finite set $\Sigma$ is proportional to a non-negative function
$g: \Sigma \to \R$, we mean $\mu(a) = g(a) / Z$ for $a \in \Sigma$, where $Z = \sum_{a \in \Sigma} g(a)$.

\begin{lemma} \label{lem:proj-proportional}
  Consider variable subsets $T \subseteq S$.
  If a distribution $\mu$ over $D^S$ is proportional to $g$, then $\pi_T (\mu)$ is proportional to $\pi_T (g)$.
\end{lemma}

\begin{proof}
  Since $\mu$ is proportional to $g$, we have $\mu = g / Z$, where $Z = \sum_{a \in D^S} g(a)$.
  Therefore $\pi_T (\mu) = \pi_T (g / Z) = \pi_T (g) / Z$.
  The desired result follows since
  \[ Z = \sum_{a \in D^S} g(a) = \sum_{b \in D^T} \sum_{\substack{a \in D^S \\ a_T = b}} g(a)
  = \sum_{b \in D^T} \pi_T (g)(b). \qedhere \]
\end{proof}

Previous LP and SDP lower bounds for $t$-wise uniform CSPs are based on the following distribution $\mu_J$ of satisfying
assignments to an instance.
Recall that each constraint $C$ has a $t$-wise uniform distribution $\mu_C$ of satisfying assignments.

\begin{definition}[{\cite[Definition~5.10]{KMOW17}}] \label{def:canonical-uniform}
  Given an instance $J = (V, \cC)$, sample an assignment $b \in D^V$ as follows:
  \begin{itemize}
    \item For each isolated variable $v \in V \setminus V(\cC)$: $\quad b_v \in D^{\set v}$ is independently
      drawn from the uniform distribution over $D$
    \item For each constraint $C \in \cC$: $\quad b_C \in D^{V(C)}$ is independently drawn from $\mu_C$, conditioned
      on agreeing at the common variables shared by multiple constraints
  \end{itemize}
  Then $b \coloneqq \bigcup_{v \in V \setminus V(\cC)} b_v \cup \bigcup_{C \in \cC} b_C$.
  Let $\mu_J$ be the distribution of $b$.
\end{definition}

It is not obvious how to generalize the above definition from $t$-wise uniformity to $t$-wise independence.
Our solution is to first restate the above definition:

\begin{definition}{\cite[Equation~1]{benabbas12sdp}} \label{def:canonical-uniform-alt}
  Let $\mu_J$ be the distribution proportional to $b \in D^V \mapsto \prod_{C \in \cC} \mu_C (b_{V(C)})$.
\end{definition}

\cref{def:canonical-uniform-alt} is the definition that we now generalize to $t$-wise independent CSPs.
Recall that each constraint $C$ has a probability density function $f_C$ with respect to $\nu^{V(C)}$ so that
$\mu_C = f_C \nu^{V(C)}$ is $t$-wise $\nu$-independent.

\begin{definition}
  Given a constraint set $\cC$ over a variable set $V$, let $f_\cC$ be the product of the $f_C$ over its constraints:
  \[ f_\cC (b) \coloneqq \prod_{C \in \cC} f_C (b_{V(C)}) \qquad \text{for } b \in D^V. \]
\end{definition}

\begin{definition} \label{def:canonical}
  Given an instance $J = (V, \cC)$, define $\mu_J$ to be the distribution over satisfying assignments $b \in A_J$,
  proportional to $f_\cC \nu^V$:
  \[ \mu_J (b) \coloneqq \frac{f_\cC (b) \nu^V (b)}{\sum_{a \in D^V} f_\cC (a) \nu^V (a)} \qquad \text{for } b \in D^V. \]
\end{definition}

When $\nu$ is uniform, \cref{def:canonical} yields the same distribution as in \cref{def:canonical-uniform} and
\cref{def:canonical-uniform-alt}.

The next lemma is a key step towards generalizing \cite[Proposition~10.5]{chan2024how} and
\cite[Lemma~6.27]{chan2025how}.

Consider an instance $J = (V, \cC)$ with variable set $V$ and constraint set $\cC$.
Given a constraint $C \in \cC$, let $J\setminus C \coloneqq (V, \cC\setminus \set C)$ be the subinstance of $J$
with $C$ removed from $\cC$.

\begin{lemma} \label{lem:marginal}
  Suppose a $k$-CSP is $t$-wise $\nu$-independent.
  For any instance $J = (V, \cC)$, variable subset $S \subseteq V$, constraint $C \in \cC$ such that
  $\size{V(C) \cap (V(\cC\setminus \set C) \cup S)} \leq t$,
  \[ \pi_S (\mu_J) = \pi_S (\mu_{J\setminus C}). \]
\end{lemma}

\begin{proof}
  $\mu_J$ is proportional to $g_J$, where $g_J (b) \coloneqq f_\cC (b) \nu^V (b)$ for $b \in D^V$.
  Let $U \coloneqq V(C) \setminus (V(\cC\setminus \set C) \cup S), T \coloneqq V \setminus U$.
  For every $a \in D^T$,
  \[ \pi_T (g_J) (a) = \sum_b f_\cC (b) \nu^V (b) \]
  where the sum is over all $b \in D^V$ such that $b_T = a$.
  Then the sum equals
  \[ f_{\cC\setminus \set{C}} (a) \nu^{V\setminus V(C)} (a_{V\setminus V(C)}) \cdot \sum_b f_C (b_{V(C)})
    \nu^{V(C)} (b_{V(C)})
  \overset{\eqref{eq:proj-density}}= f_{\cC\setminus \set C} (a) \nu^T (a). \]
  Likewise,
  \[ \pi_T (g_{J\setminus C}) (a) = \sum_b f_{\cC\setminus \set C} (b) \nu^V (b)
    = f_{\cC\setminus \set C} (a) \nu^T (a) \cdot \underbrace{\sum_{c \in D^U} \nu^U (c)}_{=1} = \pi_T
  (g_J) (a). \]
  \cref{lem:proj-proportional} implies $\pi_T (\mu_{J\setminus C})$ and $\pi_T (\mu_J)$ are proportional to $\pi_T
  (g_{J\setminus C})$ and $\pi_T (g_J)$
  respectively.
  Since $\pi_T (g_{J\setminus C}) = \pi_T (g_J)$, we have $\pi_T (\mu_{J\setminus C}) = \pi_T (\mu_J)$.
  Further projecting to $S$ yields the desired result.
\end{proof}

Our short proof of the previous lemma closely follows \cite[Proposition~10.5]{chan2024how}, which in turn is based on
\cite[Lemma~3.12]{benabbas12sdp}.
There is a different proof of an analogous result \cite[Theorem~5.12]{KMOW17}, but that proof is more cumbersome
to generalize from $t$-wise uniformity to $t$-wise independence.

By the proof of the LP case of \cite[Theorem~1.1]{chan2024how}, but replacing \cite[Proposition~10.5]{chan2024how} with
\cref{lem:marginal}, we get:

\begin{theorem} \label{independent-lp}
  Let $t \geq 2$.
  If a $k$-CSP is $t$-wise independent, then except with probability $o_{n;k}(1)$, a random instance of
  the CSP with $n$ variables and $\Delta n$ constraints has an LP hierarchy solution of
  level $\Omega_k(n/(\Delta^{2/(t-1)} \log \Delta))$.
\end{theorem}

Since this paper focuses on the binomial random model, while \cite{chan2024how} concerns independent hyperedges with
replacement, one also needs to replace \cite[Lemma~8.16]{chan2024how} with \cref{lem:random-expand}.
\cref{lem:random-expand} has the additional assumption that $t$ (or $\tau$) is strictly less than $k$, which is not
present in \cite[Lemma~8.16]{chan2024how}.
However, if $t = k$, the $k$-wise independence assumption implies every constraint of the $k$-CSP is satisfied by every
assignment, and in particular the $k$-CSP is trivially satisfiable, so a level-$n$ LP hierarchy solution exists for a
trivial reason.

\subsection{Sum-of-Squares} \label{sec:sdp-lower}

In this subsection, we prove \cref{thm:independent-sdp,thm:t-wise-value-sdp}.

Using \cref{lem:marginal}, it's possible to prove the SDP lower bound for $t$-wise independent CSPs, using
\cref{lem:marginal} in place of \cite[Lemma~10.4]{chan2024how} in \cite[Lemma~11.11]{chan2024how}.

We now show that \cite[Lemma~11.15]{chan2024how} can be generalized from $t$-wise uniform CSP to $t$-wise
independent CSP.
\cite[Lemma~11.15]{chan2024how} was also used implicitly in \cite[Lemma~6.14]{KMOW17}.

Following \cite[Section~11.1]{chan2024how}, a path in an undirected hypergraph $H$ is a sequence $v_1, \dots, v_\ell$ of
vertices of $H$, so that every two consecutive vertices $v_i$ and $v_{i+1}$ in the sequence both belong to some common
hyperedge in $H$.
Given vertex subsets $S$ and $T$ in $H$, a vertex subset $R$ is an $(S,T)$-separator if every path from some vertex in
$S$ to some vertex in $T$ contains a vertex in $R$.

\begin{lemma} \label{lem:cond-indep}
  Let $S,T$ be subsets of variables of an instance $J = (V, \cC)$, and $R$ be an $(S,T)$-separator in $J$.
  Consider picking a random assignment $b$ from $\mu_J$.
  Then $b_S$ and $b_T$ are conditionally independent given $b_R$.
\end{lemma}

\begin{proof}
  Transform $J$ into a graph $G = (V, E)$ on the same vertex set $V$ by replacing every constraint $C \in \cC$ with a
  clique on $V(C)$.
  That is, $E \coloneqq \Set{ (u,v) \in {V \choose 2} | \set{u,v} \subseteq V(C) \text{ for some } C \in \cC }$.
  A path from in $J$ is equivalent to a path in $G$.
  In particular, $R$ is also an $(S,T)$-separator in $G$.
  The desired conditional independence now follows from the global Markov property (\cref{lem:global-markov}).
\end{proof}

The following lemma is a standard result in probabilistic graphical model:
Any distribution that factorizes according to a graph satisfies the global Markov property.
Global Markov property has not been explicitly used in SoS lower bounds for CSP before, because the canonical
distributions in previous works are simpler and more explicit.
Some sources (e.g. \cite[Proposition~2.40]{lauritzen2019lectures}) only state and prove the global Markov property under
the additional assumption that the vertex subsets are disjoint.
For completeness, we show that the disjointness assumption can be removed.

\begin{lemma} \label{lem:global-markov}
  Let $\cK$ be the collection of cliques in a graph $G = (V, E)$. Suppose a distribution $\mu$ over $D^V$ is
  proportional to $\prod_{K \in \cK} f_K$, where $f_K$ is a $K$-junta for $K \in \cK$.
  For any vertex subsets $S,T \subseteq V$ and $(S,T)$-separator $R \subseteq V$, $\mu_S$ and $\mu_T$ are conditionally
  independent given $\mu_R$.
\end{lemma}

\begin{proof}
  Let $S' \coloneqq S \setminus R, T' \coloneqq T \setminus R$.
  Since $R$ is an $(S,T)$-separator, $S', T'$ and $R$ are disjoint.
  $\mu$ factorizes according to $G$ as in \cite[Definition~2.39]{lauritzen2019lectures}.
  \cite[Proposition~2.40]{lauritzen2019lectures} implies $\mu_{S'}$ and $\mu_{T'}$ are conditionally independent given
  $\mu_R$; see also \cite[Section~2.6.1]{lauritzen2019lectures} for the definition of global Markov property for
  disjoint subsets.

  For $r \in D^R, s \in D^S, t \in D^T$, $\mu_{S|R}(s,r) = \mu_{S'|R}(s_{S'},r) \bracbb{s_{S \cap R} = r_{S \cap R}}$
  and $\mu_{T|R}(t,r) = \mu_{T'|R}(t_{T'},r) \bracbb{t_{T \cap R} = r_{T \cap R}}$.
  Therefore
  \begin{align*}
    \mu_{S|R} (s,r) \mu_{T|R} (t,r) &= \mu_{S'|R} (s_{S'},r) \mu_{T'|R} (t_{T'},r)
    \bracbb{s_{S \cap R} = r_{S \cap R}} \bracbb{t_{T \cap R} = r_{T \cap R}} \\
    &\overset{(*)}= \mu_{S'T'|R} (s_{S'},t_{T'},r) \bracbb{s_{S \cap R} = r_{S \cap R}, t_{T \cap R} = r_{T \cap R}}
    = \mu_{ST|R} (s,t,r),
  \end{align*}
  where $(*)$ is the conditional independence of $\mu_{S'}$ and $\mu_{T'}$ given $\mu_R$.
\end{proof}

By the proof of the SDP case of \cite[Theorem~1.1]{chan2024how}, but replacing
\cite[Lemma~11.15]{chan2024how} with \cref{lem:cond-indep}, we get \cref{thm:independent-sdp}:

\begin{theorem}
  Let $t \geq 2$.
  If a $k$-CSP is $t$-wise independent, then except with probability $o_{n;k}(1)$, a random instance of
  the CSP with $n$ variables and $\Delta n$ constraints has an SoS solution of degree
  $\Omega_k(n/(\Delta^{2/(t-1)} \log \Delta))$.
\end{theorem}

The value of an SoS solution (pseudo-distribution) $\mu$ on an instance $\cI$ is $\pE_\mu \Val_\cI (x)$.
It also straightforward to show that there exists an SoS solution with value nearly $\opt_t (\rho)$, proving
\cref{thm:t-wise-value-sdp}:

\begin{theorem}
  Let $t \geq 2$, and $\rho$ be a distribution over $\cR$.
  Given any $k$-CSP $(D, \cR)$, except with probability $o_n(1)$, a random $\rho$-instance of the CSP with $n$
  variables and $\Delta n$ constraints has an SoS solution of degree $\Omega_k(n/(\Delta^{2/(t-1)} \log \Delta))$ and of
  value $\opt_t (\rho) + o_n(1)$.
\end{theorem}

\begin{proof}
  Let $\Paren{\mu^R}_{R \in \cR}$ be a tuple of $t$-wise independent distributions over $D^k$ that is a maximizer for
  $\opt_t (\rho)$, so that $\opt_t (\rho) = \E_{R \sim \cR} \Pr_{x \sim \mu^R} [x \in R]$.
  The proof of \cite[Theorem~1.1]{chan2024how} (together with \cref{lem:cond-indep}) constructs an SoS solution $\mu$
  to the instance $\cI$, where the distribution of assignments to each constraint $C = (\sigma,R)$ in $\cI$ is the
  $t$-wise $\nu$-independent distribution $\mu^R$.

  Denote by $\tilde\rho$ the empirical distribution of relations $R$ from $\cR$:
  \[ \tilde\rho(R') \coloneqq \Pr_{(\sigma,R) \in \cC} [R = R'] \quad \forall R' \in \cR. \]
  Chernoff and union bounds show that $\Normtv{\rho - \tilde\rho} = o_n(1)$.
  Therefore the SoS solution $\mu$ has value
  \begin{align*}
    \pE_\mu \Val_\cI (x) &= \pE_\mu \E_{(\sigma,R) \in \cC} \bracbb{x_\sigma \in R}
    = \E_{R \sim \tilde\rho} \Pr_{x \sim \mu^R} [x \in R] \\
    &= \E_{x \sim \rho} \Pr_{x \sim \mu^R} [x \in R] + \Normtv{\rho - \tilde\rho}
    = \opt_t (\rho) + o_n(1). \qedhere
  \end{align*}
\end{proof}

\subsection{Combined Hierarchies}

We also get lower bounds to combined hierarchies such as LP+AIP, SDP+AIP, and C(BLP+AIP).
See \cite{chan2024how} and \cite{chan2025how} for their definitions.

To generalize \cite[Lemma~6.27]{chan2025how}, we consider the following distribution that generalizes
\cite[Definition~6.26]{chan2025how}.

\begin{definition}
  Suppose $J = (V, \cC)$ is an instance, $W \subseteq V, q \in A_W$.
  Define $\mu_{J,q}$ to be the distribution of a random assignment $b$ from $\mu_J$, conditioned on $b_W = q$.
\end{definition}

The next lemma generalizes \cite[Lemma~6.27]{chan2025how}.
The proof of \cref{lem:extensible} is a straightforward modification of the proof of \cref{lem:marginal}, because
$\mu_{J,q}$ is proportional to $\bracbb{b_W = q} f_\cC \nu^V$.

\begin{lemma} \label{lem:extensible}
  Suppose a $k$-CSP is $t$-wise independent.
  For any instance $J = (V, \cC)$, variable subsets $W \subseteq S \subseteq V$, satisfying assignment
  $q \in A_W$ to $W$, constraint $C \in \cC$ such that $\size{V(C) \cap (V(\cC\setminus\set C) \cup S)} \leq t$,
  we have $\mu_{J,q}$ is well defined if and only if $\mu_{J\setminus C,q}$ is, and
  \[ \pi_S (\mu_{J,q}) = \pi_S (\mu_{J\setminus C,q}). \]
\end{lemma}

\begin{proof}
  If $\mu_{J,q}$ is well defined, it is proportional to $g_{J,q}$, where $g_{J,q} (b) \coloneqq \bracbb{b_W = q}
  f_\cC (b) \nu^V (b)$ for $b \in D^V$.
  Let $U \coloneqq V(C) \setminus (V(\cC\setminus \set C) \cup S), T \coloneqq V \setminus U$.
  For every $a \in D^T$,
  \[ \pi_T (g_{J,q}) (a) = \bracbb{a_W = q} \sum_b f_\cC (b) \nu^V (b) \]
  where the sum is over all $b \in D^V$ such that $b_T = a$.
  Then the sum equals
  \begin{align*}
    &\bracbb{a_W = q} f_{\cC\setminus \set{C}} (a) \nu^{V\setminus V(C)} (a_{V\setminus V(C)}) \cdot
    \sum_b f_C (b_{V(C)}) \nu^{V(C)} (b_{V(C)}) \\
    \overset{\eqref{eq:proj-density}}= & \bracbb{a_W = q} f_{\cC\setminus \set C} (a) \nu^T (a).
  \end{align*}
  Likewise,
  \begin{align*}
    \pi_T (g_{J\setminus C,q}) (a)
    &= \bracbb{a_W = q} \sum_b f_{\cC\setminus \set C} (b) \nu^V (b) \\
    &= \bracbb{a_W = q} f_{\cC\setminus \set C} (a) \nu^T (a) \cdot \underbrace{\sum_{c \in D^U} \nu^U (c)}_{=1}
    = \pi_T (g_{J,q}) (a).
  \end{align*}
  $\mu_{J\setminus C,q}$ is well-defined iff $\sum_{b \in D^V} g_{J\setminus C,q} (b) > 0$,
  iff $\sum_{a \in D^T} \pi_T (g_{J \setminus C,q}) (a) > 0$,
  iff $\sum_{a \in D^T} \pi_T (g_{J,q}) (a) > 0$,
  iff $\sum_{b \in D^V} g_{J,q} (b) > 0$,
  iff $\mu_{J,q}$ is well-defined.

  When $\mu_{J \setminus C,q}$ or $\mu_{J,q}$ is well-defined,
  \cref{lem:proj-proportional} implies $\pi_T (\mu_{J\setminus C,q})$ and $\pi_T (\mu_{J,q})$ are proportional to $\pi_T
  (g_{J\setminus C,q})$ and $\pi_T (g_{J,q})$ respectively.
  Since $\pi_T (g_{J\setminus C,q}) = \pi_T (g_{J,q})$, we have $\pi_T (\mu_{J\setminus C,q}) = \pi_T (\mu_{J,q})$.
  Further projecting to $S$ yields the desired result.
\end{proof}

By the proof of \cite[Theorem~1.4]{chan2025how}, but replacing \cite[Lemma~6.27]{chan2025how} with
\cref{lem:extensible}, we
get:

\begin{theorem} \label{thm:cblpaip}
  Let $k \geq 3$.
  Suppose a $k$-CSP is pairwise independent and lax.
  Then, for any constraint density $\Delta > 0$, with uniformly positive probability, a random instance of
  the CSP with
  $n$ variables and $\Delta n$ constraints has a solution in the C(BLP+AIP) hierarchy of level $\Omega(n)$.
\end{theorem}

By the proof of \cite[Theorem~1.2]{chan2024how}, but replacing \cite[Proposition~10.5]{chan2024how} with
\cref{lem:extensible}, we get:

\begin{theorem} \label{thm:sdpaip}
  Let $t \geq 2$.
  If a $k$-CSP is $t$-wise independent and $t$-wise neutral for AIP, then except with probability $o_{n;k}(1)$, a
  random instance of the CSP with $n$ variables and $\Delta n$ constraints with replacement has an SDP+AIP hierarchy
  solution of level $\Omega_k(n/(\Delta^{2/(t-1)} \log \Delta))$.
\end{theorem}

\section*{Acknowledgments}
We thank Mitali Bafna, Amey Bhangale, Shuo Pang, Aaron Potechin, Madhur Tulsani, Dmitriy Zhuk, Standa \v{Z}ivn\'{y}
for helpful discussions. We also acknowledge the use of AI tools in assisting with the polishing of the writing.

\clearpage
\newpage
\bibliographystyle{alpha}
\bibliography{bib.bib}
\clearpage
\newpage
\appendix

\section{Deferred Full Algorithms for Spectral Refutation}
\label{sec:full-algorithms}
\subsection{Refutation of Instances with Multiple Predicates}
\label{sec:appendix-mixed-predicates}

\begin{algorithm}[h]
\caption{Spectral Certification Algorithm for Multiple (Boolean) Predicates}
\label{alg:multiple-predicate-certification}
\begin{center}
\begin{minipage}{0.92\linewidth}

\textbf{Input:} A $k$-CSP instance $\calH$ on $n$ Boolean variables with a family of boolean predicates $ \{P_\sigma\}$, and a parameter $\eps$.\\[4pt]
\textbf{Output:} Certified upper bound on the normalized objective value.\\ \\
\textbf{Step 0 (Decomposing into Subinstances of Fixed Predicates).} \\
For any predicate $P_\sigma\in \calP$, let $\calH_\sigma$ denote the instance obtained by restricting to predicate $P_\sigma$ in $\calH$. We write $\calH= \cup_\sigma \calH_{\sigma}$.\\\\
\textbf{Step 1 (Spectral Certification).} \\
Apply Spectral Certification analogous to Step-$1$ from \cref{alg:single-predicate-certification} to each sub-instance of single predicate;
\\\\
\textbf{Step 2 (Bias Enumeration).} \\
For each possible bias $\nu$ (there are $2n$ possibilities in the Boolean domain, or any bias in the discretized net for the general domain $D^k$):
\begin{itemize}
    \item For each predicate $P_\sigma$, if the spectral certification of the corresponding subinstance fails in the previous step, set $\val_t(Q_{\sigma, \nu}) =1  $. 
    \item  Otherwise, solve the corresponding linear program to obtain a dominating polynomial $Q_{\sigma, \nu}(\cdot)$ depending on $\nu$. And set $\val_t(Q_{\sigma, \nu}) \coloneqq \sum_{|S|\leq t } c_S \cdot \nu^{|S|}    $ with an added subscript $\sigma$ to emphasize its dependence on $P_\sigma$.
\end{itemize}

\textbf{Step 3 (Final Certification).} \\
Output
$
\max_{\nu}   \E_{\sigma \sim  E(\calH)} [ \val_t(Q_{\sigma, \nu})]  + \varepsilon
$
as the certified upper bound on the instance’s objective value.


\end{minipage}
\end{center}
\end{algorithm}

\subsection{Full Spectral Refutation Algorithm for General Domain}
We give the full algorithm of spectral refutation for multiple general domain predicates for completeness.
\begin{algorithm}[h]
\caption{Spectral Refutation for Multiple Predicates over a General Domain}
\label{alg:multiple-predicate-general-domain}
\begin{center}
\begin{minipage}{0.94\linewidth}

\textbf{Input:} A $k$-CSP instance $\calH$ over domain $D$ with predicates $\{P_\sigma\}$, parameters $t,\ell$, and $\eps>0$.\\
\textbf{Output:} Certified upper bound on $\max_{x\in D^n}\Psi(x)$, or \textsc{Failure}.\\[4pt]

\textbf{Step 0 (Decomposition).}  
Write $\calH=\bigcup_\sigma \calH_\sigma$, where $\calH_\sigma$ contains constraints with predicate $P_\sigma$.\\[6pt]

\textbf{Step 1 (Spectral Certification).}  
For each $\sigma$, certify all degree-$\le t$ basis polynomials (indexed by $W\in [k]^{\le t},\, b\in D^W$) via their matrices.  
If any bound fails, mark $\sigma$ as \textsc{Bad}.\\[4pt]

\textbf{Step 2 (Marginal Enumeration).}  
Let $\calS$ be an $\eps$-net over $\Delta_D$. For each $\nu\in \calS$ and predicate $\sigma$:
\[
\val_t(Q_{\sigma,\nu}) \coloneqq
\begin{cases}
1, & \text{if }\sigma\text{ is \textsc{Bad}},\\[4pt]
\displaystyle \max_{Q\ge P_\sigma}\ \E_{x\sim \nu}[Q(x)], & \text{otherwise},
\end{cases}
\]
where $Q$ ranges over degree-$t$ polynomials.\\[4pt]

\textbf{Step 3 (Output).}  
Return
\[
\max_{\nu\in \calS}\ \E_{\sigma\sim E(\calH)}\big[\val_t(Q_{\sigma,\nu})\big] + \eps.
\]

\end{minipage}
\end{center}
\end{algorithm}
\section{Deferred Calculations}

\begin{claim}[Squares are negligible]
\label{claim:sq-term-negligible-appendix}
With probability at least \(1-o_n(1)\) over \(\calH\), for any \(y\in [-1,1]^n\),
\[ \sum_{t\in [n]} \sum_{\alpha \in [n]^{|S|-1}} C_S[\alpha, t]^2\, y_\alpha \lesssim_k m\sqrt{\log n}. \]
\end{claim}

\begin{proof}
We prove the general-domain bound; the Boolean case is the special case \(y\in\{\pm1\}^n\).

Since \(\abs{y_\alpha} \le 1\), it suffices to show
\[ \sum_{t\in [n]} \sum_{\alpha \in [n]^{|S|-1}} C_S[\alpha, t]^2 \lesssim_k m\sqrt{\log n} \]
with probability \(1-o_n(1)\).

For any ordered \(|S|\)-tuple \(\Gamma \in [n]^{|S|}\),
\[ C_S[\Gamma] = \sum_{R:\,R_S=\Gamma}(\mathbf{1}[R\in E(\calH)]-p), \qquad p=\tfrac{m}{n^k}, \]
is a sum of \(N_\Gamma=\Theta_k(n^{k-|S|})\) independent mean-zero variables. Hence
\[
\E[C_S[\Gamma]^2]
\le N_\Gamma p
= O_k\!\left(\frac{m}{n^{|S|}}\right).
\]
Summing over all \(\Gamma\in[n]^{|S|}\),
\[
\E\!\left[\sum_{\Gamma} C_S[\Gamma]^2\right] = O_k(m).
\]

Each term \(C_S[\alpha,t]^2\) corresponds to some \(\Gamma\), and each
\(\Gamma\) arises from only \(O_{|S|}(1)\) such decompositions. Therefore
\[ \sum_t \sum_\alpha C_S[\alpha, t]^2 \lesssim_{|S|} \sum_{\Gamma} C_S[\Gamma]^2. \]

Applying Markov’s inequality,
\[
\Pr\!\left[\sum_{t}\sum_\alpha C_S[\alpha,t]^2 \ge Cm\sqrt{\log n}\right]
\lesssim \frac{m}{Cm\sqrt{\log n}}=o_n(1),
\]
so w.h.p.
\[ \sum_{t}\sum_\alpha C_S[\alpha,t]^2 \lesssim_k m\sqrt{\log n}. \]

Multiplying by \(y_\alpha \le 1\) yields the claim.
\end{proof}
\jnote{end new proof}\snote{Thanks for writing it. I've changed the notations.}
\KikuchiClaim*
\begin{proof} 
\label{sec:appendix-verification}

\textbf{Verification for Even $|S|$.}
We start by verifying for the even $|S|$. Recall that our spectral algorithm gives a certification  
\begin{align*}
	|C_S(x)|  \leq  \frac{\|\widetilde{x}^{\otimes \ell} \|_2^2 \cdot \|M_{S,\ell}\|_{sp} }{ \binom{n}{ \ell - |S|/2  } \cdot \binom{|S|}{|S|/2} } 
\end{align*}

Plugging the norm bounds and rearranging with the following observations,
\begin{enumerate}
	\item \[ \|\widetilde{x}^{\otimes \ell} \|_2^2 = \binom{n}{\ell}  \]
	\item \[  \frac{ \binom{n}{\ell} } {\binom{n}{\ell - |S|/2  } }  \leq (n/\ell)^{|S|/2}   \]
	\item Our even $|S|$ norm bound is \[ 
	 \|M_{S,\ell}\|_{sp}\leq  O_{|S|}\left( \sqrt{ \frac{m}{n^{|S|/2}} } \cdot \sqrt{ \ell }^{|S|/2 }  \cdot \sqrt{\ell  \log n}  \right) 
	\]
\end{enumerate}
LHS now becomes\begin{align*}
	 &\frac{\binom{n}{\ell } \cdot   \sqrt{  \frac{m}{n^{|S|/2}} \cdot \ell^{|S|/2+1} \cdot \log n  } }{\binom{n}{\ell - |S|/2  } \cdot \binom{|S|}{|S|/2}}  \leq  O_{|S|}(1) \cdot \frac{1}{\binom{|S|}{|S|/2}} \cdot (\frac{n}{\ell})^{|S|/2} \cdot \ell^{|S|/4 + 1}\cdot  \sqrt{m} \cdot n^{-|S|/4}  \sqrt{\log n}
\end{align*}

Picking the above to be less than $m$ and solve for $m$ gives us \[ 
m\geq \Omega_{|S|}( n^{|S|/2} \cdot \frac{1}{\ell^{|S|/2 - 1 } } \cdot \log n)\,.
\]
\jnote{maybe its ok to keep the even case only because the odd case (boolean) has been removed. keeping for now in case we need anything from here.}

\paragraph{Verification for Odd $|S|$}
Recall that for odd-case, we have\[ 
|C_S(x)| \leq \tilde{C}_S(x)^{1/2} \cdot \sqrt{n}
\]
\jnote{add back the square term}
where $\tilde{C}_S(x)^{1/2}$ is the polynomial obtained from the Cauchy-Schwarz trick. It suffices for us t verify \[ 
\tilde{C}_S(x)^{1/2} \cdot \sqrt{n} = o(m)\,,
\]
In other words, \[ 
\frac{\|\widetilde{x}^{\otimes \ell} \|_2^2 \cdot \|M_{S,\ell}\|_{sp} }{ \binom{n}{ \ell - |S| + 1 } \cdot \binom{|S|-1}{|S|/2}^2 }  = o(\frac{m^2}{n})
\]
Plugging our norm bound for odd $|S|$, we have \[ 
\|M_{S,\ell}\|_{sp} \leq O_{|S|}(1)\frac{m }{n^{|S|/2}} \cdot \ell^{|S|/2}  \sqrt{\log n}
\,.\]
Thus we have \[ O_{S}(1) \cdot 
\frac{\binom{n}{\ell} }{ \binom{n}{\ell-|S|+1} } \cdot \frac{m }{n^{|S|/2}} \cdot \ell^{|S|/2}  \sqrt{\log n} \leq o(\frac{m^2}{n} )
\]
Solving out $m$ , it suffices for us to take  \[m\geq \Omega_{|S|}( n^{|S|/2} \cdot \frac{1}{\ell^{|S|/2 - 1 } } \cdot \sqrt{ \log n} )\,.
 \]
\end{proof}

\section{Deferred Details for Trace Method Calculation}
\paragraph{From Block Value Bound to Norm Bound}

\begin{claim} \label{claim:block-value-to-norm} 
	For an $N\times N$ matrix $M$, let $B_q$ be an valid block-value bound from the assignment scheme, then for any $\eps>0$, we have \[ 
	\|M\|_{sp} \leq (1+\eps) \cdot  B_q
	\]
	with probability at least $1-o_n(1)$.
\end{claim}
\begin{proof}
	Finally, this translates to a matrix norm bound immediately by Markov's inequality.
For any constant $\varepsilon > 0$, we have
\begin{align*}
\Pr\!\left( \|M\|_{\mathrm{sp}} \ge (1+\varepsilon)\, B_q \right)
&\leq 
\frac{\mathbb{E}\!\left[\operatorname{Tr}\!\left( (M M^\top)^q \right)\right]}
{(1+\varepsilon)^{2q} \, B_q^{2q}}
  \\& \;\le\;
\frac{N \cdot B_q^{2q}}{(1+\varepsilon)^{2q} \, B_q^{2q}}
 \\&\leq 
(1+\varepsilon)^{-2q}
c^{-\, q / \log n},	
\end{align*}
for some constant $c>0$. Plugging $q =C\log N=  C \ell  \log n $ for our applications for some large enough constant $C$ completes the proof. 
\end{proof}

\section{Coefficient Bounds for Dominating Polynomials}

\begin{lemma}\label{lem:lp-bounded-dual-norm}
The dual LP in \cref{eq:dual-lp-alphabet} (as well as its specialization to the Boolean domain \cref{eq:dual-lp-boolean})
 has an optimal solution $Q$ with
\[
\|Q\|_1 \;\le\; |D|^{k\bigl(|D|^k+1\bigr)} = O_{|D|, k}(1).
\]
\end{lemma}

\begin{proof}
Treat the linear inequalities of~\cref{eq:dual-lp-alphabet} as a system $Ac \ge y$ with unknown
$c \in \mathbb{R}^{Z}$, where
\[
Z \;\coloneqq\; \bigcup\Bigl\{ D^{W} \,\Big|\, W \in \binom{[k]}{t}\Bigr\}.
\]
Here the matrix $A$ and vector $y$ have rows indexed by $x \in D^k$. Explicitly,
\[
A_{x,b} \;=\; \mathbf{1}[x_W=b]
\qquad\text{and}\qquad
y_x \;=\; P(x)
\quad\text{for } x\in D^k,\; W\in \binom{[k]}{t},\; b\in D^{W}.
\]
Strong duality yields optimal primal and dual solutions $(\mu^\star,c^\star)$.
Let $S \coloneqq \mathrm{supp}(\mu^\star)$. Complementary slackness implies the set of
optimal dual solutions is the polyhedron
\[
\mathcal{R} \;\coloneqq\; \{\,c \mid Ac \ge y,\; A_S c = y_S\,\}.
\]
$\mathcal{R}$ has a minimal face
\[
\mathcal{F} \;=\; \{\,c \mid A_T c = y_T\,\} \subseteq \mathcal{R}
\quad\text{for some row subset } T \subseteq D^k
\]
(see, e.g., \cite[Theorem~8.4]{Sch86}).
Lemma~\ref{lem:01-linear-system} implies there exists a solution $c^\star \in \mathcal{F}$
with $\|c^\star\|_1 \le m^{m+1}$, where $m \coloneqq |T| \le |D|^k$.
This yields an optimal dual polynomial $Q$ with
\[
\|Q\|_1 \;\le\; |D|^k\bigl(|D|^k+1\bigr),
\]
as claimed.
\end{proof}

\begin{lemma} \label{lem:01-linear-system}
Consider a system of linear equations $Ax=b$ in unknown $x\in\mathbb{R}^d$,
where $A\in\{0,1\}^{m\times d}$ and $b\in\{0,1\}^d$ have zero-one entries.
If the system has a solution, then it has a rational solution $x^\star\in\mathbb{Q}^d$
such that $\|x^\star\|_1 \le m^{m+1}$.
\end{lemma}

\begin{proof}
Without loss of generality, assume $A$ has linearly independent rows, and further
$A=[B\;\;C]$ where the square submatrix $B$ of $A$ is nonsingular. Then
\[
x^\star \;\coloneqq\;
\begin{bmatrix}
B^{-1}b\\[2pt]
0
\end{bmatrix}
\]
is a solution to $Ax=b$.
Cramer's rule implies $x^\star_i = \det(B_i)/\det(B)$ for $i\in[m]$,
where $B_i$ is $B$ with its $i$-th column replaced with $b$.
For the numerator, $|\det(B_i)| \le m^m$ by the determinant formula and the fact
that $A$ and $b$ are entry-wise at most $1$ in magnitude.
For the denominator, $|\det(B)|\ge 1$ because $B$ is a nonsingular integer matrix.
Therefore $|x^\star_i| \le m^m$ for $i\in[m]$, hence
\[
\|x^\star\|_1 \;\le\; m\cdot m^m \;=\; m^{m+1}. \qedhere
\]
\end{proof}

\section{Efron--Stein Decomposition} \label{sec:efron-stein}

In this section, we recall the basics of Efron--Stein decomposition.

Given a subset $S \subseteq [k]$, let $Y_S \subseteq \mathbb{R}^{D^k}$ be the vector space of $S$-juntas,
consisting of functions $f: D^k \to \mathbb{R}$ depending only on $S$.
That is, there exists $g: D^S \to \mathbb{R}$ such that $f(a) = g(a_S)$ for $a \in D^k$.

Below when we consider the inner product, orthogonality, and orthogonal projection ``under a distribution
$\mu$'', we mean ``in the real Hilbert space $L^2(\mu)$ of (equivalence classes of square integrable) functions with
respect to $\mu$''.

\begin{definition}[Efron--Stein]
  Let $\nu^k$ be a product distribution over $D^k$.
  Given a subset $S \subseteq [k]$, define $Y_S$ to be the vector space of $S$-juntas $Z_S$ to be the subspace of $Y_S$
  orthogonal to all the subspaces of proper subsets of $S$:
  \[ Z_S \coloneqq Y_S \cap \Span\set{Y_T | T \subsetneq S}^\perp \quad \text{under } \nu^S. \]
  The subspaces $\set{Z_S | S \subseteq [k]}$ are mutually orthogonal under $\nu^k$ and together they span
  $L^2(\nu^k)$.
  The Efron--Stein decomposition is the orthogonal decomposition
  \begin{equation} \label{eq:efron-stein}
    L^2 (\nu^k) = \bigoplus_{S \subseteq [k]} Z_S \quad \text{under } \nu^k.
  \end{equation}
\end{definition}

Details about the Efron–Stein decomposition can be found in \cite[Section~8.3]{o2014analysis}.

Given $f \in L^2(\nu^k)$, denote by $f^S$ the orthogonal projection of $f$ to $Z_S$.
Then \cref{eq:efron-stein} implies for any $f, g \in L^2(\nu^k)$,
\[ f = \sum_{S \subseteq [k]} f^S \quad \text{and} \quad \E_{\nu^k} [f g] = \sum_{Q \subseteq [k]} \E_{\nu^k}
[f^Q g^Q]. \]
\label{eq:efron-stein-ortho}
Given $f \in L^2(\nu^k), S \subseteq [k]$, denote by
\[ f^{\subseteq S} \coloneqq \sum_{T \subseteq S} f^T. \]
Then \cite[Definition~8.17 and Equation~8.2]{o2014analysis} imply that for $a \in D^k$,
\begin{equation}  \label{eq:cond-exp-proj}
  f^{\subseteq S} (a) = \E_{b \sim \nu^{[k]\setminus S}} f(a_S,b).
\end{equation}

Since $f^{\subseteq Q}$ depends only on coordinates in $Q$,
\[ f^{\subseteq Q}(a) = f^{\subseteq Q}\vert_Q (a_Q) \qquad \forall a \in D^k \]
for some (unique) function $f^{\subseteq Q}\vert_Q$ on $D^Q$.

\begin{lemma} \label{lem:proj-density}
  Given a distribution $\mu$ over $D^k$ having density $f$ wrt $\nu^k$, $\pi_Q (\mu)$ has density
  $f^{\subseteq Q}\vert_Q$ wrt $\nu^Q$ for any $Q \subseteq [k]$.
\end{lemma}

\begin{proof}
  Denote by $\overline Q \coloneqq [k] \setminus Q$.
  For $a \in D^Q$,
  \begin{align*}
    \pi_Q (\mu) (a) &= \sum_{b \in D^{\overline Q}} \mu(a,b) = \sum_{b \in D^{\overline Q}} f(a,b) \nu^Q (a)
    \nu^{\overline Q} (b) \\
    &= \E_{b \sim \nu^{\overline Q}} f(a,b) \nu^Q (a) \overset{\eqref{eq:cond-exp-proj}}= f^{\subseteq Q}\vert_Q (a)
    \nu^Q (a). \qedhere
  \end{align*}
\end{proof}

Throughout this paper, $\norm{f}_p \coloneqq (\E_{\nu^k} [\abs{f}^p])^{1/p}$ for any $p \geq 1$, and $\norm{f}
\coloneqq \norm{f}_2$.
Note that $\norm{f}_1$ is different from $\norm{f}_{\ell_1} \coloneqq \sum_{a \in D^k} \abs{f(a)}$.

\section{Simple Probability Results}

\begin{lemma} \label{lem:independent-error}
  For $t \in \N$, let $U_t$ be the uniform distribution over $[n]^{\underline t}$.
  Then $\norm{U_{s+t} - U_s U_t}_{\ell_1} = O(1/n)$ for any $s,t \in \N$.
  In particular, $\norm{U_{2s} - U_s U_s}_{\ell_1} = O(1/n)$ for any $s \in \N$.
\end{lemma}

\begin{proof}
  When $\alpha \in [n]^{\underline s}$ and $\beta \in [n]^{\underline t}$ are disjoint (and there are
    $O(n^{s+t})$ such
  $(\alpha,\beta)$),
  \begin{align*}
    \abs{U_{s+t} (\alpha,\beta) - U_s (\alpha) U_t (\beta)}
    = \frac1{n^{\underline{s+t}}} - \frac1{n^{\underline s} n^{\underline t}}
    = \frac1{n^{\underline s}} \Paren{\frac1{(n-s)^{\underline t}} - \frac1{n^{\underline t}}}
    = O\Paren{\frac1{n^{s+t+1}}}\,.
  \end{align*}
  When $\alpha \in [n]^{\underline s}$ and $\beta \in [n]^{\underline t}$ are not disjoint (and there are
    $O(n^{s+t-1})$
  such $(\alpha,\beta)$),
  \[ \abs{U_{s+t} (\alpha,\beta) - U_s (\alpha) U_t (\beta)}
  = \frac1{n^{\underline s} n^{\underline t}} = O\Paren{\frac1{n^{s+t}}}\,. \]
  Thereore $\norm{U_{s+t} - U_s U_t}_{\ell_1} = O(1/n)$.
\end{proof}

Given a distribution $\mu$ over a finite domain $D$ and a function $f: D \to D'$, denote by $f(\mu)$ the pushforward
distribution over $D'$ such that
\[ f(\mu)(a) \coloneqq \sum_{b \in D, f(b) = a} \mu(b)\,. \]

\begin{lemma} \label{lem:pushforward-dist}
  For any distributions $\mu$ and $\zeta$ over $D$, any function $f: D \to D'$,
  \[ \norm{f(\mu) - f(\zeta)}_{\ell_1} \leq \norm{\mu - \zeta}_{\ell_1}. \]
\end{lemma}

\begin{proof}
  \begin{align*}
    \norm{f(\mu) - f(\zeta)}_{\ell_1} &= \sum_{a \in D'} \abs{f(\mu)(a) - f(\zeta)(a)}
    = \sum_{a \in D'} \Abs{\sum_{b \in D, f(b) = a} f(\mu)(b) - f(\zeta)(b)} \\
    &\leq \sum_{b \in D} \abs{f(\mu)(b) - f(\zeta)(b)} = \norm{\mu - \zeta}_{\ell_1} \,. \qedhere
  \end{align*}
\end{proof}

\begin{lemma} \label{lem:product-close}
  Suppose $\nu_i$ and $\eta_i$ are distributions over a common finite probability space $\Omega_i$, for $1
  \leq i \leq t$.
  Let $\nu \coloneqq \prod_{1 \leq i \leq t} \nu_i$ and $\eta \coloneqq \prod_{1 \leq i \leq t} \eta_i$ be their product
  distributions.
  Then
  \[ \Norm{\nu - \eta}_{\ell_1} \leq \sum_{1 \leq i \leq t} \Norm{\nu_i - \eta_i}_{\ell_1}. \]
\end{lemma}

\begin{proof}
  For $0 \leq i \leq t$, introduce the $i$-th hybrid $\mu^{(i)} \coloneqq \prod_{1 \leq j \leq i} \nu_j \times
  \prod_{i < j \leq t} \eta_j$ between $\nu$ and $\eta$.
  Then $\nu = \mu^{(t)}$ and $\eta = \mu^{(0)}$.
  Then
  \[ \Norm{\nu - \eta}_{\ell_1} = \Norm{\mu^{(t)} - \mu^{(0)}}_{\ell_1}
  \leq \sum_{1 \leq i \leq t} \Norm{\mu^{(i)} - \mu^{(i-1)}}_{\ell_1}. \]
  \cref{lem:independent-close} implies $\Norm{\mu^{(i)} - \mu^{(i-1)}}_{\ell_1} = \Norm{\nu_i - \eta_i}_{\ell_1}$
  for every $1 \leq i \leq t$.
\end{proof}

\begin{lemma} \label{lem:independent-close}
  Given any distributions $\nu$ and $\eta$ over a finite probability space $\Omega$, and any distribution $\mu$
  over a finite probability space $\Gamma$,
  \[ \norm{\nu \times \mu - \eta \times \mu}_{\ell_1} = \norm{\nu - \eta}_{\ell_1}. \]
\end{lemma}

\begin{proof}
  As $a$ runs over $\Omega$, and $b$ runs over $\Gamma$,
  \begin{align*}
    \norm{\nu \times \mu - \eta \times \mu}_{\ell_1} &= \sum_{a,b} \abs{\nu(a)\mu(b) - \eta(a)\mu(b)}
    = \sum_{a,b} \mu(b) \abs{\nu(a) - \eta(a)} \\
    &= \sum_{a} \abs{\nu(a) - \eta(a)} = \norm{\nu - \eta}_{\ell_1}. \qedhere
  \end{align*}
\end{proof}

The next lemma is a sharper version of \cite[Lemma~3]{AD22} and \cite[Lemma~B.3]{chan2024how} by pinning down the value
of the instance.
We do not use \cite[Lemma~B.3]{chan2024how} as is because that lemma concerns independent constraints, while this paper
focuses on the binomial random model.

\begin{lemma} \label{random-unsat}
  Consider a distribution $\rho$ over $D^k$.
  As long as $\Delta \geq \exp(O(|D|^k))$, except with probability $o_{n;|D|,k}(1)$ over a binomial $\rho$-random
  $k$-CSP instance $\cI = (V, \cC)$ with $n$ variables and $\Delta n$ expected constraints,
  \[ \opt(\cI) = \opt_k(\rho) + o_{n;|D|,k}(1)\,. \]
\end{lemma}

\begin{proof}
  Fix any assignment $b \in D^V$.
  In our binomial random model (\cref{def:random-instance}), a random constraint is imposed with probability $p$
  on each \emph{distinct} $k$-subset $e \in {V \choose k}$ (together with a random bijection $\sigma: [k] \to e$).
  Consider an alternative model where a random constraint is imposed with probability $p/k!$ on each (not necessarily
  distinct) $k$-tuple $\tau \in V^k$.
  We argue that $b$ satisfies nearly the same fraction of constraints in both models.

  The number $A$ of constraints satisfied by $b$ in the binomial model has expectation
  \[ p \sum_{e \in {V \choose k}} \E_{\sigma: [k] \to e} \E_{R \sim \rho} \Iv{b_{e \circ \sigma} \in R} \,.
  \]
  The number $B$ of constraints satisfied by $b$ in the alternative model has expectation
  \[ \frac{p}{k!} \sum_{\tau \in V^k} \E_{R \sim \rho} \Iv{b_\tau \in R}\,. \]
  Note that every distinct $k$-tuple $\tau \in V^{\underline k}$ contributes equally to both expectations.
  Therefore
  \[ |\E B - \E A| = \E B - \E A \leq \frac{p}{k!} |V^k \setminus V^{\underline k}| \lesssim_k p n^{k-1}. \]
  Bernstein inequality (\cref{lem:bernstein}) implies that except with probability $\exp(-pn^k)$,
  \[ |\E A - A| \lesssim \sqrt{pn^k} \qquad\text{and}\qquad |\E B - B| \lesssim \sqrt{pn^k}\,. \]
  By the same calculations, the number $X \coloneqq |\cC(\cI)|$ of total constraints in the binomial model and the
  number $Y \coloneqq |\cC(\cI')|$ of total constraints in an instance $\cI'$ from the alternative model also satisfy
  $\abs{\E Y - \E X} \lesssim_k p n^{k-1}$.
  Bernstein inequality also implies that except with probability $\exp(-pn^k)$, $|\E X - X| \lesssim \sqrt{pn^k}$ and
  $|\E Y - Y| \lesssim \sqrt{pn^k}$.
  Therefore except with probability $O(\exp(-pn^k))$,
  \[ \Val_\cI(b) = \frac A X \sim \frac A{p {n \choose k}} \sim \frac{\E A}{p {n \choose k}}\,, \]
  where $P \sim Q$ means $P/Q = 1+o_{n;k}(1)$.
  Likewise
  \[ \Val_{\cI'}(b) = \frac B Y \sim \frac B{p {n \choose k}} \sim \frac{\E B}{p {n \choose k}}\,. \]
  Thus
  \[ \Abs{\Val_\cI(b) - \Val_{\cI'}(b)} = \frac{|(1+o_{n;k}(1)) \E A - (1+o_{n;k}(1)) \E B|}{p {n \choose k}}
  = o_{n;k}(1) \,, \]
  as claimed.

  Now let $\nu$ be the distribution over $D$ of $b(v)$ by sampling a uniformly random $v \in V$.
  Then
  \[ \E B = \frac {p n^k}{k!} \E_{\tau \in V^k} \E_{R \sim \rho} \Iv{b_\tau \in R}
    = \frac{p n^k}{k!} \E_{a \sim \nu^k} \E_{R \sim \rho} \Iv{a \in R} \eqqcolon \frac{p n^k}{k!}
  \Val_\rho(\nu^k) \,. \]
  That is, $\Val_{\cI'}(b) \sim \Val_\rho (\nu^k) \leq \opt_k (\rho)$.
  By a union bound over all $b \in D^k$, whp over $\cI$,
  \[ \opt(\cI) \leq \opt_k (\rho) +o_{n;|D|,k}(1)\,. \]

  For the reverse inequality, consider a distribution $\eta$ over $D$ such that $\Val_\rho(\eta^k) = \opt_k (\rho)$.
  Enumerate the domain $D$ as $[q] = \set{1,\dots,q}$ and the vertex set $V$ as $\set{v_1,\dots,v_n}$.
  Define an assignment $b \in D^V$ from $\eta$ by ``rounding'':
  Let $f: [q] \to \R$ be the cdf of $\eta$, so that $f(i) \coloneqq \sum_{j \leq i} \eta(j)$.
  Then $b(v_i) \coloneqq \arg\min_{j \in [q]} f(j) \geq i/n$.
  Let $\nu$ be the distribution over $D$ of $b(v)$ by sampling a uniformly random $v \in V$.
  Then $|\nu(a) - \eta(a)| = O(1/n)$ for every $a \in D$, and hence $\norm{\nu - \eta}_{\ell_1} = o_{n;|D|}(1)$.
  The above arguments show that whp over $\cI$,
  \begin{align*}
    \opt(\cI) &\geq \Val_\cI(b) \geq \Val_{\cI'}(b) - o_{n;k}(1) \geq \Val_\rho(\nu^k) - o_{n;k}(1) \\
    &\geq \Val_\rho(\eta^k) - o_{n;|D|,k}(1) \,. \qedhere
  \end{align*}
\end{proof}

Note that if a $k$-CSP $(D, \supp(\rho))$ is not trivially satisfiable by a single value (i.e.~there is no $a \in D$
such that $a^k \in R$ for every $R \in \supp(\rho)$), then $\opt_k(\rho) \leq 1-\alpha/|D|^k$, where $\alpha$ is the
minimum nonzero probability of $\rho$.
This is because every $k$-wise independent distribution is of the form $\nu^k$ for some distribution $\nu$ over $D$.
Then $\nu(a) \geq 1/|D|$ for some $a \in D$, and there is $R \in \supp(\rho)$ such that $a^k \not\in R$.
Therefore with probability at least $\rho(R)\nu(a^k) \geq \alpha/|D|^k$, a random constraint from $\rho$ is violated by
a random assignment from $\nu^k$.

The next lemma translates \cite[Lemma~8.16 or B.5]{chan2024how} from the model with replacement to the binomial model.
We include the proof for completeness.
See \cite[Definitions~5.1, 8.1, 8.4]{chan2024how} for the definition of $(t,\gamma)$-expanding
(and also sparsity and $S$-closed, all of which depend implicitly on $\tau$).

\begin{lemma} \label{lem:random-expand}
  Let $2 \leq \tau < k$, $\lambda \coloneqq \tau - 1$, $0 < \gamma \leq \lambda/2$,
  $\zeta \coloneqq \Delta^{2/(\lambda-\gamma)}/n$.
  Except with probability $o_{\zeta;k}(1)$, a random $k$-uniform undirected hypergraph with $n$ vertices
  and $\Delta n$ expected undirected hyperedges $\cH$ from the binomial model is $(t,\gamma)$-expanding, where
  \[ t = \frac n{\Delta^{2/(\lambda-\gamma)}} \cdot \frac 1{2^{O(k)}}. \]
\end{lemma}

\begin{proof}
  \cite[Lemma~8.16]{chan2024how} implies a $k$-uniform undirected hypergraph with $n$ vertices and $\hat\Delta n$
  hyperedges with replacement is $(t,\gamma)$-expanding, where
  $t = \frac{n}{{\hat\Delta}^{2/(\lambda-\gamma)}} \frac 1{2^{O(k)}}$.
  Since $\lambda = \tau-1 \leq k-2$, there is $C_k > 0$ such that whenever $\hat\Delta n \geq C_k n^{k/2}$,
  such $t$ is less than $1$ and the bound is vacuous.

  We want to convert this result to a random hypergraph $\cH$ in the binomial model with $\Delta n$ expected hyperedges.
  Let $m$ be the actual number of hyperedges in $\cH$.
  Chernoff bound implies $\Delta n/2 \leq m \leq 2\Delta n$ whp.
  $\cH$ is a uniformly random undirected hypergraph with exactly $m$ hyperedges.
  Suppose, instead of \emph{without} replacement, sample the $m$ hyperdeges \emph{with} replacement to get $\cH'$.

  Case 1: $\Delta n \leq 2C_k n^{k/2}$.
  Since $\hat\Delta n = m \leq 2\Delta n$, \cite[Lemma~8.16]{chan2024how} implies $\cH'$ is $(t,\gamma)$-expanding whp
  for $t = \frac{n}{{\hat\Delta}^{2/(\lambda-\gamma)}} \frac 1{2^{O(k)}} \gtrsim
  \frac{n}{\Delta^{2/(\lambda-\gamma)}} \frac 1{2^{O(k)}}$.
  Conditioned on the event $\cE$ of no repeated hyperedges, $\cH'$ distributes as $\cH$.
  Since $m \lesssim_k n^{k/2}$, it follows from the asymptotic distribution of the birthday problem
  (e.g.~\cite[Example~3.2.5]{Durrett19}) that $\Pr[\cE] \geq c_k - o_n(1)$ for some $c_k > 0$.
  Then
  \begin{align*}
    \Pr[\cH \text{ not $(t,\gamma)$-expanding}]
    &= \Pr[\cH' \text{ not $(t,\gamma)$-expanding} | \cE] \\
    &\leq \frac{\Pr[\cH' \text{ not $(t,\gamma)$-expanding}]}{c_k - o_n(1)} = o_{n;k}(1) \,,
  \end{align*}
  where the inequality uses the fact that for any events $\cF$ and $\cE$, $\Pr[\cF|\cE] = \Pr[\cF \wedge \cE]/\Pr[\cE]
  \leq \Pr[\cF]/\Pr[\cE]$.
  Therefore $\cH$ is also $(t,\gamma)$-expanding whp.

  Case 2: $\Delta n \geq 2C_k n^{k/2}$.
  Then $\hat\Delta n = m \geq \Delta n/2 \geq C_k n^{k/2}$, so $t = \frac{n}{{\hat\Delta}^{2/(\lambda-\gamma)}} \frac
  1{2^{O(k)}} < 1$ by definition of $C_k$, and every hypergraph is $(t,\gamma)$-expanding (equivalently,
  $(0,\gamma)$-expanding).
\end{proof}

\section{$t$-wise independent CSP}

We collect some easy observations about $t$-wise independent CSPs.

We first show that there are pairwise independent CSPs that are not pairwise uniform (even when restricted to any
subdomain), so our lower bound strictly generalizes those in previous works.
Following \cite{AustrinHastad2011}, we consider $k$-CSPs with only one set $R$ of satisfying assignments, where $R$
is chosen to contain $\Theta(k^2)$ strings from $\set{0,1}^k$ from the $p$-biased distribution.

\begin{theorem}[{Part of \cite[Theorem~3.2]{AustrinHastad2011}}] \label{thm:poly-separator}
  Let $R \subseteq D^k$.
  There is a $t$-wise $\nu$-independent distribution $\eta$ such that $\supp(\eta) \subseteq R$ if and only if there is
  no degree-$t$ polynomial $f \in L^2(D^k, \nu^k)$ such that $f(a) > \E_{\nu^k} [f]$ for every $a \in R$.
\end{theorem}

\begin{theorem} \label{thm:indep-not-uniform}
  Fix any $0 < p < 1/16$.
  Consider the $p$-biased distribution $\nu_p$ over $\set{0,1}$ where $\nu_p (1) = p$.
  Given $s \in \N$, let $R$ be a collection of $s$ random strings over $\set{0,1}^k$ sampled independently
  with replacement from $\nu_p^k$.
  Except with probability $o_k(1)$, there is a $\nu_p$-independent distribution supported on $R$ but not
  any pairwise uniform distribution (even when restricted to subdomains), provided $s = \Theta(k^2)$.
\end{theorem}

\begin{proof}
  \cite[Theorem~6.1]{AustrinHastad2011} implies $R$ is pairwise $\nu_p$-independent except with probability $o_k (1)$.

  Let $\cU = \nu_{1/2}$ be the uniform distribution over $\set{0,1}$.
  To show that $R$ likely does not support any pairwise uniform distribution, thanks to
  \cref{thm:poly-separator} above (from \cite{AustrinHastad2011}), it suffices to construct a quadratic polynomial
  $f$ such that $\E_{\cU^k} [f] = 0$ and $f(a) > 0$ for every $a \in R$.

  Take $f(a) \coloneqq \sum_{1 \leq i < j \leq k} (1 - 2a_i)(1 - 2a_j)$ for $a \in \set{0,1}^k$.
  In other words, $f$ is the sum of all degree-$2$ Fourier characters with respect to $\cU$.
  It follows that $\E_{\cU^k} [f] = 0$.

  Further, Chernoff and union bounds imply that except with probability $o_k (1)$, every $a \in R$ has at
  most $2p k$ many ones.
  It is easy to verify that $f(a) > 0$ for any such $a$ due to the overwhelming fraction of zeros in $a$.
  Indeed, by counting the number of pairs $(i,j)$ such that $(a_i,a_j)$ are both zeros or has exactly a single zero, a
  string $a \in \set{0,1}^k$ containing exactly $r$ ones satisfies
  \[ f(a) \geq \underbrace{\frac{(k-r)(k-r-1)}2}_\text{both zeros}
    - \underbrace{r(k-r)}_\text{single zero}
  \geq (1-o_k (1)) \frac{k-3r}2 (k-r) > 0 \]
  when $r \leq k\/4$.

  Chernoff and union bounds also imply that except with probability $o_k (1)$, every $a \in R$ contains at
  least $p k / 2$ many ones.
  Since every $a \in R$ contains both zeros and ones, $R$ cannot support the pairwise uniform distribution
  over the subdomain $\set{0}$ or the subdomain $\set{1}$.
\end{proof}

We now show that a predicate CSP is $t$-wise uniform if and only if it is $t$-wise independent.
We also show that the ``pairwise uniform'' condition in previous upper bounds is ultimately due to the uniformly random
literals.
Let us first formally define a predicate CSP with literals.

\begin{definition} \label{def:predicate}
  A $k$-CSP $(D, \cR)$ is a predicate CSP with literals if $D$ is an abelian group and there is $Q \subseteq D^k$ such
  that $\cR = \Set{ Q + b | b \in D^k }$.
\end{definition}

The above definition captures as a special case similar definitions in previous works that considered $D = \Z_q$
(\cite[Section~3.1]{austrin2008conditional}, \cite[Section~B.1]{AOW15}).

The next lemma implies that a predicate CSP with literals is $t$-wise uniform if and only if it is $t$-wise independent.

\begin{theorem} \label{thm:predicate-independent-uniform}
  If a predicate CSP is $t$-wise independent, then it is also $t$-wise uniform.
\end{theorem}

\begin{proof}
  Let $Q$ be the set from \cref{def:predicate}.
  Suppose the CSP is $t$-wise $\nu$-independent.

  For $b \in D^k$ and distribution $\mu$ over $D^k$, let $\mu + b$ be the ``shifted'' distribution so that
  $(\mu + b) (a) \coloneqq \mu(a - b)$ for $a \in D^k$.

  For any $b \in D^k$, there is a $t$-wise $\nu$-independent distribution $\mu_b$ with $\supp(\mu_b) \subseteq Q + b$.
  \cref{lem:supp-shift} implies $\supp(\mu_b - b) = \supp(\mu_b) - b \subseteq Q + b - b = Q$.
  Then their uniform average $\overline\mu \coloneqq \E_{b \sim \cU_D} [\mu_b - b]$ is also supported on $Q$.

  For any $T \in {[k] \choose t}$, $a \in D^T$,
  \begin{equation} \label{eq:proj-shift-2}
    \pi_T (\mu_b + b) (a) \overset{\eqref{eq:proj-shift}}= \pi_T (\mu_b)(a - b_T) = \nu^T (a - b_T)
  \end{equation}
  where the last equality is due to $t$-wise $\nu$-independence of $\mu_b$.
  Therefore
  \begin{equation} \label{eq:proj-shift-3}
    \pi_T (\overline\mu) (a) = \E_{b \sim \cU_D} [\pi_T (\mu_b + b) (a)]
    \overset{\eqref{eq:proj-shift-2}}= \E_{b \sim \cU_D} \nu^T (a - b_T) =
    1/\Abs{D^T},
  \end{equation}
  implying $\overline\mu$ is a $t$-wise uniform distribution supported on $Q$.

  Given any $b \in D^k$, $\overline\mu + b$ is supported on $Q + b$ by \cref{lem:supp-shift}, and is also $t$-wise
  uniform because for $a \in D^T$
  \[ \pi_T (\overline\mu + b) (a) \overset{\eqref{eq:proj-shift}}= \pi_T (\overline\mu)(a - b_T)
  \overset{\eqref{eq:proj-shift-3}}= 1/\Abs{D^T}. \qedhere \]
\end{proof}

\begin{lemma} \label{lem:supp-shift}
  Given any distribution $\mu$ over $D^k$ and $b \in D^k$, $\supp(\mu + b) = \supp(\mu) + b$.
\end{lemma}

\begin{proof}
  $a \in \supp(\mu + b)$ iff $0 \neq (\mu + b)(a) = \mu(a - b)$, iff $a - b \in \supp(\mu)$, iff $a \in \supp(\mu) + b$.
\end{proof}

\begin{lemma}
  Given any distribution $\mu$ over $D^k$, $b \in D^k$, $T \in {[k] \choose t}$, $a \in D^T$,
  \begin{equation} \label{eq:proj-shift}
    \pi_T (\mu + b)(a) = \pi_T (\mu)(a - b_T).
  \end{equation}
\end{lemma}

\begin{proof}
  \[ \pi_T (\mu_b + b)(a) = \sum_c (\mu_b + b)(a \cup c) = \sum_c \mu_b ((a \cup c) - b)
  = \sum_c \mu_b ((a - b_T) \cup (c - b_U)) \]
  where the sum is over all $c \in D^U$ and $U \coloneqq [k]\setminus T$.
  As $c$ runs over $D^U$, so does $c' \coloneqq c + b_U$.
  And the sum becomes
  \[ \sum_{c'} \mu_b ((a - b_T) \cup c') = \pi_T (\mu_b)(a - b_T). \qedhere \]
\end{proof}

\section{Concentration} \label{sec:additional-tools}

The main result of this section is concentration of marginals in SoS (\cref{thm:concentration-random-instance}) that
avoids the logarithmic factors in constraint density for even $t$.

\subsection{Matrix $\infty$-to-$1$ norm}

Given an $m$-by-$n$ matrix $M$, consider the following optimization problem, known as the $\infty$-to-$1$
norm $\norm{M}_{\infty\to1}$ of $M$:
\begin{equation} \label{eq:inf-one-norm}
  \sup \, \Set{ x^\top M y | x \in \R^m, y \in \R^n, \norm{x}_\infty \leq 1, \norm{y}_\infty \leq 1 }.
\end{equation}

\begin{remark} \label{rem:inf-one-norm}
  As is well known, \cref{eq:inf-one-norm} has the same value as the \emph{discrete} optimization over $\pm1$ vectors:
  \begin{equation} \label{eq:inf-one-norm-discrete}
    \max \, \Set{ x^\top M y | x \in \set{\pm1}^m, y \in \set{\pm1}^n },
  \end{equation}
  because after fixing any $\overline y \in [-1,1]^n$, the objective function $x^\top M \overline y$ is linear in
  $x \in [-1,1]^m$, and must be maximized at some extreme point $x^* \in \set{\pm1}^m$.
  Now for this fixed $x^*$, the objective function ${x^*}^\top M y$ is linear in $y \in [-1,1]^n$, and
  must be maximized at some extreme point $y^* \in \set{\pm1}^n$.
\end{remark}

It is easy to see that
\begin{equation} \label{eq:inf-one-norm-abs}
  \max_{x,y} \, \Abs{x^\top M y} = \max_{x,y} \, x^\top M y \quad (= \norm{M}_{\infty\to1})
\end{equation}
over $x \in \set{\pm1}^m, y \in \set{\pm1}^n$, because
\[ \max_{x,y} x^\top M y \leq \max_{x,y} \, \Abs{x^\top M y} = \max_{x,y} max \, \set{x^\top M y, (-x^\top) M y}
\leq \max_{x,y} x^\top M y. \]

There is a natural vector program relaxation for \cref{eq:inf-one-norm-discrete}
\begin{align} \label{eq:grothendieck-vp}
  \sup \, & \sum_{i \in [m], j \in [n]} M_{i j} \inner{u_i, v_j} \quad \text{over unit vectors } u_i, v_j \in
  \R^{m+n} \text{ for } i \in [m], j \in [n].
\end{align}

It is well known that the quantities in \cref{eq:inf-one-norm} and \cref{eq:grothendieck-vp} are within a small
multiplicative factor of each other:
\begin{equation} \label{eq:grothendieck}
  \eqref{eq:inf-one-norm} \leq \eqref{eq:grothendieck-vp} \leq K_G \cdot \eqref{eq:inf-one-norm}
\end{equation}
where $1.67 \leq K_G \leq 1.79]$ is the Grothendieck constant.
The right inequality of \cref{eq:grothendieck} is due to Grothendieck inequality.

Consider the $\ell$-by-$\ell$ matrix
$A = \displaystyle{\frac 12 \begin{pmatrix}0 & M \\ M^\top & 0\end{pmatrix}}$, where $\ell = m+n$.
\cref{eq:grothendieck-vp} can be expressed as the SDP:
\begin{align*}
  \max \, &\inner{A, X} \\
  X_{i i} &= 1 \quad \text{for } i \in [\ell] \\
  X &\succeq 0,
\end{align*}
The above SDP has the following dual SDP:
\begin{align*}
  \min \, &\Tr(\Lambda) \\
  \Lambda &\succeq A \\
  \text{diagonal } &\Lambda,
\end{align*}

Slater’s condition implies strong duality holds, so primal and dual optimal values coincide.
The dual SDP corresponds to the SOS program
\begin{equation} \label{eq:grothendieck-sos}
\begin{aligned}
  \min \;& \sum_{i \in [m]} a_i + \sum_{j \in [n]} b_j \\
  \text{s.t.}\;& \text{the polynomial }
  \sum_{i \in [m]} a_i u_i^2 + \sum_{j \in [n]} b_j v_j^2 
  - \sum_{i \in [m], j \in [n]} M_{i j} u_i v_j
  \text{ is SOS.}
\end{aligned}
\end{equation}
\subsection{Concentration of injective norm in SDP}

Given a finite set $W$ (such as $[n]$), a tensor on $W^k$ is an array $T: W^k \to \R$ of real numbers.
Its flattened matrix $M(T)$ is an $W^{\floor{k/2}}$-by-$W^{\ceil{k/2}}$ matrix with entries
$M(T)_{\beta \gamma} \coloneqq T_{\beta \gamma}$ for $\beta \in W^{\floor{k/2}}, \gamma \in W^{\ceil{k/2}}$.

  \begin{lemma}[Bernstein Inequality] \label{lem:bernstein}
  Let $X_1, … X_\ell$ be independent $0$-mean random variables such that $\abs{X_e} \leq B$ for $e \in [\ell]$.
  Then for any $a > 0$,
    \[ \Pr\Brac{\sum_{e \in [\ell]} X_e > a} \leq 
    \exp\Paren{\frac{-\frac12 a^2}{\sum_{e \in [\ell]} \E[X_e^2] + \frac 13 Ba}}. \]
\end{lemma}

The next theorem will be used as a tighter version of \cite[Lemma~6.2]{d2023ihara} without
$\polylog n$. When $\ell = 1$ and $t$ is even, it is also a tighter version of \cite[Theorem~6.2]{AOW15}
without $\polylog n$.
Unlike \cite[Lemma~6.2]{d2023ihara} or \cite[Theorem~6.2]{AOW15}, our proof works exclusively with
$\infty$-to-$1$ norm without considering spectral norm, and was sketched in the introduction of \cite{{d2023ihara}}.

\begin{theorem} \label{thm:tensor-norm-bound}
  Let $(T^1, \cdots, T^\ell)$ be independent random tensors on $[n]^t$ such that for every $e \in [\ell]$,
  \[ \E[T^e] = 0, \quad \E\Brac{\norm{M(T^e)}^2_{\infty\to1}} \leq p, \quad \norm{M(T^e)}_{\infty\to1} \leq B. \]
  For $t \geq 1, p \ell \geq n^{\ceil{t/2}}$, except with probability $o_n (1)$,
  \[ \Norm{M\Paren{\sum_{e \in [\ell]} T^e}}_{\infty\to1} \lesssim B \sqrt{p \ell n^{\ceil{t/2}}}. \]
\end{theorem}

\begin{proof}
  Fix any $y: [n]^{\floor{t/2}} \to \set{\pm1}$ and $z: [n]^{\ceil{t/2}} \to \set{\pm1}$.
  Consider
  \begin{equation} \label{eq:sum-tensor-entries}
    \sum_{\beta \in [n]^{\floor{t/2}}} \sum_{\gamma \in [n]^{\ceil{t/2}}} \sum_{e \in [\ell]}
    T^e_{\beta \gamma} y_\beta z_\gamma = \sum_{e \in [\ell]} y^\top M(T^e) z .
  \end{equation}
  For $e \in [\ell]$, let $X_e \coloneqq y^\top M(T^e) z$. Then
  \begin{align*}
    \E X_e &= y^\top \E \Brac{M(T^e)} z = 0,\\
    \E[X_e^2] &= \E\Brac{\Paren{y^\top M(T^e) z}^2} \overset{\eqref{eq:inf-one-norm-abs}}\leq
    \E\Brac{\Norm{M(T^e)}^2_{\infty\to1}} \leq p,\\
    \Abs{X_e} &\overset{\eqref{eq:inf-one-norm-abs}}\leq \Norm{M(T^e)}_{\infty\to1} \leq B.
  \end{align*}
  Provided $p \ell \geq n^{\ceil{t/2}}$, except with probability at most $\exp\Paren{-n^{\ceil{t/2}}}$ over
  $(T^1, \cdots, T^\ell)$,
  Bernstein inequality implies \cref{eq:sum-tensor-entries} is $O \Paren{B \sqrt{p \ell n^{\ceil{t/2}}}}$.

  By \cref{rem:inf-one-norm}, a union bound over the $2^{n^{\floor{t/2}}} \times 2^{n^{\ceil{t/2}}}$ choices of $y$
  and $z$ implies, except with probability $o_n (1)$,
  \[ \Norm{M\Paren{\sum_{e \in [\ell]} T^e}}_{\infty\to1} \lesssim B \sqrt{p \ell n^{\ceil{t/2}}}\,. \qedhere \]
\end{proof}

Most importantly, upper bounds on $\norm{M(T)}_{\infty\to1}$ imply SOS upper bounds on tensor norm:

\begin{theorem} \label{thm:tensor-sos-bound}
  Given a tensor $T$ on $W^t$, there is a sum-of-squares proof of
  \[ \set{x_i^2 \leq 1}_{i \in W} \quad \vdash_{2\ceil{t/2}} \quad
  \sum_{\alpha \in W^t} T_\alpha x^\alpha \leq K_G \Norm{M(T)}_{\infty\to1}. \]
\end{theorem}

\begin{proof}
  Grothendieck inequality in \cref{eq:grothendieck} implies the value of \cref{eq:grothendieck-vp} is at most $K_G
  \norm{M(T)}_{\infty\to1}$.
  \cref{eq:grothendieck-sos} yields an identity of degree-$2$ polynomials (in indeterminates ${u_\beta, v_\gamma}$)
  \[ \sum_\beta a_\beta u_\beta^2 + \sum_\gamma b_\gamma v_\gamma^2
  - \sum_{\beta,\gamma} T_{\beta \gamma} u_\beta v_\gamma = \sum_j h_j^2(u_\beta, v_\gamma),\]
  where $\beta$ runs over $W^{\floor{t/2}}$ and $\gamma$ over $W^{\ceil{t/2}}$, and $h_j$ are degree-$1$ polynomials.
  We have $\sum_\beta a_\beta + \sum_\gamma b_\gamma \leq K_G \norm{M(T)}_{\infty\to1}$, thanks to strong duality of
  the SDP. 
  Replacing indeterminates ${u_\beta,v_\gamma}$ by monomials ${x^\beta,x^\gamma}$, we get an identity of the
  degree-$2\ceil{t/2}$ polynomials (in indeterminates ${x_i}$)
  \[
    \sum_\beta a_\beta x^{2\beta} + \sum_\gamma b_\gamma x^{2\gamma}
    - \sum_{\beta,\gamma} T_{\beta \gamma} x^{\beta \gamma} = \sum_j h_j^2(x^\beta, x^\gamma),
  \]
  implying
  \[
    \set{x^{2\beta} \leq 1, x^{2\gamma} \leq 1} \quad \vdash_{2\ceil{t/2}} \quad
    \sum_{\alpha \in W^t} T_\alpha x^\alpha \leq \sum_\beta a_\beta + \sum_\gamma b_\gamma
    \leq K_G \norm{M(T)}_{\infty\to1}.
  \]
  The axioms $\set{x_i^2 \leq 1}$ yield degree-$2\abs{\gamma}$ SOS proofs of ${x^{2 \gamma} \leq 1}$ for any
  $\gamma \in W^\ell$ by induction on $\abs{\gamma}$.
  The theorem now follows.
\end{proof}

\subsection{Functions on non-boolean domain} \label{sec:func-non-boolean}

We encode an assignment $x \in D^V$ in an SOS proof by indeterminates $x_{v,a}$ for $v \in V, a \in D$.
Each $x_{v,a}$ intents to represent the $\set{0,1}$-indicator of whether $x(v) = a$.
This representation is equivalent to the one in \cite[Section~B.2]{AOW15}.

As in \cite[Section~B.2]{AOW15}, given a finite set $W$,
there is a bijection $\phi_W$ from $D^W$ to
$\Omega_W \coloneqq \set{z \in \set{0,1}^{W \times D} | \sum_{a \in D} z(i,a) = 1 \, \forall i \in W}$:
\[ \phi_W (b)_{i,a} \coloneqq \bracbb{b_i = a} \quad \text{for } i \in W, a \in D, b \in D^W. \]

Every $f: D^W \to \R$ has a boolean version $f^\Omega: \Omega_W \to \R$, such that
$f(b) = f^\Omega (\phi_W (b))$ for $b \in D^W$, by defining
\begin{equation} \label{eq:boolean-version}
  f^\Omega (z) \coloneqq \sum_{b \in D^W} f(b) \prod_{i \in W} z(i,b_i) \quad\text{for } z \in \Omega.
\end{equation}
Conversely, given $f^\Omega: \Omega_W \to \R$, one can recover $f: D^W \to \R$ simply by defining
$f(b) \coloneqq f^\Omega (\phi_W (b))$ for $b \in D^W$.
For brevity, write $z^b \coloneqq \prod_{i \in Q} z(i,b_i)$ for $Q \subseteq W, b \in D^Q, z \in \R^{W \times D}$.

We now define functions of a given pure degree. The definitions for $f: D^W \to \mathbb{R}$ and
$f^\Omega: \Omega_W \to \R$ are different.

\begin{definition}
  $f: D^W \to \R$ has pure degree $t$ if $f$ is a linear combination of $\1_b$ over $b \in D^S, S \in {W \choose t}$.
\end{definition}

\begin{definition}
  $f^\Omega: \Omega_W \to \R$ has pure degree $t$ if $f^\Omega$ is a linear combination of $z^b$ over
  $S \in {W \choose t}, b \in D^S$.
\end{definition}

Further define a function of degree at most $t$ to be a linear combination of functions of pure degree at most $t$. This
definition is equivalent to the Fourier analytic one in \cite[Appendix~B]{AOW15}, but is much shorter and will
simplify our proofs.

\begin{lemma}[{Essentially \cite[Claim~B.3]{AOW15}}]
  $f: D^W \to \R$ has pure degree $t$ if and only if $f^\Omega: \Omega_W \to \R$ does.
\end{lemma}

\begin{proof}
  Let $\overline S \coloneqq W \setminus S$.
  For every $S \in {W \choose t}, b \in D^S, z \in \Omega_W$,
  \[
    \1_b^\Omega (z) = \sum_{a \in D^W} \1_b (a) \prod_{i \in W} z(i, a_i) = \prod_{i \in S} z(i, b_i) \cdot
    \prod_{j \in \overline S} \Paren{\sum_{a \in D} z(j, a)} .
  \]
  Each factor $\sum_{a \in D} z(i,a)$ in the brackets equals $1$ because $z \in \Omega_W$.
  Therefore $\1_b^\Omega (z) = \prod_{i \in S} z(i, b_i) = z^b$.

  Suppose $f = \sum_b \eta_b \1_b$ where $\eta_b \in \mathbb{R}$ for $S \in {W \choose t}, b \in D^S$.
  Since $f^\Omega$ is linear in $f$ in \cref{eq:boolean-version},
  $f^\Omega = \sum_b \eta_b \1_b^\Omega = \sum_b \eta_b z^b$.

  Conversely, for $S \in {W \choose t}, b \in D^S$, if $g^\Omega (z) = z^b$ for $z \in \Omega_W$,
  then $g(a) = g^\Omega (\phi_W (a)) = \prod_{i \in S} \phi_W (a)_{i,b_i} = \1_b (a)$ for $a \in D^W$.

  Consequently, if $f^\Omega = \sum_b \eta_b z^b$, where $\eta_b \in \R$ for
  $S \in {W \choose t}, b \in D^S$, then $f = \sum_b \eta_b \1_b$.
\end{proof}

\subsection{Tensor representation of functions} \label{sec:tensor-func}

Given $n, k \in \N$, recall that $n^{\underline k} \coloneqq n(n-1)\cdots(n-k+1)$ denote the falling factorial.
Given any set $W$, recall that
$W^{\underline k} \coloneqq \set{ \alpha \in W^k | \alpha_i \neq \alpha_j \text{ for distinct } i,j \in [k] }$
denote the set of $k$-tuples of distinct elements from $W$.

Given $f^\Omega: \Omega_k \to \R$ of pure degree $t$, it is the sum $\sum_b \eta_b z^b$ over
$Q \in {[k] \choose t}, b \in D^Q$ for some $\eta_b \in \R$.

Such $f^\Omega$ can be represented as a tensor $T$ on $([k] \times D)^{\underline t}$ such that
$\sum_\alpha T_\alpha z^\alpha = f^\Omega (z)$ for $z \in \Omega_k$, where $\alpha$ runs over
$([k] \times D)^{\underline t}$.
Indeed, $T \coloneqq \sum_b \eta_b \1_{\tau(b)}$, where $\1_\tau: ([k] \times D)^{\underline t} \to \set{0,1}$ is the
indicator function of $\tau \in ([k] \times D)^{\underline t}$, and $\tau(b) \coloneqq ((i,b_i) | i \in Q)$.
The elements in the $t$-tuple $\tau(b)$ can be ordered arbitrarily.
Then
\[
  \sum_\alpha T_\alpha z^\alpha = \sum_\alpha \sum_b \eta_b \1_{\tau(b)} (\alpha) z^\alpha
  = \sum_b \eta_b z^\tau(b) = \sum_b \eta_b z^b = f^\Omega (z).
\]

Further consider the scope $S \in V^{\underline k}$ of some constraint in an instance.
We want to express as a tensor the function $f^\Omega (z_S)$ by applying $f^\Omega$ to $x_S$.

Every $S \in V^{\underline k}$ equals $(\sigma(1), \dots, \sigma(k))$ for some injective $\sigma: [k] \to V$.
$\sigma$ induces an injective map $[k] \times D \to V \times D$, namely $\sigma(i,a) \coloneqq (\sigma(i), a)$ for
$(i,a) \in [k] \times D$.
Likewise every $z \in \Omega_V$ has restriction $z_S \in \Omega_S$ given by
$z_S (v, a) \coloneqq z(v, a)$ for $(v, a) \in S \times D$.

We will also write $f^{\Omega,\sigma}$ to be the function
$f^{\Omega,\sigma} (z) \coloneqq \sum_b \eta_b z^{b \circ \sigma^{-1}\vert_S}$, so that
$f^{\Omega,\sigma} (z) = f^\Omega (z_S)$ for $z \in \Omega_V$.

Let $T^\sigma \coloneqq \sum_\beta T_\beta \1_{\sigma(\beta)}$ be the tensor on $(V \times D)^{\underline t}$ representing
the image of $T$ under $\sigma$.
We now argue that
\begin{equation} \label{eq:tensor-func}
  \sum_\alpha T^\sigma_\alpha z^\alpha = f^\Omega (z_S) \quad \text{for } z \in \Omega_V.
\end{equation}
To see this, as $\alpha$ runs over $(V \times D)^{\underline t}$ and $\beta$ runs over
$([k] \times D)^{\underline t}$,
\[
  \sum_\alpha T^\sigma_\alpha z^\alpha
  = \sum_\alpha \sum_\beta T_\beta \1_{\sigma(\beta)} (\alpha) z^\alpha = \sum_\beta T_\beta z^{\sigma(\beta)}.
\]
The right-hand-side further equals
\[
  \sum_\beta \sum_b \eta_b \1_{\tau(b)} (\beta) z^{\sigma(\beta)}
  = \sum_b \eta_b z^{\sigma(\tau(b))} \overset{(*)}= \sum_b \eta_b z_S^b = f^\Omega (z_S),
\]
where $b$ runs over $D^Q$ as $Q$ runs over ${[k] \choose t}$, and $(*)$ follows from
\[ z^{\sigma(\tau(b))} = \prod_{i \in Q} z(\sigma(i),b_i) = \prod_{i \in Q} z_S (\sigma(i),b_i) = z_S^b. \]

\subsection{Tensor concentration in SDP}

Let $T$ be a tensor on $([k] \times D)^{\underline t}$.
A non-empty instance $I = (V, \cC)$ induces the tensor
\[ T^\cC \coloneqq \E_{(\sigma,R) \in \cC} T^\sigma. \]
Note that $T$ does not depend on the relation $R$ of the constraint.
On the other hand, let
\[ \overline T \coloneqq \E_{\sigma \in V^{\underline k}} T^\sigma \]
be the expected tensor contributed by a single potential hyperedge.
Here and below $\sigma \in V^{\underline k}$ sometimes mean $\sigma$ is an injective map $[k] \to V$ by abuse of
notation, since such an injective map is equivalent to a $k$-tuple of distinct elements from $V$.

The next theorem shows that in a random instance, $T^\cC$ will be close to its expectation $\overline T$ when measured
by the $\infty$-to-$1$ norm of their flattened matrices.
We prove two versions of the theorem, one for sampling hyperedges from the binomial model
(\cref{thm:avg-tensor-norm-bound}), another for sampling hyperedges independently with replacement
(\cref{thm:avg-tensor-norm-bound-replacement}).

\begin{theorem} \label{thm:avg-tensor-norm-bound}
  Fix $t \leq k$.
  Let $T$ be a tensor on $([k] \times D)^{\underline t}$.
  For any $\eps > 0$, there is $C = C(T,\eps)$ such that for $p \geq C n^{\ceil{t/2} - k}$, except with probability
  $o_n (1)$, a random $k$-CSP instance $I = (V, \mathcal{C})$ on $n$ variables with constraint probability $p$ satisfies
  \[ \Norm{M(T^\cC - \overline T)}_{\infty\to1} \leq \eps. \]
\end{theorem}

\begin{proof}
  For each potential hyperedge $e \in {V \choose k}$, define a tensor $T^e$ on $(V \times D)^{\underline t}$ from $I$
  as follows:
  \begin{itemize}
    \item If $I$ has no constraint on $e$, then $T^e = 0$.
    \item If the constraint on $e$ has scope $\sigma: [k] \to e$, then $T^e = T^\sigma$.
  \end{itemize}

  Define the tensor $\tilde T^e \coloneqq T^e - \E T^e$, where the expectation is over the bijection $\sigma: [k] \to e$.
  We will apply \cref{thm:tensor-norm-bound} to the tensors $(\tilde T^e)$ over $e \in {V \choose k}$.

  Let $B \coloneqq 2 \norm{M(T)}$, where in this proof $\norm{\cdot} \coloneqq \norm{\cdot}_{\infty\to1}$.
  It follows from \cref{claim:norm-bounds} that, over the randomness of $I$,
  \[ \E \tilde T^e = 0, \quad \E\Brac{\norm{M(\tilde T^e)}^2} \leq p B^2, \quad \norm{M(\tilde T^e)} \leq B. \]

  Let $q \coloneqq \abs{D}$. \cref{thm:tensor-norm-bound} implies except with probability $o_n (1)$,
  \[ \Norm{M\Paren{\sum_e \tilde T^e}} \lesssim B^2 \sqrt{p {n \choose k} (n q)^{\ceil{t/2}}}. \]
  We have
  \begin{equation} \label{eq:sum-t-tilde}
    \sum_{e \in {V \choose k}} \tilde T^e = \sum_{(\sigma,R) \in \cC} T^\sigma -
    p \sum_{e \in {V \choose k}} \E_{\sigma: [k] \to e} T^\sigma.
  \end{equation}

  By Chernoff, with probability $1-o_n (1)$, the number $m$ of constraints in $I$ satisfies
  $m = (1+o_n (1)) p {n \choose k}$.
  Then
  \begin{align*}
    T^\cC &= \frac1m \sum_{(\sigma,R) \in \cC} T^\sigma
    \overset{\eqref{eq:sum-t-tilde}}= \frac1m \sum_{e \in {V \choose k}}
    \tilde T^e + \frac pm \sum_{e \in {V \choose k}} \E_{\sigma: [k] \to e} T^\sigma
    = (1+o_n (1)) \Paren{\sum_{e \in {V \choose k}} \frac{\tilde T^e}{p {n \choose k}} + \overline T}.
  \end{align*}
  We now bound each of the two extra terms in $\norm{M(T^\cC - \overline T)}$ by $\eps/2$.
  We have
  \[
    \Norm{o_n (1) \overline T} \leq o_n (1) \E_{\sigma \in V^{\underline k}} \norm{M(T^\sigma)}
    \overset{(*)}\leq o_n (1) B \leq \eps/2
  \]
  whenever $n$ is sufficiently large depending on $B$ and $\eps$, where $(*)$ is \cref{claim:norm-embed}.
  Also
  \[
    \Norm{M\Paren{\frac{1+o_n (1)}{p {n \choose k}} \sum_{e \in {V \choose k}} \tilde T^e}}
    \lesssim B^2 \sqrt{\frac{(n q)^{\ceil{t/2}}}{p {n \choose k}}} \leq \eps/2
  \]
  wherever $n$ and sufficiently large and $p \geq C'(k,q,\eps) n^{\ceil{t/2}-k}$ for some function $C'(k,q,\eps)$.
\end{proof}

\begin{claim} \label{claim:norm-bounds}
  In the context of \cref{thm:avg-tensor-norm-bound}, $\norm{M(T^e)} \leq B/2$ and $\norm{M(\tilde T^e)} \leq B$.
  Further, $\E\Brac{\norm{M(\tilde T^e)}^2} \leq p B^2$.
\end{claim}

\begin{proof}
  For any $e \in {V \choose k}$ and any bijection $\sigma: [k] \to e$, if $e$ exists in $I$, then
  \cref{claim:norm-embed} implies
  $\norm{M(T^e)} \leq \norm{M(T)} \leq B/2$.
  If $e$ does not exist in $I$, then $\norm{M(T^e)} = 0 \leq \norm{M(T)}$.
  Therefore
  \[ \norm{M(\E T^e)} = \norm{\E M(T^e)} \leq \E \norm{M(T^e)} \leq \norm{M(T)} \]
  by linearity of $M(\cdot)$ and convexity of the norm.
  Thus triangle inequality of the norm implies
  \[ \norm{M(\tilde T^e)} = \norm{M(T^e - \E T^e)} \leq \norm{M(T^e)} + \norm{M(\E T^e)} \leq B. \]

  For the ``Further'' part, linearity of $M(\cdot)$ implies
  \[ \norm{M(\tilde T^e)}^2 = \norm{M(T^e) - \E M(T^e)}^2 \leq 2\norm{M(T^e)}^2 + 2\norm{\E M(T^e)}^2, \]
  so
  \[ \E\Brac{\norm{M(\tilde T^e)}^2} \leq 2\E\Brac{\norm{M(T^e)}^2} + 2\norm{\E M(T^e)}^2
  \leq 4\E\Brac{\norm{M(T^e)}^2}, \]
  where the last inequality is due to convexity of the norm and then Cauchy–Schwarz.
  Since $\Pr [T^e \neq 0] \leq p$,
  $ \E\Brac{\norm{M(T^e)}^2} \leq p \sup_{T^e} \, {\norm{M(T^e)}^2} \leq p B^2/4. $
\end{proof}

\begin{claim} \label{claim:norm-embed}
  Given any finite sets $W$ and $U$, tensor $T$ on $W^t$, injective $\sigma: W \to U$, we have $\norm{M(T^\sigma)} = 
  \norm{M(T)}$.
\end{claim}

\begin{proof}
  For $\ell \in \N$ and $x: U^\ell \to \R$, define $\sigma^*(x): W^\ell \to \R$ by
  $\sigma^*(x)_\beta \coloneqq x_\sigma(\beta)$ for $\beta \in W^\ell$.

  Then for any $x: U^{\floor{t/2}} \to \set{\pm1}, y: U^{\ceil{t/2}} \to \set{\pm1}$, as we sum over
  $\beta \in W^{\floor{t/2}}, \gamma \in W^{\ceil{t/2}}$,
  \[
    x^\top T^\sigma y = \sum_{\beta,\gamma} T^\sigma_{\sigma(\beta) \sigma(\gamma)} x_{\sigma(\beta)} y_{\sigma(\gamma)}
    = \sum_{\beta,\gamma} T_{\beta \gamma} \sigma^*(x)_\beta \sigma^*(y)_\gamma = \sigma^*(x)^\top T \phi(y).
  \]
  Therefore $x^\top T^\sigma y = \sigma^*(x)^\top T \phi^*(y) \leq \norm{M(T)}$.
  Taking the maximum over such $x$ and $y$ implies $\norm{M(T^\sigma)} \leq \norm{M(T)}$.

  Conversely, given $\ell \in \N$ and $x: W^\ell \to \R$, define $\sigma(x): U^\ell \to \R$ by
  $\sigma(x)_\beta \coloneqq x_\alpha$ if $\beta = \sigma(\alpha)$ for some $\alpha \in W^\ell$ (such an $\alpha$
  is unique since $\sigma$ is injective), and $\sigma(x)_\beta \coloneqq 0$ otherwise.

  Then for any $x: W^{\floor{t/2}} \to \set{\pm1}, y: W^{\ceil{t/2}} \to \set{\pm1}$, as we sum over
  $\beta \in W^{\floor{t/2}}, \gamma \in W^{\ceil{t/2}}$,
  \[ x^\top T y = \sum_{\beta,\gamma} T_{\beta \gamma} x_\beta y_\gamma
  = \sum_{\beta,\gamma} T^\sigma_{\sigma(\beta) \sigma(\gamma)} \sigma(x)_{\sigma(\beta)} \sigma(y)_{\sigma(\gamma)}. \]
  Using the fact that $T^\sigma_{\beta' \gamma'} = 0$ unless $\beta' \gamma' = \sigma(\beta) \sigma(\gamma)$ for some
  $\beta \gamma$, the right-hand-side equals
  \[ \sum_{\beta' \in U^{\floor{t/2}}} \sum_{\gamma' \in U^{\ceil{t/2}}} T^\sigma_{\beta' \gamma'}
  \sigma(x)_{\beta'} \sigma(y)_{\gamma'} = \sigma(x)^\top T^\sigma \sigma(y). \]
  Therefore $x^\top T y = \sigma(x)^\top T^\sigma \sigma(y) \leq \norm{M(T^\sigma)}$.
  Taking the maximum over such $x$ and $y$ implies $\norm{M(T)} \leq \norm{M(T^\sigma)}$.
\end{proof}

The next theorem is the analog of \cref{thm:avg-tensor-norm-bound} for random constraints with replacement.

\begin{theorem} \label{thm:avg-tensor-norm-bound-replacement}
  Fix $t \leq k$.
  Let $T$ be a tensor on $([k] \times D)^{\underline t}$.
  For any $\eps > 0$, there is $C = C(T,\eps)$ such that for $m \geq C n^{\ceil{t/2}}$, except with probability
  $o_n (1)$, a random $k$-CSP instance $I = (V, \mathcal{C})$ on $n$ variables and $m$ constraints with replacement
  satisfies
  \[ \Norm{M(T^\cC - \overline T)}_{\infty\to1} \leq \eps. \]
\end{theorem}

\begin{proof}
  For $j \in [m]$, define $T^j$ to be the randomly shifted tensor $T^\sigma$, where $\sigma: [k] \to V$ is a
  uniformly random injective map.
  Also define the tensor $\tilde T^j \coloneqq T^j - \E T^j$, where the expectation is over $\sigma$.
  We will apply \cref{thm:tensor-norm-bound} to the sequence $(\tilde T^j)$ of tensors over $j \in [m]$.

  Let $B \coloneqq 2 \norm{M(T)}$, where in this proof $\norm{\cdot} \coloneqq \norm{\cdot}_{\infty\to1}$.
  It follows from \cref{claim:norm-bounds} that, over the randomness of $I$,
  \[ \E \tilde T^e = 0, \quad \E\Brac{\norm{M(\tilde T^e)}^2} \leq B^2, \quad \norm{M(\tilde T^e)} \leq B. \]

  Let $q \coloneqq \abs{D}$. \cref{thm:tensor-norm-bound} implies except with probability $o_n (1)$,
  \[ \Norm{M\Paren{\sum_j \tilde T^j}} \lesssim B^2 \sqrt{m (n q)^{\ceil{t/2}}}. \]
  Now
  \[
    T^\cC - \overline T = \frac1m \sum_{(\sigma,R) \in \cC} (T^\sigma - \overline T)
    = \frac1m \sum_{j \in [m]} \tilde T^j,
  \]
  Therefore
  \[ \norm{M(T^\cC - \overline T)} \lesssim B^2 \sqrt{\frac{(nq)^{\ceil{t/2}}}m}, \]
  which is at most $\eps$ whenever $m \geq C(T,\eps) n^{\ceil{t/2}}$.
\end{proof}

\begin{claim} \label{claim:norm-bounds-replacement}
  In the context of \cref{thm:avg-tensor-norm-bound-replacement}, $\norm{M(T^j)} \leq B/2$ and
  $\norm{M(\tilde T^j)} \leq B$.
  Further, $\E\Brac{\norm{M(\tilde T^j)}^2} \leq B^2$.
\end{claim}

\begin{proof}
  For any $j \in [m]$ and any injection $\sigma: [k] \to V$, \cref{claim:norm-embed} implies
  $\norm{M(T^j)} = \norm{M(T)} = B/2$.
  Also
  \[ \norm{M(\E T^j)} = \norm{\E M(T^j)} \leq \E \norm{M(T^j)} = B/2 \]
  by linearity of $M(\cdot)$ and convexity of the norm.
  Thus triangle inequality of the norm implies
  \[ \norm{M(\tilde T^j)} = \norm{M(T^j - \E T^j)} \leq \norm{M(T^j)} + \norm{M(\E T^j)} \leq B. \]
  Therefore $\E\Brac{\norm{M(\tilde T^j)}^2} \leq B^2$.
\end{proof}

\subsection{Concentration of marginals}
We prove here concentration of marginals of constraints that are crucial to our argument.

Key to our refutation is concentration of functions certified using degree-$d$ SOS proofs:

\begin{definition} \label{def:concentration-instance}
  Given $f^\Omega: \Omega_k \to \R, \eps > 0, d \in \N$, a constraint set $\cC$ over variable set $V$ is
  $(f^\Omega,\eps,d)$-concentrated if there is a degree-$d$ SOS proof of
  \[
    \Set{ \textstyle \sum_{a\in D} x_{v,a}^2 \leq 1}_{v \in V} \quad \vdash_d \quad \E_{\sigma \in \mathcal C}
    f^{\Omega,\sigma} (x) \leq \E_{\sigma \in V^{\underline k}} f^{\Omega,\sigma} (x) + \eps.
  \]
\end{definition}

We do not consider the satisfying assignments of the constraints in the above definition.
When we apply the above definition, we will consider each set $R \in \cR$ of satisfying assignments separately, and
$\cC$ in the above definition will be those constraints with satisfying assignments $R$.

\begin{theorem} \label{thm:concentration-random-instance}
  Fix $t \leq k$ and $\eps > 0$. Suppose $f^\Omega: \Omega_k \to \R$ has pure degree $t$.
  There is $C = C(f,\eps)$ such that for $p \geq C n^{\ceil{t/2}-k}$, except with probability $o_n (1)$, a random Erd\H
  os--R\'enyi random $k$-uniform constraint set $\cC$ on $n$ variables with constraint probability $p$ is
  $(f^\Omega,\eps,2\lceil t/2\rceil)$-concentrated.
\end{theorem}

\begin{proof}
  Since $f^\Omega$ has pure degree $t$, \cref{sec:tensor-func} implies there is a tensor $T$ on
  $([k] \times D)^{\underline t}$ such that $\sum_\alpha T_\alpha z^\alpha = f^\Omega (z)$ for $z \in \Omega_k$.
  \cref{thm:avg-tensor-norm-bound} shows that
  \[ \Norm{M\Paren{\E_{\sigma \in \cC} T^\sigma - \E_{\sigma \in V^{\underline k}} T^\sigma}}_{\infty\to1}
  \leq \frac\eps{K_G}. \]
  \cref{thm:tensor-sos-bound} implies
  \[
    \Set{x_{v,a}^2 \leq 1}_{v \in V, a \in D} \quad \vdash_{2\ceil{t/2}} \quad
    \E_{\sigma \in \cC} \sum_\alpha T^\sigma_\alpha x^\alpha
    \leq \E_{\sigma \in V^{\underline k}} \sum_\alpha T^\sigma_\alpha x^\alpha + \eps.
  \]
  \cref{eq:tensor-func} yields the desired inequality in SOS.
\end{proof}

\cref{thm:concentration-random-instance} concerns the Erd\H os--R\'enyi random model, but analogous result also holds if
$m$ constraints are sampled independently with replacement.
Plugging \cref{thm:avg-tensor-norm-bound-replacement} in place of \cref{thm:avg-tensor-norm-bound} to the proof of
\cref{thm:concentration-random-instance}, we get:

\begin{theorem}\label{thm:concentration-random-instance-replacement}
  Fix $t \leq k$ and $\eps > 0$. Suppose $f^\Omega: \Omega_k \to \R$ has pure degree $t$.
  There is $C = C(f,\eps)$ such that for $m \geq C n^{\ceil{t/2}}$, except with probability $o_n (1)$, a random
  $k$-uniform constraint set $\cC$ on $n$ variables with $m$ constraints with replacement is
  $(f^\Omega,\eps,2\ceil{t/2})$-concentrated.
\end{theorem}

\end{document}